\documentclass[10pt]{article}
\usepackage{amssymb}
\usepackage{amsfonts}
\usepackage{amsmath}
\usepackage{amsthm}
\usepackage{amsopn}
\usepackage{amsthm}
\usepackage{amsfonts}
\usepackage{graphicx}
\usepackage{bm}
\usepackage{tikz}
\usetikzlibrary{cd}
\usepackage{url}
\usepackage{color}
\usepackage{chngcntr}

\newcommand{\func}[1]{\operatorname{#1}}

\newtheorem{theorem}{Theorem}
\newtheorem{lemma}{Lemma}

\begin{document}

\title{Rational solutions of dressing chains and higher order Painlev\'{e}
equations.}
\author{D. Gomez-Ullate$^{a,b,c}$, Y. Grandati$^{d}$, S.\ Lombardo$^{e}$ and
R. Milson$^{f}$ }
\date{$a$: Institute of Mathematical Sciences (ICMAT), C/ Nicolas Cabrera 15,
28049 Madrid, Spain.\\
$b$: Escuela Superior de Ingenier\'ia, Universidad de C\'adiz, 11519, Spain.%
\\
$c$: Departamento de F\'isica Te\'orica, Universidad Complutense de Madrid,
28040 Madrid, Spain.\\
$d$: Laboratoire de Physique et Chimie Th\'{e}oriques, Universit\'{e} de
Lorraine, 1 Bd Arago, 57078 Metz, Cedex 3, France.\\
$e$: Mathematical Sciences Department, Loughborough University, Epinal Way,
Loughborough, Leicestershire, LE11 3TU, UK.\\
$f$: Department of Mathematics and Statistics, Dalhousie University,
Halifax, NS, B3H\ 3J5, Canada.}

\maketitle

\begin{abstract}
We present a new approach to determine the rational solutions of the higher
order Painlev\'{e} equations associated to periodic dressing chain systems (A%
$_{\text{n}}^{(1)}$-Painlev\'{e} systems). We obtain new sets of solutions,
giving determinantal representations indexed by specific Maya diagrams in
the odd case or universal characters in the even case.
\end{abstract}

\section{\protect\bigskip Introduction}

It is now well known that the six Painlev\'{e} equations PI-PVI, discovered
more than one century ago by Painlev\'{e} and Gambier \cite{conte,clarkson2}
(see also Fuchs, Picard and Bonnet \cite{conte2}), define new transcendental
objects which can be thought as nonlinear analogues of special functions.
Except for the first one, the Painlev\'{e} equations all depend on some set
of parameters and if for generic values of these parameters the solutions
cannot be reduced to usual transcendental functions, it appears that for
some specific values of these parameters we retrieve classical
transcendental functions or even rational functions \cite%
{conte,clarkson2,noumi,gromak}. In the last decades, the classification and
the properties of these last solutions has been a subject of active research 
\cite{clarkson2,noumi}.

In the early nineties, Shabat, Veselov and Adler \cite{shabat,vesshab,adler}
introduce the concept of dressing chains, showing that they possess the
Painlev\'{e} property and that the Painlev\'{e} equations PII-PVI can be
described in terms of dressing chains (scalar or matricial) of low orders.
The higher order chains can then be considered as generalizing the classical
Painlev\'{e} equations and the problem of finding their rational solutions
arises naturally \cite{clarkson1,clarkson5}. For a Schr\"{o}dinger operator,
the scalar dressing chain of period $p$ constitutes an higher order
generalization of PIV for $p$ odd (the symmetric form of PIV itself
corresponds to the case $p=3$) and an higher order generalization of PV for $%
p$ even (the symmetric form of PV itself corresponds to the case $p=4$).

In the case of dressing chains of odd periodicity, previous results \cite%
{veselov} seem to indicate that these solutions are necessarily obtained
from rational extensions of the harmonic oscillator (HO) potential \cite%
{GGM1,GGM0,GGM}. A remaining question is then to determine among all these
rational dressings of \ the harmonic oscillator, the particular ones which
allow to solve a dressing chain of given periodicity. An elegant, although
indirect, way to answer to this problem pass through the approach developed
principaly by the japanese school around Okamoto, Noumi, Yamada, Umemura,
Tsuda and others (see \cite{clarkson2,noumi} and references therein, see
also \cite{adler}). It rests on the symmetry group analysis of the dressing
chain in the parameters space. The parametric symmetries of the chain
combined to B\"{a}cklund transformations constitute an "extended affine Weyl
group" (for the scalar dressing chain of period $p$, the associated extended
affine Weyl group is $A_{p-1}^{\left( 1\right) }$) which preserves the
structure of the dressing chain. Starting from some a complete set of simple
"fundamental solutions" possessing the rational character, successive
applications of the transformations belonging to these Weyl groups generates
step by step all the rational solutions of the dressing chain system.

In this article, we propose an alternative approach to build the rational
solutions of the Schr\"{o}dinger operators dressing chain systems (with non
zero shift). It has the advantage to be direct and explicit, in the sense
that it furnishes an immediate determinantal representation for the
solutions of the differential system. Moreover, it links the existence of
rational solutions for the dressing chain system to the analytical
properties of the underlying quantum potential.

We consider the whole sets of rational extensions of the harmonic and
isotonic potentials that we label by Maya diagrams and universal character 
\cite{koike,tsuda} respectively. Analyzing the combinatorial properties of
these objects allows us to select the extended potentials which solve the
odd and even periodic dressing chains respectively and gives closed form
determinantal expressions for the solutions of dressing chain system.

The paper is organized as follows. We start by recalling some basic elements
concerning the concept of dressing chains, cylic potentials and their
connection to the PIV and PV equations. In the second part, we introduce the
notion of cyclic Maya diagram\ and associated $\overrightarrow{s}$-vector
and give the general structure of a p-cyclic Maya diagram (Theorem 1). It
allows in the next part, to show how the cyclic rational extensions of the
HO give rational solutions of dressing chains system of period p for
appropriate choices of the parameters (Theorem 2). As illustrating examples,
we treat in details the cases $p=3$\ (ie the 
standard PIV equation) and $p=5$
(also called the A$_{\text{4}}$-PIV system).

In the fourth 
section, after recalling some properties of the isotonic and
its rational extensions, which are expressed in terms of Laguerre
pseudo-Wronskians and labelled by universal charaters, we show how these
last ones allow to obtain rational solutions of the even periodic dressing
chain systems for which we give new explicit representations (Theorem 4). We
illustrate these general results for the particular cases $p=4$ (ie the
standard
PV equation) and $p=6$ (ie A$_{\text{5}}$-PV system).

\section{Dressing chains and cyclic potentials}

Consider a potential $U(x)$, the associated Schr\"{o}dinger and Riccati-Schr%
\"{o}dinger (RS) equation \cite{grandati berard} being respectively

\begin{equation}
\left\{ 
\begin{array}{c}
-\psi _{\lambda }^{\prime \prime }(x)+U(x)\psi _{\lambda }(x)=E_{\lambda
}\psi _{\lambda }(x) \\ 
-w_{\lambda }^{\prime }(x)+w_{\lambda }^{2}(x)=U(x)-E_{\lambda },%
\end{array}%
\right.  \label{SetRS}
\end{equation}%
where the \textbf{RS function} $w_{\lambda }(x)$ is minus the logarithmic
derivative of the eigenfunction $\psi _{\lambda }(x)$: $w_{\lambda
}(x)=-\psi _{\lambda }^{\prime }(x)/\psi _{\lambda }(x)$. The auxiliary
spectral parameter $\lambda $ allows us in what follows to define a sequence
of eigenvalues and associated eigenfunctions.

Given a particular eigenfunction $\psi _{\nu }(x)$ (or its associated RS
function $w_{\nu }(x)$) of $U(x)$ for the eigenvalue $E_{\nu }$, we can
build a new potential $U^{\left( \nu \right) }(x)$ via the \textbf{Darboux
transformation (DT)} of \textbf{seed function} $\psi _{\nu }$ \cite{darboux}:

\begin{equation}
U^{\left( \nu \right) }(x)=U(x)+2w_{\nu }^{\prime }(x),
\end{equation}%
called an \textbf{extension} of $U(x)$. Then $\psi _{\lambda }^{\left( \nu
\right) }(x)$ and $w_{\lambda }^{\left( \nu \right) }(x)$ defined as\bigskip 
\begin{equation}
\left\{ 
\begin{array}{c}
\psi _{\lambda }^{\left( \nu \right) }(x)=W(\psi _{\nu },\psi _{\lambda
}\mid x)/\psi _{\nu }(x),\text{ }\lambda \neq \nu \\ 
\psi _{\nu }^{\left( \nu \right) }(x)=1/\psi _{\nu }(x),%
\end{array}%
\right.  \label{DTS}
\end{equation}%
and

\begin{equation}
\left\{ 
\begin{array}{c}
w_{\lambda }^{\left( \nu \right) }(x)=-w_{\nu }(x)+(E_{\lambda }-E_{\nu
})/\left( w_{\lambda }(x)-w_{\nu }(x)\right), \text{ }\lambda \neq \nu \\ 
w_{\nu }^{\left( \nu \right) }(x)=-w_{\nu }(x),%
\end{array}%
\right.  \label{DTRS}
\end{equation}%
are solutions of

\begin{equation}
\left\{ 
\begin{array}{c}
-\psi _{\lambda }^{\left( \nu \right) \prime \prime }+U^{\left( \nu \right)
}\psi _{\lambda }^{\left( \nu \right) }=E_{\lambda }\psi _{\lambda }^{\left(
\nu \right) } \\ 
-\left( w_{\lambda }^{\left( \nu \right) }\right) ^{\prime }+\left(
w_{\lambda }^{\left( \nu \right) }\right) ^{2}=U^{\left( \nu \right)
}-E_{\lambda }.%
\end{array}%
\right.
\end{equation}

By chaining such DT, we produce a sequence of extensions

\begin{equation}
\left\{ 
\begin{array}{c}
\psi _{\lambda }\overset{\nu _{1}}{\rightarrowtail }\psi _{\lambda }^{\left(
\nu _{1}\right) }\overset{\nu _{2}}{\rightarrowtail }\psi _{\lambda
}^{\left( \nu _{1},\nu _{2}\right) }...\overset{\nu _{p}}{\rightarrowtail }%
\psi _{\lambda }^{\left( \nu _{1},...,\nu _{p}\right) } \\ 
U\overset{\nu _{1}}{\rightarrowtail }U^{\left( \nu _{1}\right) }\overset{\nu
_{2}}{\rightarrowtail }U^{\left( \nu _{1},\nu _{2}\right) }...\overset{\nu
_{p}}{\rightarrowtail }U^{\left( \nu _{1},...,\nu _{p}\right) },%
\end{array}%
\right.  \label{diagn}
\end{equation}%
where

\begin{equation}
U^{\left( \nu _{1},...,\nu _{p}\right) }(x)=U(x)+2\left(
\sum_{i=1}^{p}w_{\nu _{i}}^{\left( \nu _{1},...,\nu _{i-1}\right)
}(x)\right) ^{\prime }.
\end{equation}

The Crum formulas allow us to express the extended potentials as well as
their eigenfunctions in terms of Wronskians, containing only eigenfunctions
of the initial potential \cite{crum,GGM1}:

\begin{equation}
\left\{ 
\begin{array}{c}
U^{\left( \nu _{1},...,\nu _{p}\right) }(x)=U(x)-2\left( \log W^{\left( \nu
_{1},...,\nu _{p}\right) }(x)\right) ^{\prime \prime } \\ 
\psi _{\lambda }^{\left( \nu _{1},...,\nu _{p}\right) }\left( x\right)
=W^{\left( \nu _{1},...,\nu _{p},\lambda \right) }(x)/W^{\left( \nu
_{1},...,\nu _{p}\right) }(x),%
\end{array}%
\right.  \label{crum}
\end{equation}%
where%
\begin{equation}
W^{\left( \nu _{1},...,\nu _{p}\right) }(x)=W(\psi _{\nu _{1}},...,\psi
_{\nu _{p}}\mid x)  \label{wronsk}
\end{equation}%
is the Wronskian of the family $(\psi _{\nu _{1}},...,\psi _{\nu _{p}})$.
Here we are using the convention that if a spectral index is repeated two
times in the characteristic tuple $\left( \nu _{1},...,\nu _{p}\right) $,
then we suppress the corresponding eigenfunction in the Wronskians of the
right-hand members of Eq(\ref{crum}). In order to simplify the notation we
temporarly use the simplified notation $\nu _{i}\rightarrow i$. We can now
define the notion of \textbf{cyclicity}.

A potential $U(x)$ is said to be $p$\textbf{-cyclic} if there exists a chain
of $p$ DT such that

\begin{equation}
U^{\left( 1,...,p\right) }(x)=U(x)+\Delta ,  \label{Cyclicity}
\end{equation}%
ie at the end of the chain we recover \textit{exactly} the initial potential
translated by an energy shift $\Delta $. This condition is then stronger
than the usual $p^{th}$-order shape invariance \cite{gendenshtein,GGM1}. As
shown by Veselov and Shabat \cite{vesshab}, the successive RS seed functions
then satisfy the following first order non linear system ($\varepsilon
_{ij}=E_{i}-E_{j}$)

\begin{equation}
\left\{ 
\begin{array}{c}
-\left( w_{2}^{\left( 1\right) }(x)+w_{1}(x)\right) ^{\prime }+\left(
w_{2}^{\left( 1\right) }(x)\right) ^{2}-\left( w_{1}(x)\right)
^{2}=\varepsilon _{12} \\ 
-\left( w_{3}^{\left( 1,2\right) }(x)+w_{2}^{\left( 1\right) }(x)\right)
^{\prime }+\left( w_{3}^{\left( 1,2\right) }(x)\right) ^{2}-\left(
w_{2}^{\left( 1\right) }(x)\right) ^{2}=\varepsilon _{23} \\ 
... \\ 
-\left( w_{p}^{\left( 1,...,p-1\right) }(x)+w_{p-1}^{\left( 1,...,p-2\right)
}(x)\right) ^{\prime }+\left( w_{p}^{\left( 1,...,p-1\right) }(x)\right)
^{2}-\left( w_{p-1}^{\left( 1,...,p-2\right) }(x)\right) ^{2}=\varepsilon
_{p-1,p} \\ 
-\left( w_{1}(x)+w_{p}^{\left( 1,...,p-1\right) }(x)\right) ^{\prime
}+\left( w_{1}(x)\right) ^{2}-\left( w_{3}^{\left( 1,...,p-1\right)
}(x)\right) ^{2}=\varepsilon _{p1}-\Delta .%
\end{array}%
\right.  \label{pcyclicdress}
\end{equation}%
called a \textbf{dressing chain of period }$p$. We will say that \textbf{the
potential }$U(x)$\textbf{\ solves the dressing chain}. $\Delta $ and the $%
\varepsilon _{i,i+1}$ are called the \textbf{parameters} of the dressing
chain.

The cyclicity condition Eq(\ref{Cyclicity}) gives also

\begin{equation}
2\left( w_{1}(x)+...+w_{p}^{\left( 1,...,p-1\right) }(x)\right) ^{\prime
}=\Delta ,
\end{equation}%
that is, with an appropriate choice of the integration constant

\begin{equation}
w_{1}(x)+...+w_{p}^{\left( 1,...,p-1\right) }(x)=\frac{\Delta }{2}x.
\label{addconst}
\end{equation}

In all what follows, we suppose that (\textbf{non zero shift} assumption) 
\begin{equation}
\Delta \neq 0.  \label{delta}
\end{equation}

As proven by Veselov and Shabat \cite{vesshab}, this system passes the
Painlev\'{e} test and, as we will see below, for $p=3$ and $p=4$ the
corresponding dressing chains can be seen as symmetrized forms of the PIV
and PV equations respectively.

As illustrating examples we now consider the lowest order cases of cyclicity.

\subsection{1-step cyclic potential}

It is straightforward to show that the HO is in fact the unique potential
possessing the $1$-step cyclicity property. Indeed, if we want

\begin{equation}
U^{\left( \nu \right) }(x)=U(x)+\Delta ,  \label{1 step cyclic}
\end{equation}%
Eq(\ref{addconst}) gives immediately

\begin{equation}
w_{\nu }(x)=\Delta x/2  \label{RSOHfond}
\end{equation}%
and

\begin{equation}
U(x)=E_{\nu }-w_{\nu }^{\prime }(x)+w_{\nu }^{2}(x)=\frac{\Delta ^{2}}{4}%
x^{2}+E_{\nu }-\frac{\Delta }{2},
\end{equation}%
ie $U(x)$ is an HO potential with frequency $\omega =\Delta $, the $1$-step
cyclicity condition coinciding for this potential with the usual shape
invariance property \cite{gendenshtein,GGM1}.

\subsection{2-step cyclic potentials}

If the potential $U(x)$ is $2$-step cyclic, the corresponding dressing chain
of period $2$ is (see Eq(\ref{pcyclicdress}) and Eq(\ref{addconst}))

\begin{equation}
\left\{ 
\begin{array}{c}
-\left( w_{2}^{\left( 1\right) }(x)+w_{1}(x)\right) ^{\prime }+\left(
w_{2}^{\left( 1\right) }(x)\right) ^{2}-\left( w_{1}(x)\right)
^{2}=\varepsilon _{12} \\ 
-\left( w_{1}(x)+w_{2}^{\left( 1\right) }(x)\right) ^{\prime }+\left(
w_{1}(x)\right) ^{2}-\left( w_{2}^{\left( 1\right) }(x)\right)
^{2}=\varepsilon _{21}-\Delta ,%
\end{array}%
\right.  \label{2 step cyclic}
\end{equation}%
with

\begin{equation}
w_{1}(x)+w_{2}^{(1)}(x)=\frac{\Delta }{2}x.  \label{cond12step}
\end{equation}

Taking the difference of the two equations in Eq(\ref{2 step cyclic}) and
combining with Eq(\ref{cond12step}), we obtain in a straightforward way

\begin{equation}
w_{2}^{(1)}(x)=\frac{\Delta }{4}x+\frac{\varepsilon _{12}/\Delta +1/2}{x},\
w_{1}(x)=\frac{\Delta }{4}x-\frac{\varepsilon _{12}/\Delta +1/2}{x}.
\end{equation}

Consequently

\begin{equation}
U(x)=E_{1}-w_{1}^{\prime }(x)+w_{1}^{2}(x)=\frac{\omega ^{2}}{4}x^{2}+\frac{%
\left( \alpha +1/2\right) \left( \alpha -1/2\right) }{x^{2}}-\omega (\alpha
+1)=V(x;\omega ,\alpha ),  \label{2 step cyclic2}
\end{equation}%
where $\omega =\Delta /2$ and $\alpha =\varepsilon _{12}/\Delta $. We deduce
that the unique $2$-step cyclic potential is the isotonic oscillator (IO)
with frequency $\omega $ and "angular momentum" $a=$ $\alpha -1/2$ \cite%
{grandati3,GGM2}.

\subsection{3-step cyclic potentials and Painlev\'{e} IV}

The dressing chain of period $p=3$ has the form (see Eq(\ref{pcyclicdress}))

\begin{equation}
\left\{ 
\begin{array}{c}
-\left( w_{2}^{\left( 1\right) }(x)+w_{1}(x)\right) ^{\prime }+\left(
w_{2}^{\left( 1\right) }(x)\right) ^{2}-\left( w_{1}(x)\right)
^{2}=\varepsilon _{12} \\ 
-\left( w_{3}^{\left( 1,2\right) }(x)+w_{2}^{\left( 1\right) }(x)\right)
^{\prime }+\left( w_{3}^{\left( 1,2\right) }(x)\right) ^{2}-\left(
w_{2}^{\left( 1\right) }(x)\right) ^{2}=\varepsilon _{23} \\ 
-\left( w_{1}(x)+w_{3}^{\left( 1,2\right) }(x)\right) ^{\prime }+\left(
w_{1}(x)\right) ^{2}-\left( w_{3}^{\left( 1,2\right) }(x)\right)
^{2}=\varepsilon _{31}-\Delta ,%
\end{array}%
\right.  \label{3chain2}
\end{equation}%
with (see Eq(\ref{addconst}))

\begin{equation}
w_{1}(x)+w_{2}^{\left( 1\right) }(x)+w_{3}^{\left( 1,2\right) }(x)=\frac{%
\Delta }{2}x.  \label{3SIP2}
\end{equation}

By defining

\begin{equation}
\left\{ 
\begin{array}{c}
y=\sqrt{\frac{2}{\Delta }}\left( w_{1}(x)-\Delta x/2\right) \\ 
t=\sqrt{\frac{2}{\Delta }}x,%
\end{array}%
\right.  \label{chv}
\end{equation}%
we can then easily show \cite{vesshab,adler} that $y$ satisfies the PIV
equation (see \cite{conte,clarkson,clarkson2})

\begin{equation}
y^{\prime \prime }=\frac{1}{2y}\left( y^{\prime }\right) ^{2}+\frac{3}{2}%
y^{3}+4ty^{2}+2\left( t^{2}-a\right) y+\frac{b}{y},  \label{PIV}
\end{equation}%
(here the prime denotes the derivative with respect to $t$) with parameters

\begin{equation}
a=-\left( \Delta +\varepsilon _{23}+2\varepsilon _{12}\right) /\Delta ,\quad
b=-\frac{2\varepsilon _{23}^{2}}{\Delta ^{2}}.  \label{paramPIV}
\end{equation}

\subsection{4-step cyclic potentials and Painlev\'{e} V}

For $p=4$, the dressing chain of Eq(\ref{pcyclicdress}) becomes

\begin{equation}
\left\{ 
\begin{array}{c}
-\left( w_{2}^{\left( 1\right) }(x)+w_{1}(x)\right) ^{\prime }+\left(
w_{2}^{\left( 1\right) }(x)\right) ^{2}-\left( w_{1}(x)\right)
^{2}=\varepsilon _{12} \\ 
-\left( w_{3}^{\left( 1,2\right) }(x)+w_{2}^{\left( 1\right) }(x)\right)
^{\prime }+\left( w_{3}^{\left( 1,2\right) }(x)\right) ^{2}-\left(
w_{2}^{\left( 1\right) }(x)\right) ^{2}=\varepsilon _{23} \\ 
-\left( w_{4}^{\left( 1,2,3\right) }(x)+w_{3}^{\left( 1,2\right) }(x)\right)
^{\prime }+\left( w_{4}^{\left( 1,2,3\right) }(x)\right) ^{2}-\left(
w_{3}^{\left( 1,2\right) }(x)\right) ^{2}=\varepsilon _{34} \\ 
-\left( w_{1}(x)+w_{4}^{\left( 1,2,3\right) }(x)\right) ^{\prime }+\left(
w_{1}(x)\right) ^{2}-\left( w_{4}^{\left( 1,2,3\right) }(x)\right)
^{2}=\varepsilon _{41}-\Delta .%
\end{array}%
\right.
\end{equation}%
with the cyclicity condition (see Eq(\ref{addconst}))

\begin{equation}
w_{1}(x)+w_{2}^{\left( 1\right) }(x)+w_{3}^{\left( 1,2\right)
}(x)+w_{4}^{\left( 1,2,3\right) }(x)=\frac{\Delta }{2}x.
\end{equation}

As shown by Adler \cite{adler}, the function

\begin{equation}
y\left( t\right) =1-\frac{\Delta x}{2\left( w_{1}(x)+w_{2}(x)\right) },\
t=x^{2},  \label{defPV}
\end{equation}%
satisfies the PV equation (see \cite{conte,clarkson4,clarkson2})

\begin{equation}
y^{\prime \prime }=\left( \frac{1}{2y}+\frac{1}{y-1}\right) \left( y^{\prime
}\right) ^{2}-\frac{y^{\prime }}{t}+\frac{\left( y-1\right) ^{2}}{t^{2}}%
\left( ay+\frac{b}{y}\right) +c\frac{y}{t}+d\frac{y\left( y+1\right) }{y-1},
\label{PV}
\end{equation}%
with parameters

\begin{equation}
a=\frac{\varepsilon _{12}^{2}}{2\Delta ^{2}},\ b=-\frac{\varepsilon _{34}^{2}%
}{2\Delta ^{2}},\ c=\frac{1}{4}\left( \Delta -\varepsilon _{41}+\varepsilon
_{23}\right) ,\ d=-\frac{\Delta ^{2}}{32}.  \label{paramPV}
\end{equation}

\section{Cyclic Maya diagrams}

\subsection{First definitions}

We define a \textbf{Sato Maya diagram} as an infinite row of boxes, called 
\textbf{levels}, labelled by relative integers and which can be empty or
filled by at most one "particle" (graphically represented by a bold dot) 
\cite{ohta,hirota,GGM,GGM2}. All the levels sufficiently far away on the
left are filled and all the levels sufficiently far away on the right are
empty. The set of Sato Maya diagrams can be put in one to one correspondence
with the set of tuple of relative integers of the form $N_{m}=\left(
n_{1},...,n_{m}\right) \in 
\mathbb{Z}
^{m}$ The tuple $N_{m}$ contains the indices of the filled levels above zero
(included) and the indices of the empty levels strictly below zero, or in
other words, the filled levels in corresponding Sato Maya diagram are
indexed by the set

\begin{equation}
\{j<0:j\notin N_{m}\}\cup \{j\geq 0:j\in N_{m}\}.
\end{equation}

In the following we use the term \textbf{Maya diagram} to designate both the
Sato Maya diagram in graphical form, and the associated tuple\textbf{. }If
all the $n_{i}$ are positive then the Maya diagram is said to be \textbf{%
positive} (respectively \textbf{negative}). In the case of a positive Maya
diagram, if all the $n_{i}$ are non-zero, then the Maya diagram is said to
be \textbf{strictly positive}.

Two Maya diagrams $N_{m}$ and $N_{m^{\prime }}^{\prime }$ are \textbf{%
equivalent} if they differ only by a global translation of all the particles
in the levels and we note 
\begin{equation}
N_{m}\approx N_{m^{\prime }}^{\prime }.  \label{eqMD}
\end{equation}

The \textbf{canonical representative} of such an equivalence class is the
unique strictly positive Maya diagram of the class for which the zero level
is empty, i.e. the zero level is the first empty level.

A $k$\textbf{-translation} applied to the canonical Maya diagram $%
N_{m}=\left( n_{1},...,n_{m}\right) \in \left( 
\mathbb{N}
^{\ast }\right) ^{m}$ generates the Maya diagram, denoted $N_{m}\oplus k$,
obtained from by shifting by $k$ all the particles in the levels of $N_{m}$.
For $k>0$, we have

\begin{equation}
N_{m}\oplus k=\left( 0,...,k-1\right) \cup \left( N_{m}+k\right) ,
\end{equation}%
where $N_{m}+k=\left( n_{1}+k,...,n_{m}+k\right) $.

Starting from a given Maya diagram $N_{m}$ we can modify it by suppressing
particles in the filled levels or by filling empty levels. We call such an
action on a level, a \textbf{flip}. If $\left( \nu _{1},...,\nu _{p}\right) $
are the indices of the flipped levels then the characteristic tuple of the
resulting Maya diagram can be denoted $\left( N_{m},\nu _{1},...,\nu
_{p}\right) $ with the convention that a twice repeated index has to be
suppressed.

\subsection{$\protect\overrightarrow{s}$-vectors and cyclicity}

A Maya diagram $N_{m}$ is said $p$\textbf{-cyclic with translation of }$k>0$
if we can translate it by $k$, by acting on it with $p$ flips. In other
words, there must exist a tuple of $p$ positive integers $\left( \nu
_{1},...,\nu _{p}\right) \in 
\mathbb{N}
^{p}$ such that

\begin{equation}
N_{m}\approx \left( N_{m},\nu _{1},...,\nu _{p}\right) .  \label{cyclMD1}
\end{equation}

In particular, if $N_{m}$ is canonical, there must exist a tuple of $p$
positive integers $\left( \nu _{1},...,\nu _{p}\right) \in 
\mathbb{N}
^{p}$ such that

\begin{equation}
\left( N_{m},\nu _{1},...,\nu _{p}\right) =N_{m}\oplus k.  \label{cyclMD2}
\end{equation}

The ordered tuple $\left( \nu _{1},...,\nu _{p}\right) $ is called a $p$%
\textbf{-cyclic chain} associated to $N_{m}$.

We can readily see that every extension is trivially $\left( 2m+k\right) $%
-cyclic with translation of $k$ via the chain: 
\begin{equation}
N_{m}\cup \left( N_{m}+k\right) \cup \left( 0,...,k-1\right) =\left(
n_{1},...,n_{m},n_{1}+k,...,n_{m}+k,0,...,k-1\right) .
\end{equation}

It is clear that every $p$-cyclic with translation of $k$ Maya diagram is
also $\left( p+2l\right) $-cyclic with translation of $k$ for any integer $l$%
, since Eq(\ref{cyclMD2}) implies

\begin{equation}
\left( N_{m},\nu _{1},...,\nu _{p},\nu _{p+1},\nu _{p+1},...,\nu _{l},\nu
_{l}\right) =N_{m}\oplus k.
\end{equation}

To a given Maya diagram $N_{m}$, we can associate in a one to one way a%
\textbf{\ }$\overrightarrow{s}$\textbf{-vector}, which is an infinite
sequence of ``spin variables'', $\overrightarrow{s}=\left(
s_{n}\right) _{n\in 
\mathbb{Z}
},$ where $s_{n}=+1$ (\textbf{up spin}) or $s_{n}=-1$ (\textbf{down spin}),
in the following way

* If level $n$ is filled then $s_{n}=-1$.

* If level $n$ is empty then $s_{n}=+1$.

For a canonical Maya diagram, it means that $N_{m}$ gives the positions of
the down spins in the positively indexed part of the $\overrightarrow{s}$%
-vector:

* If $n\in N_{m}$ or $n<0$, then $s_{n}=-1$.

* If $n\geq 0$ and $n\notin N_{m}$, then $s_{n}=+1$.

The $\overrightarrow{s}$\textbf{-vector} is subject to the \textbf{%
topological constraint}

\begin{equation}
\underset{n\rightarrow -\infty }{\lim }s_{n}=-1,\ \underset{n\rightarrow
+\infty }{\lim }s_{n}=1.  \label{topoconstraint}
\end{equation}

In the particular case of a canonical Maya diagram, we have more precisely

\begin{equation}
\left\{ 
\begin{array}{c}
s_{n<0}=-1 \\ 
s_{0}=+1 \\ 
s_{n>n_{m}}=+1.%
\end{array}%
\right.  \label{topoconst}
\end{equation}

A flip at the level $\nu $ in $N_{m}$ corresponds to change the sign of $%
s_{\nu }$($s_{\nu }\rightarrow -s_{\nu }$): suppressing a particle in a
level corresponds to a \textbf{positive flip} of $s_{\nu }$ $\left(
-1\rightarrow +1\right) $, while filling a level corresponds to a \textbf{%
negative flip} of $s_{\nu }$ $\left( +1\rightarrow -1\right) $.

In the sequel, we let $\mathcal{F}^{\left( i_{1},...,i_{p}\right) }$ denote
the \textbf{flip operator} which, when acting on $\overrightarrow{s}$, flips
the spins $s_{i_{1}},...,s_{i_{p}}$. We also let $\mathcal{T}_{k}$ denote
the \textbf{translation operator} of amplitude $k$ on the $\overrightarrow{s}
$-vector :

\begin{equation}
\mathcal{T}_{k}\overrightarrow{s}=\overrightarrow{s}^{\prime }\text{, with }%
s_{n}^{\prime }=s_{n-k},\ \forall n\in 
\mathbb{Z}
.  \label{trans}
\end{equation}

A Maya diagram $N_{m}$ is $p$-cyclic with translation of $k$ iff $\exists
\left( i_{1},...,i_{p}\right) \in 
\mathbb{N}
^{p}$ such that the corresponding $\overrightarrow{s}$-vector is in the
kernel of $\left( \mathcal{F}^{\left( i_{1},...,i_{p}\right) }-\mathcal{T}%
_{k}\right) $:

\begin{equation}
\mathcal{F}^{\left( i_{1},...,i_{p}\right) }\overrightarrow{s}=\mathcal{T}%
_{k}\overrightarrow{s}.  \label{cyclicspin}
\end{equation}

Suppose that $N_{m}$ is a canonical Maya diagram which is $p$-cyclic with
translation of $k>0$. We let $p_{-}$ denote the number of negative flips and 
$p_{+}$ the number of positive flips in the $p$-cyclic chain. We then have
the following lemma:

\begin{lemma}
$k$\textit{\ has the same parity as }$p$\textit{\ and}%
\begin{equation}
\left\{ 
\begin{array}{c}
p=p_{-}+p_{+}\in 
\mathbb{N}
^{\ast }, \\ 
k=p_{-}-p_{+}\in \left\{ 1,...,p\right\} .%
\end{array}%
\right. 
\end{equation}
\end{lemma}

\begin{proof}
Initially, above $n_{m}$ all the $s_{n}$ are up and below $0$ they are all
down. Between $0$ and $n_{m}$, we have $m$ spins down. After the $p$-cyclic
chain, all the spins above $n_{m}+k$ are up and below $k$ they are all down.
In the block of indices $\left\{ k,..,n_{m}+k\right\} $ we still have $m$
down spins and in the set of levels $\left\{ 0,..,n_{m}+k\right\} $ we now
have $m+k$ down spins while before the $p$-cyclic chain the same set
contained only $m$ down spins. Out of the set $\left\{ 0,..,n_{m}+k\right\} $
the Maya $\overrightarrow{s}$-vector remains unchanged and at the end, the
action of the $p$-cyclic chain with translation $k$ is to have flipped
negatively $k$ spins. Consequently%
\begin{equation}
\left\{ 
\begin{array}{c}
p=p_{-}+p_{+}\in 
\mathbb{N}
^{\ast }, \\ 
k=p_{-}-p_{+}\in \left\{ 1,...,p\right\} .%
\end{array}%
\right. 
\end{equation}
\end{proof}

For example, for a $3$-cyclic chain we can have $k=1$ ($p_{-}=2,p_{+}=1$) or 
$k=1$ ($p_{-}=3,p_{+}=0$) and for a $4$-cyclic chain, we can have $k=2$ ($%
p_{-}=3,p_{+}=1$) or $k=4$ ($p_{-}=4,p_{+}=0$).

\subsection{Structure of the cyclic Maya diagrams}

Let first give some supplementary definitions.

We let

\begin{equation}
\left( r\mid s\right) _{k}=\left( r,r+k,...,r+(s-1\right) k),\ r,s\in 
\mathbb{N}
,  \label{block}
\end{equation}%
and call such a set of indices a \textbf{block} of length $s$.

A block of the type $\left( r\mid s\right) _{1}$, containing $s$ consecutive
integers, is called a \textbf{Generalized Hermite (GH) block }\cite%
{clarkson,clarkson2} of length $s$ and in the particular case $r=0$, $\left(
0\mid s\right) _{1}$ is called a \textbf{removable block} of length $s$.

Be $I=\left\{ n_{j}\right\} _{j\in 
\mathbb{Z}
}\subset 
\mathbb{Z}
$ a subset of integers indices. We say that $s_{n}$ is \textbf{discontinuous}
at $n_{j}$ on $I$ if $s_{n_{j+1}}=-s_{n_{j}}$. A subset of indices $I_{l}=k%
\mathbb{Z}
+l,\ l\in \left\{ 0,...,k-1\right\} $ is called a $k$\textbf{-support}. Note
that there is only one $1$-support which identifies with $%
\mathbb{Z}
$.

We can now establish the main theorem concerning the cyclic Maya diagram.

\begin{theorem}
\textit{Any }$p$\textit{-cyclic Maya diagram has the following form}%
\begin{equation}
N_{m}=\left( \left( 1\mid \alpha _{1}\right) _{k},...,\left( k-1\mid \alpha
_{k-1}\right) _{k};\left( \lambda _{1}\mid \mu _{1}\right) _{k},...,\left(
\lambda _{j}\mid \mu _{j}\right) _{k}\right) ,  \label{pcyclicMD}
\end{equation}%
\textit{with }$j=\left( p-k\right) /2\in 
\mathbb{N}
$\textit{\ and where the }$\alpha _{i},\lambda _{i},\mu _{i}$\textit{\ are
arbitrary positive integers, constituting a set of }$p-1$\textit{\ arbitrary
integer parameters on which depends the }$p$\textit{-cyclic Maya diagram }$%
N_{m}$\textit{. The }$k-1$\textit{\ blocks of the type }$\left( l\mid \alpha
_{l}\right) _{k},\ l\in \left\{ 1,...,k-1\right\} ,$\textit{\ }$\alpha _{l}$%
\textit{\ arbitrary, are called }$k$\textit{-Okamoto blocks. The }$j$\textit{%
\ blocks }$\left( \lambda _{l}\mid \mu _{l}\right) _{k},\ l\in \left\{
1,...,j\right\} ,$\textit{\ are called blocks of the second type.}\newline
\textit{The corresponding }$p$\textit{-cyclic chains are obtained by forming
tuples from the set of indices}%
\begin{equation}
\left\{ 0,1+\alpha _{1}k,...,\left( k-1\right) +\alpha _{k-1}k,\ \ \lambda
_{1},\lambda _{1}+\mu _{1}k,...,\ \lambda _{j},\lambda _{j}+\mu
_{j}k\right\} .  \label{pcyclicchain}
\end{equation}

\textit{Concretely, the above set contain }$0$\textit{, the first element of
each block of the second type and the last element of each block (Okamoto
and second type) after having increased its length by }$1$\textit{.}
\end{theorem}

\begin{proof}
Eq(\ref{cyclicspin}) gives the following infinite linear system for the $%
s_{n}$%
\begin{equation}
\left\{ 
\begin{array}{c}
s_{n}=s_{n-k}\text{, if }n\notin N_{m} \\ 
s_{n}=-s_{n-k}\text{, if }n\in N_{m}.%
\end{array}%
\right.   \label{spinsystem}
\end{equation}%
This implies that $s_{n}$, as a function of $n$, is piecewise constant on
each $k$-support $I_{l}$ with a finite number of discontinuities on each of
these supports and with the asymptotic behaviour (see Eq(\ref{topoconst}))%
\begin{equation}
\left\{ 
\begin{array}{c}
s_{l+jk}=+1\text{, if }l+jk>n_{m} \\ 
s_{l+jk}=-1\text{, \ if }l+jk<0.%
\end{array}%
\right.   \label{topoconst2}
\end{equation}%
The discontinuities of $s_{n}$ are located at the flip positions minus $k$, $%
\left( i_{1}-k,...,i_{p}-k\right) $, and the possible $p$-cyclic canonical
Maya diagrams are obtained by sharing these $p$\ discontinuities into the
different $k$-supports $I_{l},\ l\in \left\{ 0,...,k-1\right\} $. Remark
that the constraint Eq(\ref{topoconst2}) implies that we have an odd number
of discontinuities in each $k$-support $I_{l}$. Due to the canonical choice
Eq(\ref{topoconst}) and the cyclicity condition Eq(\ref{cyclMD2}), we know
that necessarily we have to do one flip at the level $0$ (which is on the
support $I_{0}$) and one flip at the level $n_{m}+k$ (which is on the $k$%
-support $I_{l_{m}}$ if $n_{m}=l_{m}\func{mod}k$). This means also that we
have necessarily at least one discontinuity in the $k$-support $I_{0}$ (in $%
-k$) and one in the $k$-support $I_{l_{m}}$ (in $n_{m}$).\newline
Let $p_{l}\in 2%
\mathbb{N}
+1$ the number of discontinuities in the $k$-support $I_{l}$, $%
\sum\limits_{l=0}^{k-1}p_{l}=p$ (with $p-k\in 2%
\mathbb{N}
$), and denote $l+j_{i}^{\left( l\right) }k,\ i\in \left\{ 1,...,p_{l}\right\}
,\ j_{i}^{\left( l\right) }\in 
\mathbb{N}
,$ the positions of these discontinuities.\newline
In the positively indexed part of the $k$-support $I_{l}$, the down spins
are then located at%
\begin{equation}
\left( l,...,l+\left( j_{1}^{\left( l\right) }-1\right) k\right) ,\left(
l+j_{2}^{\left( l\right) }k,...,l+\left( j_{3}^{\left( l\right) }-1\right)
k\right) ,...,\left( l+j_{p_{l}-1}^{\left( l\right) }k,...,l+\left(
j_{p_{l}}^{\left( l\right) }-1\right) k\right) ,
\end{equation}%
which can be rewritten in terms of blocks as (see Eq(\ref{block})),%
\begin{equation}
\left( l\mid j_{1}^{\left( l\right) }\right) _{k},\left( l+j_{2}^{\left(
l\right) }k\mid j_{3}^{\left( l\right) }-j_{2}^{\left( l\right) }\right)
_{k},...,\left( l+j_{p_{l}-1}^{\left( l\right) }k\mid j_{p_{l}}^{\left(
l\right) }-j_{p_{l}-1}^{\left( l\right) }\right) _{k}.  \label{bloc1}
\end{equation}%
Note that due to the canonical choice, it gives in particular for the $k$%
-support $I_{0}$%
\begin{equation}
\left( j_{2}^{\left( 0\right) }k\mid j_{3}^{\left( 0\right) }-j_{2}^{\left(
0\right) }\right) _{k},...,\left( j_{p_{0}-1}^{\left( 0\right) }k\mid
j_{p_{0}}^{\left( 0\right) }-j_{p_{0}-1}^{\left( 0\right) }\right) _{k}.
\label{bloc0}
\end{equation}%
If we denote $j_{1}^{\left( l\right) }=\alpha _{l}$ and%
\begin{equation}
\left\{ 
\begin{array}{c}
j_{2r+1}^{\left( l\right) }-j_{2r}^{\left( l\right) }=\beta _{2r}^{\left(
l\right) } \\ 
l+j_{2r}^{\left( l\right) }k=\beta _{2r-1}^{\left( l\right) },%
\end{array}%
\right. r\geq 1,
\end{equation}%
Eq(\ref{bloc0}) and Eq(\ref{bloc1}) become%
\begin{equation}
\left( \beta _{1}^{\left( 0\right) }\mid \beta _{2}^{\left( 0\right)
}\right) _{k},...,\left( \beta _{p_{0}-2}^{\left( 0\right) }\mid \beta
_{p_{0}-1}^{\left( 0\right) }\right) _{k}
\end{equation}%
and%
\begin{equation}
\left( l\mid \alpha _{l}\right) _{k},\left( \beta _{1}^{\left( l\right)
}\mid \beta _{2}^{\left( l\right) }\right) _{k},...,\left( \beta
_{p_{l}-2}^{\left( l\right) }\mid \beta _{p_{l}-2}^{\left( l\right) }\right)
_{k}.
\end{equation}%
By taking into account all the $k$-supports, globally we have $\left(
k-1\right) $ $k$-Okamoto blocks $\left( l\mid \alpha _{l}\right) _{k},$
where the $\alpha _{l}$ are arbitrary, and $j=\sum\limits_{l=0}^{k-1}\left(
p_{l}-1\right) /2=\left( p-k\right) /2$ blocks of the second type $\left(
\beta _{2r-1}^{\left( l\right) }\mid \beta _{2r}^{\left( l\right) }\right)
_{k},\ r\geq 1,$ where the $\beta _{i}^{\left( l\right) }$ are arbitrary. If
we denote them as 
\begin{equation}
\left( \lambda _{i}\mid \mu _{i}\right) _{k},\ i\in \left\{ 1,...,j\right\} ,
\label{second type blocks}
\end{equation}%
where the $\lambda _{i}$ and $\mu _{i}$ are arbitrary, we arrive to Eq(\ref%
{pcyclicMD}).

The positions of the flips, obtained by adding $k$ to the discontinuities
positions, gives the $p$-cyclic chain associated to $N_{m}$. We have $p_{+}=j
$ positive flips located in $\lambda _{1},...,\lambda _{j}$ and $p_{-}=j+k$
negative flips located in $0,1+\alpha _{1}k,...,\left( k-1\right) +\alpha
_{k-1}k,\ \lambda _{1}+\mu _{1}k,...,\lambda _{j}+\mu _{j}k.$
\end{proof}

For example, consider the $p$-cyclic chain $\left( \lambda _{1}+\mu
_{1}k,...,\lambda _{j}+\mu _{j}k,\lambda _{1},...,\lambda _{j},1+\alpha
_{1}k,...,\left( k-1\right) +\alpha _{k-1}k,0\right) $. Its action on $N_{m}$
can be explicitely written as%
\begin{eqnarray}
&&N_{m}\overset{\lambda _{1}}{\rightarrow }\left( \left( 1\mid \alpha
_{1}\right) _{k},...,\left( k-1\mid \alpha _{k-1}\right) _{k};\left( \lambda
_{1}+k\mid \mu _{1}\right) _{k},...,\left( \lambda _{j}\mid \mu _{j}\right)
_{k}\right)  \notag \\
&&...\overset{\lambda _{j}}{\rightarrow }\left( \left( 1\mid \alpha
_{1}\right) _{k},...,\left( k-1\mid \alpha _{k-1}\right) _{k};\left( \lambda
_{1}+k\mid \mu _{1}-1\right) _{k},...,\left( \lambda _{j}+k\mid \mu
_{j}-1\right) _{k}\right)  \notag \\
&&\overset{\lambda _{1}+\mu _{1}k}{\rightarrow }\left( \left( 1\mid \alpha
_{1}\right) _{k},...,\left( k-1\mid \alpha _{k-1}\right) _{k};\left( \lambda
_{1}+k\mid \mu _{1}\right) _{k},...,\left( \lambda _{j}+k\mid \mu
_{j}-1\right) _{k}\right)  \notag \\
&&...\overset{\lambda _{j}+\mu _{j}k}{\rightarrow }\left( \left( 1\mid
\alpha _{1}\right) _{k},...,\left( k-1\mid \alpha _{k-1}\right) _{k};\left(
\lambda _{1}+k\mid \mu _{1}\right) _{k},...,\left( \lambda _{j}+k\mid \mu
_{j}\right) _{k}\right)  \notag \\
&&\overset{1+\alpha _{1}k}{\rightarrow }\left( \left( 1\mid \alpha
_{1}+1\right) _{k},...,\left( k-1\mid \alpha _{k-1}\right) _{k};\left(
\lambda _{1}+k\mid \mu _{1}\right) _{k},...,\left( \lambda _{j}+k\mid \mu
_{j}\right) _{k}\right)  \notag \\
&&...\overset{\left( k-1\right) +\alpha _{k-1}k}{\rightarrow }\left( \left(
1\mid \alpha _{1}+1\right) _{k},...,\left( k-1\mid \alpha _{k-1}+1\right)
_{k};\left( \lambda _{1}+k\mid \mu _{1}-1\right) _{k},...,\left( \lambda
_{j}+k\mid \mu _{j}\right) _{k}\right)  \notag \\
&&\overset{0}{\rightarrow }N_{m}\oplus k,
\end{eqnarray}%
since

\begin{eqnarray}
&&\left( 0,\left( 1\mid \alpha _{1}\right) _{k},...,\left( k-1\mid \alpha
_{k-1}\right) _{k};\left( \lambda _{1}+k\mid \mu _{1}\right) _{k},...,\left(
\lambda _{j}+k\mid \mu _{j}\right) _{k}\right)  \notag \\
&=&\left( 0,1,...,k-1\right) \cup \left( \left( 1+k\mid \alpha _{1}\right)
_{k},...,\left( 2k-1\mid \alpha _{k-1}\right) _{k};\left( \lambda _{1}+k\mid
\mu _{1}\right) _{k},...,\left( \lambda _{j}+k\mid \mu _{j}\right)
_{k}\right)  \notag \\
&=&\left( 0\mid k\right) _{1}\cup \left( N_{m}+k\right) .
\end{eqnarray}

Let us see now the two extremal cases $k=p$ and $k=1$.

When in Eq(\ref{pcyclicMD}) some blocks merge or overlap, we say that the
block structure of $N_{m}$ is \textbf{degenerate}.

\subsubsection{k = p}

We have no second type block ($j=0$) and $\left( p-1\right) $ $p$-Okamoto
blocks:

\begin{equation}
N_{m}=\left( \left( 1\mid \alpha _{1}\right) _{p},...,\left( p-1\mid \alpha
_{p-1}\right) _{p}\right) ,\ m=\alpha _{1}+...+\alpha _{p-1}.
\end{equation}

Each $p$-support contains exactly one discontinuity ($p_{l}=1,\ \forall l\in
\left\{ 0,...,p-1\right\} $) and $N_{m}$ contains no multiple of $p$. We
call $N_{m}$ a $p$\textbf{-Okamoto Maya diagram}.

The corresponding $p$-cyclic chains are built from

\begin{equation}
\left\{ 0,1+\alpha _{1}p,...,\left( p-1\right) +\alpha _{p-1}p\right\} ,
\end{equation}%
which contains $p$ negative flips.

\subsubsection{k = 1}

We have no Okamoto block and $\left( p-1\right) /2$ second type blocks:

\begin{equation}
N_{m}=\left( \left( \lambda _{1}\mid \mu _{1}\right) _{1},...,\left( \lambda
_{(p-1)/2}\mid \mu _{(p-1)/2}\right) _{1}\right) ,\ m=\mu _{1}+...+\mu
_{(p-1)/2},
\end{equation}%
which are GH-blocks. There is only one $1$-support which $%
\mathbb{Z}
$ is itself and which contains all the $p$ discontinuities ($k-1=0$ and $%
p_{0}=p$). $N_{m}$ is called a $p$\textbf{-GH Maya diagram}.

In the corresponding $p$-cyclic chains, we have $\left( p+1\right) /2$
negative flips at $0,\lambda _{1}+\mu _{1},...,\lambda _{j}+\mu _{j}$ and $%
\left( p-1\right) /2$ negative flips at $\lambda _{1},...,\lambda _{j}.$

\section{Rational extensions of the HO and rational solutions of the
dressing chains of odd periodicity}

\subsection{Rational extensions of the HO}

The HO potential is defined on the real line by

\begin{equation}
V\left( x;\omega \right) =\frac{\omega ^{2}}{4}x^{2}-\frac{\omega }{2},\
\omega \in 
\mathbb{R}
.  \label{OH}
\end{equation}

With Dirichlet boundary conditions at infinity and supposing $\omega \in 
\mathbb{R}
^{+}$, $V\left( x;\omega \right) $ has the following spectrum ($z=\sqrt{%
\omega /2}x$)%
\begin{equation}
\left\{ 
\begin{array}{c}
E_{n}\left( \omega \right) =n\omega \\ 
\psi _{n}(x;\omega )=\psi _{0}(x;\omega )H_{n}\left( z\right)%
\end{array}%
\right. ,\quad n\geq 0,  \label{spec OH}
\end{equation}%
with $\psi _{0}(x;\omega )=\exp \left( -z^{2}/2\right) .$ It is the most
simple example of translationally shape invariant potential \cite%
{gendenshtein,grandati berard} with

\begin{equation}
V^{\left( 0\right) }\left( x;\omega \right) =V\left( x;\omega \right)
+\omega .  \label{SI}
\end{equation}

It also possesses a unique parametric symmetry, $\Gamma _{3}$, which acts as 
\cite{GGM,GGM1}

\begin{equation}
\left\{ 
\begin{array}{c}
\omega \overset{\Gamma _{3}}{\rightarrow }\left( -\omega \right) \\ 
V(x;\omega )\overset{\Gamma _{3}}{\rightarrow }V(x;-\omega )=V(x;\omega
)+\omega ,%
\end{array}%
\right.  \label{gam3}
\end{equation}%
and then generates the \textbf{conjugate spectrum} of $V\left( x;\omega
\right) $

\begin{equation}
\left\{ 
\begin{array}{c}
E_{n}\left( \omega \right) \\ 
\psi _{n}(x;\omega )%
\end{array}%
\right. \overset{\Gamma _{3}}{\rightarrow }\left\{ 
\begin{array}{c}
E_{-\left( n+1\right) }\left( \omega \right) =-n\omega <0 \\ 
\psi _{n}(x;-\omega )=iH_{n}\left( iz\right) \exp \left( z^{2}/2\right)
=\psi _{-\left( n+1\right) }(x;\omega ).%
\end{array}%
\right.  \label{conjhermite}
\end{equation}

The union of the spectrum and the conjugate spectrum forms the \textbf{%
extended spectrum} of the HO which contains all the \textbf{quasi-polynomial}
(ie polynomial up to a gauge factor) eigenfunctions of this potential. All
the rational extensions of the HO are then obtained via chains of DT with
seed functions chosen in the extended spectrum and they are then labelled by
tuples of spectral indices which relative integers $N_{m}=\left(
n_{1},...,n_{m}\right) \in 
\mathbb{Z}
^{m}$. This establishes a one to one correspondence between the set of
rational extensions of the HO and the set of Maya diagrams. A rational
extension associated to a canonical Maya diagram is said to be a \textbf{%
canonical extension}. The energy spectrum of $V^{^{\left( N_{m}\right)
}}(x;\omega )\ $ is the set of $\left\{ E_{j}(\omega )\right\} _{j\in 
\mathbb{Z}
}$ for the indices associated to the empty levels of the Maya diagram $N_{m}$%
.\ 

As shown in \cite{GGM}, this correspondence preserves the equivalence
relation, in the sense that if 
\begin{equation}
N_{m}\approx N_{m^{\prime }}^{\prime },
\end{equation}%
then the corresponding extended potentials\ are identical up to an additive
constant:

\begin{equation}
V^{\left( N_{m}\right) }(x;\omega )\ =V^{^{\left( N_{m^{\prime }}^{\prime
}\right) }}(x;\omega )+q\omega ,\ q\in 
\mathbb{Z}
.  \label{Eqext1}
\end{equation}

In particular, if $N_{m}$ is canonical

\begin{equation}
V^{\left( N_{m}\oplus k\right) }(x;\omega )\ =V^{\left( N_{m}\right)
}(x;\omega )+k\omega .  \label{Eqext2}
\end{equation}

It follows that to describe all the rational extensions of the HO, we can
restrict ourself to those associated to canonical Maya diagrams.

We have also equivalence relations for the Wronskians, in particular

\begin{equation}
W^{\left( N_{m}\oplus k\right) }(x;\omega )\ =\left( \psi _{0}(x;\omega
)\right) ^{k}W^{\left( N_{m}\right) }(x;\omega ),  \label{Eqext3}
\end{equation}%
or, if $k\leq n_{1}\leq ...\leq n_{m}$

\begin{equation}
W^{\left( 0,...,k-1\right) \cup N_{m}}(x;\omega )\ =\left( \psi
_{0}(x;\omega )\right) ^{k}W^{\left( N_{m}-k\right) }(x;\omega ).
\label{Eqext4}
\end{equation}

Using Eq(\ref{spec OH}) and the usual properties of the Wronskians \cite%
{muir} 
\begin{equation}
\left\{ 
\begin{array}{c}
W\left( uy_{1},...,uy_{m}\mid x\right) =u^{m}W\left( y_{1},...,y_{m}\mid
x\right) \\ 
W\left( y_{1},...,y_{m}\mid x\right) =\left( \frac{dz}{dx}\right)
^{m(m-1)/2}W\left( y_{1},...,y_{m}\mid z\right) ,%
\end{array}%
\right.  \label{wronskprop}
\end{equation}%
we can write

\begin{equation}
W^{\left( N_{m}\right) }(x;\omega )\ \propto \left( \psi _{0}(x;\omega
)\right) ^{m}\mathcal{H}^{\left( N_{m}\right) }\left( z\right) ,  \label{WH}
\end{equation}%
where $\mathcal{H}^{\left( N_{m}\right) }$ is the following Wronskian
determinant ($i=0,...,m-1$)%
\begin{equation}
\mathcal{H}^{\left( N_{m}\right) }\left( z\right) =W\left(
H_{n_{1}},...,H_{n_{m}}\mid z\right) =\left\vert 
\begin{array}{ccc}
H_{n_{1}}(z) & ... & H_{n_{m}}(z) \\ 
... &  & ... \\ 
\left( n_{1}\right) _{i}H_{n_{1}-i}(z) & ... & \left( n_{m}\right)
_{i}H_{n_{m}-i}(z) \\ 
... &  & ... \\ 
\left( n_{1}\right) _{m-1}H_{n_{1}-m+1}(z) & ... & \left( n_{m}\right)
_{m-1}H_{n_{m}-m+1}(z)%
\end{array}%
\right\vert ,  \label{PW}
\end{equation}%
$_{i}\left( x\right) $ and $\left( x\right) _{i}$ being respectively the
rising and falling factorials

\begin{equation}
_{i}\left( x\right) =x(x+1)...(x+i-1),\ \left( x\right)
_{i}=x(x-1)...(x-i+1),  \label{poch}
\end{equation}%
with the convention that $H_{n}(z)=0$ if $n<0$.

Eq(\ref{Eqext3}) and Eq(\ref{Eqext4}) \ give then

\begin{equation}
\left\{ 
\begin{array}{c}
\mathcal{H}^{\left( N_{m}\oplus k\right) }(z)\ \propto \mathcal{H}^{\left(
N_{m}\right) }\left( z\right) \\ 
\mathcal{H}^{\left( 0,...,k-1\right) \cup N_{m}}(z)\ \propto \mathcal{H}%
^{\left( N_{m}-k\right) }\left( z\right) .%
\end{array}%
\right.  \label{Eqext5}
\end{equation}

\subsection{$p$-cyclic extensions of the HO and rational solutions of the
periodic dressing chains}

First we can 
prove
the following lemma:

\begin{lemma}
\textit{A rational extension }$V^{^{\left( N_{m}\right) }}(x;\omega )$%
\textit{\ of the HO is a }$p$\textit{-cyclic potential iff its associated
Maya diagram }$N_{m}$\textit{\ is }$p$\textit{-cyclic.}
\end{lemma}

\begin{proof}
If $N_{m}$ is a canonical Maya diagram which is\ $p$-cyclic with translation
of $k>0$, there exist a tuple of $p$ positive integers $\left( \nu
_{1},...,\nu _{p}\right) \in 
\mathbb{N}
^{p}$ such that (see Eq(\ref{Eqext2}))%
\begin{equation}
V^{\left( N_{m},\nu _{1},...,\nu _{p}\right) }(x;\omega )=V^{\left(
0,...,k-1\right) \cup \left( N_{m}+k\right) }(x;\omega )=V^{\left(
N_{m}\right) }(x;\omega )+k\omega 
\end{equation}%
and consequently $V^{\left( N_{m}\right) }(x;\omega )$ is $p$-cyclic. More
generally, using Eq(\ref{Eqext1}), we deduce that if $N_{m}$ is an arbitrary 
$p$-cyclic Maya diagram, then $V^{^{\left( N_{m}\right) }}(x;\omega )$ is a $%
p$-cyclic potential. \newline
Conversely, if $V^{^{\left( N_{m}\right) }}(x;\omega )$ is $p$-cyclic, there
exist a tuple of $p$ integers $\left( \nu _{1},...,\nu _{p}\right) \in 
\mathbb{Z}
^{p}$ such that%
\begin{equation}
V^{\left( N_{m},\nu _{1},...,\nu _{p}\right) }(x;\omega )=V^{\left(
N_{m}\right) }(x;\omega )+\Delta ,
\end{equation}%
and necessarily the energy shift $\Delta $ must be an integer multiple of $%
\omega $. Due to the correspondence between rational extensions and Maya
diagrams, this implies immediately that $\left( N_{m},\nu _{1},...,\nu
_{p}\right) $ and $N_{m}$ are equivalent. Consequently, $N_{m}$ is $p$%
-cyclic.
\end{proof}

We then deduce the theorem

\begin{theorem}
\textit{The rational extensions of the HO }$V^{^{\left( N_{m}\right)
}}(x;\omega )$\textit{\ where }$N_{m}$\textit{\ is given by Eq(\ref%
{pcyclicMD}), solve the dressing chain of period }$p$\textit{\ Eq(\ref%
{pcyclicdress}) for the set of parameters}%
\begin{equation}
\Delta =k\omega \text{ and }\varepsilon _{i,i+1}=\left( \nu _{P(i)}-\nu
_{P(i+1)}\right) \omega ,\ i\in \left\{ 1,...,p\right\} ,
\end{equation}%
\textit{where }$P$\textit{\ is any permutation of }$S_{p}$\textit{\ and (}$%
\nu _{p+1}=\nu _{1}$\textit{)}%
\begin{equation}
\left( \nu _{1},...,\nu _{p}\right) \mathit{=}\left( 0,1+\alpha
_{1}k,...,\left( k-1\right) +\alpha _{k-1}k,\ \ \lambda _{1},\lambda
_{1}+\mu _{1}k,...,\ \lambda _{j},\lambda _{j}+\mu _{j}k\right) \mathit{.}
\end{equation}%
\textit{The solutions of the dressing chain system are then given by (}$z=%
\sqrt{\omega /2}x$\textit{)}%
\begin{equation}
\left\{ 
\begin{array}{c}
w_{i}^{\left( 1,...,i-1\right) }(x)=-\omega x/2+\sqrt{\frac{\omega }{2}}%
\frac{d}{dz}\left( \log \left( \frac{\mathcal{H}^{\left( N_{m},\nu _{P\left(
1\right) },...,\nu _{P\left( i-1\right) }\right) }\left( z\right) }{\mathcal{%
H}^{\left( N_{m},\nu _{P\left( 1\right) },...,\nu _{P\left( i\right)
}\right) }\left( z\right) }\right) \right) ,\text{ if the flip in }\nu
_{P\left( i\right) }\text{ is positive,} \\ 
w_{i}^{\left( 1,...,i-1\right) }(x)=\omega x/2+\sqrt{\frac{\omega }{2}}\frac{%
d}{dz}\left( \log \left( \frac{\mathcal{H}^{\left( N_{m},\nu _{P\left(
1\right) },...,\nu _{P\left( i-1\right) }\right) }\left( z\right) }{\mathcal{%
H}^{\left( N_{m},\nu _{P\left( 1\right) },...,\nu _{P\left( i\right)
}\right) }\left( z\right) }\right) \right) ,\text{ if the flip in }\nu
_{P\left( i\right) }\text{ is negative,}%
\end{array}%
\right.   \label{soldress}
\end{equation}%
\textit{with the convention that if a spectral index is repeated two times
in the tuple }$\left( \nu _{1},...,\nu _{p}\right) $\textit{\ characterizing
the chain, then we suppress the corresponding eigenfunction in the }$%
\mathcal{H}$\textit{\ determinants.}
\end{theorem}

\begin{proof}
The first part is a direct consequence of the preceding lemma combined to
Theorem 1. With the Crum formulas Eq(\ref{crum}), we can write%
\begin{equation}
w_{i}^{\left( 1,...,i-1\right) }(x)=-\left( \log \left( \frac{W^{\left(
N_{m},\nu _{P\left( 1\right) },...,\nu _{P\left( i\right) }\right)
}(x;\omega )}{W^{\left( N_{m},\nu _{P\left( 1\right) },...,\nu _{P\left(
i-1\right) }\right) }(x;\omega )}\right) \right) ^{\prime },
\end{equation}%
with the convention that if a spectral index is repeated two times in the
tuple $\left( \nu _{1},...,\nu _{p}\right) $ characterizing the chain, then
we suppress the corresponding eigenfunction in the Wronskians.\newline
If $\nu _{P\left( i\right) }\in N_{m},$ the flip in $\nu _{P\left( i\right) }
$ is positive and the tuple $\left( N_{m},\nu _{P\left( 1\right) },...,\nu
_{P\left( i\right) }\right) $ contains one less index than $\left( N_{m},\nu
_{P\left( 1\right) },...,\nu _{P\left( i-1\right) }\right) $. Using Eq(\ref%
{WH}) and Eq(\ref{spec OH}), we deduce the first equality of Eq(\ref%
{soldress}).\newline
If $\nu _{P\left( i\right) }\notin N_{m},$ the flip in $\nu _{P\left(
i\right) }$ is negative and the tuple $\left( N_{m},\nu _{P\left( 1\right)
},...,\nu _{P\left( i\right) }\right) $ contains one more index than $\left(
N_{m},\nu _{P\left( 1\right) },...,\nu _{P\left( i-1\right) }\right) $.
Using Eq(\ref{WH}) and Eq(\ref{spec OH}), we deduce the second equality of
Eq(\ref{soldress}).
\end{proof}

\subsection{Examples}

\subsubsection{One-cyclic extensions of the HO and rational solutions of the
dressing chain of period $1$.}

For $p=1$, we have $k=1,\ j=0$ and then immediately \textit{Theorem 3} gives 
$N_{m}=\varnothing $. It means that the unique rational potential solving
the $1$-cyclic chain is the HO itself and the corresponding cyclic chain
reduces to $\left( 0\right) $ ie to the SUSY partnership. This is perfectly
coherent with our previous results on one step cyclicity (see Eq(\ref%
{RSOHfond})).

\subsubsection{Two-cyclic extensions of the HO and rational solutions of the
dressing chain of period $2$.}

For $p=2$, we have $k=2,\ j=0$ and then we deduce from \textit{Theorem 2}
that the Okamoto extension of the HO associated to the Maya diagram

\begin{equation}
N_{m}=\left( \left( 1\mid m\right) _{2}\right) =\left( 1,3,...,2m-1\right)
,\ 
\end{equation}%
solves the dressing chain of period $2$.

This particular type of Okamoto Maya diagram is called an \textbf{Umemura
staircase}. The corresponding $2$-cyclic chain is $\left( 0,2m+1\right) $.

Nevertheless, \textit{Theorem 3} is not applicable here since $p$ is even
and we cannot conclude that this is the most general rational solution of
the $2$-cyclic chain. In fact, we can write (see Eq(\ref{crum}))

\begin{equation}
V^{^{\left( N_{m}\right) }}(x;\omega )=V(x;\omega )-2\left( \log W^{\left(
1,3...,2m-1\right) }(x;\omega )\right) ^{\prime \prime },  \label{2ext}
\end{equation}%
where (see Eq(\ref{Eqext5}) and Eq(\ref{spec OH}))

\begin{equation}
W^{\left( 1,3...,2m-1\right) }(x;\omega )\propto
e^{-mz^{2}/2}W(H_{1}(z),H_{3}(z)...,H_{2m-1}(z)\mid z),
\end{equation}%
with $z=\sqrt{\omega /2}x$. But we also have \cite{magnus,szego}

\begin{equation}
\left\{ 
\begin{array}{c}
H_{2j+1}(z)=\left( -1\right) ^{j}2^{2j+1}j!t^{1/2}L_{j}^{1/2}(t) \\ 
H_{2j}(z)=\left( -1\right) ^{j}2^{2j}j!L_{j}^{-1/2}(t)%
\end{array}%
\right. ,  \label{correspHL}
\end{equation}%
where $L_{j}^{\alpha }$ is the usual Laguerre polynomial and $t=z^{2}$. It
results (see Eq(\ref{wronskprop}))

\begin{equation}
W^{\left( 1,3...,2m-1\right) }(x;\omega )\propto
e^{-mt/2}t^{m(m+1)/4}W(L_{0}^{1/2}(t),L_{1}^{1/2}(t)...,L_{m}^{1/2}(t)\mid
t).
\end{equation}

Using the well known derivation properties of the Laguerre polynomials \cite%
{magnus,szego}

\begin{equation}
\frac{dL_{j}^{\alpha }\left( t\right) }{dt}=-L_{j-1}^{\alpha +1}\left(
t\right) ,  \label{derivL}
\end{equation}%
with $L_{0}^{\alpha }(t)=1$ and $L_{-n}^{\alpha }(t)=0$, we obtain
straightforwardly

\begin{equation}
W(L_{0}^{1/2}(t),L_{1}^{1/2}(t)...,L_{m}^{1/2}(t)\mid t)=\left( -1\right)
^{m(m+1)/2}
\end{equation}%
and%
\begin{equation}
W^{\left( 1,3...,2m-1\right) }(x;\omega )\propto x^{m(m+1)/2}e^{-m\omega
x^{2}/4}.
\end{equation}

Substituting in Eq(\ref{2ext}), we arrive to

\begin{equation}
V^{^{\left( N_{m}\right) }}(x;\omega )=\omega ^{2}x^{2}/4+\frac{m(m+1)}{x^{2}%
}-\left( m+1/2\right) \omega =V(x;\omega ,m-1/2),
\end{equation}%
which is an isotonic oscillator with an integer angular momentum possessing,
as expected the trivial monodromy property.

In this lowest even case, we notice that using \textit{Theorem 2} doesn't
allow us to recover the general IO potential for arbitrary values of the $%
\alpha $ parameter, which is,as shown previously (see Eq(\ref{2 step cyclic2}%
)), the most general solution potential of the dressing chain of period $2$ .

\subsubsection{Three-cyclic extensions of the HO and rational solutions of
PIV.}

We have $p=3$ and consequently we can refer to \textit{Theorem 3} to
determine all the rational solutions of the dressing chain of period $3$, ie
of the PIV equation. For the possible values of $k$ (and $j$) we have only
two possibilities: $k=1\ \left( j=1\right) $ or $k=3\ \left( j=0\right) $.

\subsubsection{k=1}

The $3$-cyclic Maya diagram with $k=1$ is a $3$-GH Maya diagram of the form

\begin{equation}
\left( \lambda \mid \mu \right) _{1}=\left( \lambda ,...,\lambda +\mu
-1\right) =H_{\lambda ,\mu }.  \label{3GH}
\end{equation}

Following Clarkson's terminology \cite{clarkson,clarkson2,clarkson3}, the $3$%
-cyclic extensions $V^{^{\left( H_{\lambda ,\mu }\right) }}(x;\omega )$ are
called the \textbf{3-generalized Hermite (3-GH)} extensions.

From the Krein-Adler theorem \cite{krein,adler2,GGM1}, we deduce immediately
that 3-GH extensions which are regular on the real line are those for which $%
\mu $ is even.

The possible $3$-cyclic chains corresponding to $H_{\lambda ,\mu }$ are
built by permutation of $\left\{ 0,\lambda ,\lambda +\mu \right\} $. For the
particular choice $\left( 0,\lambda ,\lambda +\mu \right) $, the parameters
in the dressing chain system of period $3$ (see Eq(\ref{pcyclicdress})) are

\begin{equation}
\left\{ 
\begin{array}{c}
\varepsilon _{12}=\left( -\lambda -\mu \right) \omega \\ 
\varepsilon _{23}=\mu \omega \\ 
\varepsilon _{31}-\Delta =\left( \lambda -1\right) \omega .%
\end{array}%
\right.
\end{equation}

As for the associated parameters in the PIV\ equation associated to this
chain they are given by

\begin{equation}
a=-\left( 1-\mu -2\lambda \right) ,\quad b=-2\mu ^{2},
\end{equation}%
or

\begin{equation}
a=m\in 
\mathbb{Z}
,\quad b=-2\left( 1+m-2n\right) ^{2}.
\end{equation}

The same type of result (up to redefinition of $a$ and $b$) can be obtained
for all the possible choices of $\varepsilon _{1}/\omega $ in the set $%
\left\{ 0,\lambda ,\lambda +\mu \right\} $.

For $\varepsilon _{1}/\omega =0$, the corresponding solution of PIV is
obtained as (see Eq(\ref{chv}) and Eq(\ref{WH}))

\begin{equation}
y_{0}=\sqrt{\frac{2}{\omega }}\left( w_{0}^{\left( H_{\lambda ,\mu }\right)
}-w_{0}\right) =-\frac{d}{dz}\left( \log \left( \frac{\psi _{0}^{\left(
H_{\lambda ,\mu }\right) }}{\psi _{0}}\right) \right) ,
\end{equation}%
with ($z=\sqrt{\omega /2}x$)

\begin{equation}
\frac{\psi _{0}^{\left( H_{\lambda ,\mu }\right) }}{\psi _{0}}\propto \frac{%
W^{\left( 0,H_{\lambda ,\mu }\right) }\left( x;\omega \right) }{W^{\left(
H_{\lambda ,\mu }\right) }\left( x;\omega \right) \psi _{0}\left( x;\omega
\right) }\propto \frac{\mathcal{H}^{\left( H_{\lambda -1,\mu }\right)
}\left( z\right) }{\mathcal{H}^{\left( H_{\lambda ,\mu }\right) }\left(
z\right) },
\end{equation}%
where we have used Eq(\ref{Eqext5}) and the fact that

\begin{equation}
H_{\lambda ,\mu }-1=H_{\lambda -1,\mu }.
\end{equation}

Then

\begin{equation}
y_{0}\left( z\right) =\frac{d}{dz}\log \left( \frac{\mathcal{H}^{\left(
H_{\lambda ,\mu }\right) }\left( z\right) }{\mathcal{H}^{\left( H_{\lambda
-1,\mu }\right) }\left( z\right) }\right) .
\end{equation}

The second possible choice is $\varepsilon _{1}/\omega =\lambda $ in which
case

\begin{equation}
y_{\lambda }=\sqrt{\frac{2}{\omega }}\left( w_{\lambda }^{\left( H_{\lambda
,\mu }\right) }-w_{0}\right) =-\frac{d}{dz}\left( \log \left( \frac{\psi
_{\lambda }^{\left( H_{\lambda ,\mu }\right) }}{\psi _{0}}\right) \right) ,
\end{equation}%
where, using Eq(\ref{Eqext5}) and the fact that $\left( H_{\lambda ,\mu
},\lambda \right) =\left( \lambda +1,...,\lambda +\mu -1\right) =H_{\lambda
+1,\mu -1}$

\begin{equation}
\frac{\psi _{\lambda }^{\left( H_{\lambda ,\mu }\right) }}{\psi _{0}}\propto 
\frac{W^{\left( \lambda ,H_{\lambda ,\mu }\right) }\left( x;\omega \right) }{%
W^{\left( H_{\lambda ,\mu }\right) }\left( x;\omega \right) \psi _{0}\left(
x;\omega \right) }\propto \frac{\mathcal{H}^{\left( H_{\lambda +1,\mu
-1}\right) }\left( z\right) }{\mathcal{H}^{\left( H_{\lambda ,\mu }\right)
}\left( z\right) \psi _{0}^{2}\left( x;\omega \right) }\propto e^{z^{2}}%
\frac{\mathcal{H}^{\left( H_{\lambda +1,\mu -1}\right) }\left( z\right) }{%
\mathcal{H}^{\left( H_{\lambda ,\mu }\right) }\left( z\right) }.
\end{equation}

Then

\begin{equation}
y_{\lambda }\left( z\right) =-2z+\frac{d}{dz}\log \left( \frac{\mathcal{H}%
^{\left( H_{\lambda ,\mu }\right) }\left( z\right) }{\mathcal{H}^{\left(
H_{\lambda +1,\mu -1}\right) }\left( z\right) }\right) .
\end{equation}

The last possible choice is $\varepsilon _{1}/\omega =\lambda +\mu $ giving

\begin{equation}
y_{\lambda +\mu }=\sqrt{\frac{2}{\omega }}\left( w_{\lambda +\mu }^{\left(
H_{\lambda ,\mu }\right) }-w_{0}\right) =-\frac{d}{dz}\left( \log \left( 
\frac{\psi _{\lambda +\mu }^{\left( H_{\lambda ,\mu }\right) }}{\psi _{0}}%
\right) \right) ,
\end{equation}%
with ($\left( H_{\lambda ,\mu },\lambda +\mu \right) =\left( \lambda
,...,\lambda +\mu \right) =H_{\lambda ,\mu +1}$)

\begin{equation}
\frac{\psi _{\lambda +\mu }^{\left( H_{\lambda ,\mu }\right) }}{\psi _{0}}%
\propto \frac{W^{\left( \lambda +\mu ,H_{\lambda ,\mu }\right) }\left(
x;\omega \right) }{W^{\left( H_{\lambda ,\mu }\right) }\left( x;\omega
\right) \psi _{0}\left( x;\omega \right) }\propto \frac{\mathcal{H}^{\left(
H_{\lambda ,\mu +1}\right) }\left( z\right) }{\mathcal{H}^{\left( H_{\lambda
,\mu }\right) }\left( z\right) }.
\end{equation}

Then

\begin{equation}
y_{\lambda +\mu }\left( z\right) =\frac{d}{dz}\log \left( \frac{\mathcal{H}%
^{\left( H_{\lambda ,\mu }\right) }\left( z\right) }{\mathcal{H}^{\left(
H_{\lambda ,\mu +1}\right) }\left( z\right) }\right) .
\end{equation}

We retrieve then the three usual form for the rational solutions associated
to the generalized Hermite polynomials, namely (with $k=1$, $z=t$) \cite%
{clarkson,clarkson2,clarkson3}

\begin{equation}
\left\{ 
\begin{array}{c}
y_{0}(t)=\frac{d}{dt}\log \left( \mathcal{H}^{\left( H_{\lambda ,\mu
}\right) }\left( t\right) /\mathcal{H}^{\left( H_{\lambda -1,\mu }\right)
}\left( t\right) \right) \\ 
y_{\lambda }(t)=-2t+\frac{d}{dt}\log \left( \mathcal{H}^{\left( H_{\lambda
,\mu }\right) }\left( t\right) /\mathcal{H}^{\left( H_{\lambda +1,\mu
-1}\right) }\left( t\right) \right) \\ 
y_{\lambda +\mu }(t)=\frac{d}{dt}\log \left( \mathcal{H}^{\left( H_{\lambda
,\mu }\right) }\left( t\right) /\mathcal{H}^{\left( H_{\lambda ,\mu
+1}\right) }\left( t\right) \right) ,%
\end{array}%
\right.
\end{equation}%
obtained for the following integer values of the $a$ and $b$ parameters in
Eq(\ref{PIV})

\begin{equation}
a=-\left( 1-\mu -2\lambda \right) ,\quad b=-2\mu ^{2}.
\end{equation}

\subsubsection{k=3}

Now consider the case $k=3$. The $3$-cyclic Maya diagram with $k=3$ is a $3$%
-Okamoto Maya diagram of the form

\begin{equation}
\left( \left( 1\mid \alpha _{1}\right) _{3},\left( 2\mid \alpha _{2}\right)
_{3}\right) =\left( 1,...,1+3(\alpha _{1}-1);2,...,2+3(\alpha _{2}-1\right)
)=\Omega _{\alpha _{1},\alpha _{2}}.
\end{equation}

Note that for some small values of $\alpha _{1}$ and $\alpha _{2}$, the
3-Okamoto Maya diagrams are also 3-GH Maya diagram set

\begin{equation}
\Omega _{1,1}=\left( 1,2\right) =H_{1,2};\ \Omega _{1,0}=\left( 1\right)
=H_{1,1};\ \Omega _{0,1}=\left( 2\right) =H_{2,1}.
\end{equation}

The Krein-Adler theorem implies that the 3-Okamoto extensions which are
regular on the real line are those for which $\alpha _{1}=\alpha _{2}$,
namely correspond to 3-Okamoto Maya diagrams of the form $\Omega _{\alpha
,\alpha }$.

The possible $3$-cyclic chains associated to $\Omega _{\alpha _{1},\alpha
_{2}}$ are\ built by permutation from $\left\{ 0,1+3\alpha _{1},2+3\alpha
_{2}\right\} $. For the chain $\left( 0,1+3\alpha _{1},2+3\alpha _{2}\right) 
$, the parameters of the dressing chain system of period $3$ (see Eq(\ref%
{pcyclicdress})) are

\begin{equation}
\left\{ 
\begin{array}{c}
\varepsilon _{12}=\left( -1+3(\alpha _{1}-\alpha _{2})\right) \omega \\ 
\varepsilon _{23}=\left( 2+3\alpha _{2}\right) \omega \\ 
\varepsilon _{31}-\Delta =\left( -4-3\alpha _{1}\right) \omega .%
\end{array}%
\right.
\end{equation}

The corresponding parameters in the PIV\ equation are then given by ($\omega
=1,\Delta =k=3$)

\begin{equation}
a=\alpha _{1}+\alpha _{2},\quad b=-\frac{2}{9}\left( -1+3(\alpha _{1}-\alpha
_{2})\right) ^{2},
\end{equation}%
or

\begin{equation}
a=j\in 
\mathbb{Z}
,\quad b=-2\left( 1/3-j+2\alpha _{2}\right) ^{2}.
\end{equation}

The same result (up to redefinition of the integers) can be obtained with
all the possible choices of $\varepsilon _{1}$ in the set $\left\{
0,1+3\alpha _{1},2+3\alpha _{2}\right\} $.

For $\varepsilon _{1}/\omega =0$, the corresponding solution of PIV is
obtained as (see Eq(\ref{chv}))

\begin{equation}
y_{0}=\sqrt{\frac{2}{\Delta }}\left( w_{0}^{\left( \Omega _{\alpha
_{1},\alpha _{2}}\right) }-3w_{0}\right) =-\sqrt{\frac{2}{3\omega }}\frac{d}{%
dx}\left( \log \left( \frac{\psi _{0}^{\left( \Omega _{\alpha _{1},\alpha
_{2}}\right) }}{\psi _{0}^{3}}\right) \right) ,
\end{equation}%
with ($z=\sqrt{\omega /2}x$). If we suppose $\alpha _{1},\alpha _{2}\geq 2$,
by using Eq(\ref{Eqext5}), we obtain 
\begin{equation}
\frac{\psi _{0}^{\left( \Omega _{\alpha _{1},\alpha _{2}}\right) }}{\psi
_{0}^{3}}\propto \frac{W^{\left( \Omega _{\alpha _{1}-1,\alpha _{2}-1}\oplus
3\right) }\left( x;\omega \right) }{W^{\left( \Omega _{\alpha _{1},\alpha
_{2}}\right) }\left( x;\omega \right) \psi _{0}^{3}(x;\omega )}\propto \frac{%
\mathcal{H}^{\left( \Omega _{\alpha _{1}-1,\alpha _{2}-1}\right) }\left(
z\right) }{\mathcal{H}^{\left( \Omega _{\alpha _{1},\alpha _{2}}\right)
}\left( z\right) \psi _{0}^{2}(x;\omega )}\propto e^{z^{2}}\frac{\mathcal{H}%
^{\left( \Omega _{\alpha _{1}-1,\alpha _{2}-1}\right) }\left( z\right) }{%
\mathcal{H}^{\left( \Omega _{\alpha _{1},\alpha _{2}}\right) }\left(
z\right) }.
\end{equation}
Note that $\left( 0,\Omega _{\alpha _{1},\alpha
_{2}}\right) =\left( 0,1,2\right) \cup \left( 4,...,1+3(\alpha
_{1}-1);5,...,2+3(\alpha _{2}-1\right) )=\Omega _{\alpha _{1}-1,\alpha
_{2}-1}\oplus 3$.

Then

\begin{equation}
y_{0}=-\sqrt{\frac{2\omega }{3}}x+\sqrt{\frac{2}{3\omega }}\frac{d}{dx}\log
\left( \frac{\mathcal{H}^{\left( \Omega _{\alpha _{1},\alpha _{2}}\right)
}\left( z\right) }{\mathcal{H}^{\left( \Omega _{\alpha _{1}-1,\alpha
_{2}-1}\right) }\left( z\right) }\right) ,
\end{equation}%
that is,

\begin{equation}
y_{0}(t)=-\frac{2}{3}t+\frac{d}{dt}\log \left( \frac{\mathcal{H}^{\left(
\Omega _{\alpha _{1},\alpha _{2}}\right) }\left( t/\sqrt{3}\right) }{%
\mathcal{H}^{\left( \Omega _{\alpha _{1}-1,\alpha _{2}-1}\right) }\left( t/%
\sqrt{3}\right) }\right) .
\end{equation}

The second possible choice is $\varepsilon _{1}/\omega =1+3\alpha _{1}$ in
which case

\begin{equation}
y_{1+3\alpha _{1}}(x)=\sqrt{\frac{2}{3\omega }}\left( w_{1+3\alpha
_{1}}^{\left( \Omega _{\alpha _{1},\alpha _{2}}\right) }-3w_{0}\right) =-%
\sqrt{\frac{2}{3\omega }}\frac{d}{dx}\left( \log \left( \frac{\psi
_{1+3\alpha _{1}}^{\left( \Omega _{\alpha _{1},\alpha _{2}}\right) }}{\psi
_{0}^{3}}\right) \right) ^{\prime },
\end{equation}%
with ($\left( 1+3\alpha _{1},\Omega _{\alpha _{1},\alpha _{2}}\right)
=\Omega _{\alpha _{1}+1,\alpha _{2}}$)

\begin{equation}
\frac{\psi _{1+3\alpha _{1}}^{\left( \Omega _{\alpha _{1},\alpha
_{2}}\right) }}{\psi _{0}}=\frac{W^{\left( \Omega _{\alpha _{1}+1,\alpha
_{2}}\right) }\left( x;\omega \right) }{W^{\left( \Omega _{\alpha
_{1},\alpha _{2}}\right) }\left( x;\omega \right) \psi _{0}^{3}(x;\omega )}%
\propto \frac{\mathcal{H}^{\left( \Omega _{\alpha _{1}+1,\alpha _{2}}\right)
}\left( z\right) }{\mathcal{H}^{\left( \Omega _{\alpha _{1},\alpha
_{2}}\right) }\left( z\right) \psi _{0}^{2}(x;\omega )}\propto e^{z^{2}}%
\frac{\mathcal{H}^{\left( \Omega _{\alpha _{1}+1,\alpha _{2}}\right) }\left(
z\right) }{\mathcal{H}^{\left( \Omega _{\alpha _{1},\alpha _{2}}\right)
}\left( z\right) }.
\end{equation}

Then

\begin{equation}
y_{1+3\alpha _{1}}=-\sqrt{\frac{2\omega }{3}}x+\sqrt{\frac{2}{3\omega }}%
\frac{d}{dx}\log \left( \frac{\mathcal{H}^{\left( \Omega _{\alpha
_{1},\alpha _{2}}\right) }\left( z\right) }{\mathcal{H}^{\left( \Omega
_{\alpha _{1}+1,\alpha _{2}}\right) }\left( z\right) }\right) ,
\end{equation}%
that is

\begin{equation}
y(t)=-\frac{2}{3}t+\frac{d}{dt}\log \left( \frac{\mathcal{H}^{\left( \Omega
_{\alpha _{1},\alpha _{2}}\right) }\left( t/\sqrt{3}\right) }{\mathcal{H}%
^{\left( \Omega _{\alpha _{1}+1,\alpha _{2}}\right) }\left( t/\sqrt{3}%
\right) }\right) .
\end{equation}

The last possible choice is $\varepsilon _{1}/\omega =2+3\alpha _{2}$ giving

\begin{equation}
y_{2+3\alpha _{2}}=\sqrt{\frac{2}{\Delta }}\left( w_{2+3\alpha _{2}}^{\left(
\Omega _{\alpha _{1},\alpha _{2}}\right) }-3w_{0}\right) =-\sqrt{\frac{2}{%
3\omega }}\frac{d}{dx}\left( \log \left( \frac{\psi _{2+3\alpha
_{2}}^{\left( \Omega _{\alpha _{1},\alpha _{2}}\right) }}{\psi _{0}}\right)
\right) ,
\end{equation}%
with ($\left( 2+3\alpha _{1},\Omega _{\alpha _{1},\alpha _{2}}\right)
=\Omega _{\alpha _{1},\alpha _{2}+1}$)

\begin{equation}
\frac{\psi _{2+3\alpha _{2}}^{\left( \Omega _{\alpha _{1},\alpha
_{2}}\right) }}{\psi _{0}}\propto \frac{W^{\left( \Omega _{\alpha
_{1},\alpha _{2}+1}\right) }\left( x;\omega \right) }{W^{\left( \Omega
_{\alpha _{1},\alpha _{2}}\right) }\left( x;\omega \right) \psi
_{0}^{3}(x;\omega )}\propto \frac{\mathcal{H}^{\left( \Omega _{\alpha
_{1},\alpha _{2}+1}\right) }\left( z\right) }{\mathcal{H}^{\left( \Omega
_{\alpha _{1},\alpha _{2}}\right) }\left( z\right) \psi _{0}^{2}(x;\omega )}%
\propto e^{z^{2}}\frac{\mathcal{H}^{\left( \Omega _{\alpha _{1},\alpha
_{2}+1}\right) }\left( z\right) }{\mathcal{H}^{\left( \Omega _{\alpha
_{1},\alpha _{2}}\right) }\left( z\right) }.
\end{equation}

Then

\begin{equation}
y_{2+3\alpha _{2}}=-\sqrt{\frac{2\omega }{3}}x+\sqrt{\frac{2}{3\omega }}%
\frac{d}{dx}\log \left( \frac{\mathcal{H}^{\left( \Omega _{\alpha
_{1},\alpha _{2}}\right) }\left( z\right) }{\mathcal{H}^{\left( \Omega
_{\alpha _{1},\alpha _{2}+1}\right) }\left( z\right) }\right) ,
\end{equation}%
that is

\begin{equation}
y_{2+3\alpha _{2}}\left( t\right) =-\frac{2}{3}t+\frac{d}{dt}\log \left( 
\frac{\mathcal{H}^{\left( \Omega _{\alpha _{1},\alpha _{2}}\right) }\left( t/%
\sqrt{3}\right) }{\mathcal{H}^{\left( \Omega _{\alpha _{1},\alpha
_{2}+1}\right) }\left( t/\sqrt{3}\right) }\right) .
\end{equation}

We retrieve then the three usual form for the rational solutions associated
to the generalized Hermite polynomials, namely (with $k=1$, $z=t$) \cite%
{clarkson,clarkson2,clarkson3}

\begin{equation}
\left\{ 
\begin{array}{c}
y(t)=-2t/3+\frac{d}{dt}\log \left( \mathcal{H}^{\left( \Omega _{\alpha
_{1},\alpha _{2}}\right) }\left( t/\sqrt{3}\right) /\mathcal{H}^{\left(
\Omega _{\alpha _{1}-1,\alpha _{2}-1}\right) }\left( t/\sqrt{3}\right)
\right) \\ 
y(t)=-2t/3+\frac{d}{dt}\log \left( \mathcal{H}^{\left( \Omega _{\alpha
_{1},\alpha _{2}}\right) }\left( t/\sqrt{3}\right) /\mathcal{H}^{\left(
\Omega _{\alpha _{1}+1,\alpha _{2}}\right) }\left( t/\sqrt{3}\right) \right)
\\ 
y(t)=-2t/3+\frac{d}{dt}\log \left( \mathcal{H}^{\left( \Omega _{\alpha
_{1},\alpha _{2}}\right) }\left( t/\sqrt{3}\right) /\mathcal{H}^{\left(
\Omega _{\alpha _{1},\alpha _{2}+1}\right) }\left( t/\sqrt{3}\right) \right)
,%
\end{array}%
\right.
\end{equation}%
obtained for the following integer values of the $a$ and $b$ parameters in
Eq(\ref{PIV})

\begin{equation}
a=\alpha _{1}+\alpha _{2}=j\in 
\mathbb{Z}
,\quad b=-2\left( 1/3-j+2\alpha _{2}\right) ^{2}.
\end{equation}

\subsection{5-cyclic chains and new solutions of the A$_{\text{4}}$-PIV}

We have $p=5$ and consequently we can have $k=1\ \left( j=2\right) ,\ k=3\
\left( j=1\right) $ or $k=5\ \left( j=0\right) $.

\subsubsection{k=1}

The $5$-cyclic extension with $k=1$ is associated to a $5$-GH Maya diagram
of the form

\begin{equation}
N_{m}=\left( \left( \lambda _{1}\mid \mu _{1}\right) _{1},\left( \lambda
_{2}\mid \mu _{2}\right) _{1}\right) ,\ m=\mu _{1}+\mu _{2}.
\end{equation}

When $\lambda _{1}+\mu _{1}\geq \lambda _{2}$, the block structure of $N_{m}$
is degenerate and we recover a $3$-GH Maya diagram $H_{\lambda ,\mu }$ (see
Eq(\ref{3GH})), which is $3$ but also trivially $5$-cyclic. For simplicity
we suppose in the following $\lambda _{1}+\mu _{1}<\lambda _{2}$. Explicitely

\begin{equation}
N_{m}=\left( \lambda _{1},...,\lambda _{1}+\mu _{1}-1,\lambda
_{2},...,\lambda _{2}+\mu _{2}-1\right) .
\end{equation}

The corresponding $5$-cyclic chain is built by permutation of

\begin{equation}
\left\{ \lambda _{1},\lambda _{1}+\mu _{1},\lambda _{2},\lambda _{2}+\mu
_{2},0\right\} .
\end{equation}

For the chain $\left( \lambda _{1},\lambda _{1}+\mu _{1},\lambda
_{2},\lambda _{2}+\mu _{2},0\right) $, the parameters of the dressing chain
system of period 5 (see Eq(\ref{pcyclicdress})) are

\begin{equation}
\left\{ 
\begin{array}{c}
\varepsilon _{12}=-\mu _{1}\omega \\ 
\varepsilon _{23}=\left( \lambda _{1}-\lambda _{2}+\mu _{1}\right) \omega \\ 
\varepsilon _{34}=-\mu _{2}\omega \\ 
\varepsilon _{45}=\left( \lambda _{2}+\mu _{2}\right) \omega \\ 
\varepsilon _{51}-\Delta =-\left( \lambda _{1}+1\right) \omega ,%
\end{array}%
\right.
\end{equation}%
and the solutions of this system are given by (see Eq(\ref{soldress}))

\begin{eqnarray}
w_{1}(x) &=&w_{\lambda _{1}}^{\left( \left( \lambda _{1}\mid \mu _{1}\right)
_{1},\left( \lambda _{2}\mid \mu _{2}\right) _{1}\right) }(x)  \notag \\
&=&-\omega x/2+\sqrt{\frac{\omega }{2}}\frac{d}{dz}\left( \log \left( \frac{%
\mathcal{H}^{\left( \left( \lambda _{1}\mid \mu _{1}\right) _{1},\left(
\lambda _{2}\mid \mu _{2}\right) _{1}\right) }\left( z\right) }{\mathcal{H}%
^{\left( \lambda _{1},\left( \lambda _{1}\mid \mu _{1}\right) _{1},\left(
\lambda _{2}\mid \mu _{2}\right) _{1}\right) }\left( z\right) }\right)
\right)  \notag \\
&=&-\omega x/2+\sqrt{\frac{\omega }{2}}\frac{d}{dz}\left( \log \left( \frac{%
\mathcal{H}^{\left( \left( \lambda _{1}\mid \mu _{1}\right) _{1},\left(
\lambda _{2}\mid \mu _{2}\right) _{1}\right) }\left( z\right) }{\mathcal{H}%
^{\left( \left( \lambda _{1}+1\mid \mu _{1}-1\right) _{1},\left( \lambda
_{2}\mid \mu _{2}\right) _{1}\right) }\left( z\right) }\right) \right) ,
\end{eqnarray}

\begin{eqnarray}
w_{2}(x) &=&w_{\lambda _{1}+\mu _{1}}^{\left( \lambda _{1},N_{m}\right) }(x)
\notag \\
&=&\omega x/2+\sqrt{\frac{\omega }{2}}\frac{d}{dz}\left( \log \left( \frac{%
\mathcal{H}^{\left( \left( \lambda _{1}+1\mid \mu _{1}-1\right) _{1},\left(
\lambda _{2}\mid \mu _{2}\right) _{1}\right) }\left( z\right) }{\mathcal{H}%
^{\left( \lambda _{1}+\mu _{1},\left( \lambda _{1}+1\mid \mu _{1}-1\right)
_{1},\left( \lambda _{2}\mid \mu _{2}\right) _{1}\right) }\left( z\right) }%
\right) \right)  \notag \\
&=&\omega x/2+\sqrt{\frac{\omega }{2}}\frac{d}{dz}\left( \log \left( \frac{%
\mathcal{H}^{\left( \left( \lambda _{1}+1\mid \mu _{1}-1\right) _{1},\left(
\lambda _{2}\mid \mu _{2}\right) _{1}\right) }\left( z\right) }{\mathcal{H}%
^{\left( \left( \lambda _{1}+1\mid \mu _{1}\right) _{1},\left( \lambda
_{2}\mid \mu _{2}\right) _{1}\right) }\left( z\right) }\right) \right) ,
\end{eqnarray}

\begin{eqnarray}
w_{3}(x) &=&w_{\lambda _{2}}^{\left( \left( \lambda _{1}+1\mid \mu
_{1}\right) _{1},\left( \lambda _{2}\mid \mu _{2}\right) _{1}\right) }(x) 
\notag \\
&=&-\omega x/2+\sqrt{\frac{2}{\omega }}\frac{d}{dz}\left( \log \left( \frac{%
\mathcal{H}^{\left( \left( \lambda _{1}+1\mid \mu _{1}\right) _{1},\left(
\lambda _{2}\mid \mu _{2}\right) _{1}\right) }\left( z\right) }{\mathcal{H}%
^{\left( \lambda _{2},\left( \lambda _{1}+1\mid \mu _{1}\right) _{1},\left(
\lambda _{2}\mid \mu _{2}\right) _{1}\right) }\left( z\right) }\right)
\right)  \notag \\
&=&-\omega x/2+\sqrt{\frac{2}{\omega }}\frac{d}{dz}\left( \log \left( \frac{%
\mathcal{H}^{\left( \left( \lambda _{1}+1\mid \mu _{1}-1\right) _{1},\left(
\lambda _{2}\mid \mu _{2}\right) _{1}\right) }\left( z\right) }{\mathcal{H}%
^{\left( \left( \lambda _{1}+1\mid \mu _{1}\right) _{1},\left( \lambda
_{2}+1\mid \mu _{2}-1\right) _{1}\right) }\left( z\right) }\right) \right) ,
\end{eqnarray}

\begin{eqnarray}
w_{4}(x) &=&w_{\lambda _{2}+\mu _{2}}^{\left( \left( \lambda _{1}+1\mid \mu
_{1}\right) _{1},\left( \lambda _{2}+1\mid \mu _{2}-1\right) _{1}\right) }(x)
\notag \\
&=&\omega x/2+\sqrt{\frac{\omega }{2}}\frac{d}{dz}\left( \log \left( \frac{%
\mathcal{H}^{\left( \left( \lambda _{1}+1\mid \mu _{1}\right) _{1},\left(
\lambda _{2}+1\mid \mu _{2}-1\right) _{1}\right) }\left( z\right) }{\mathcal{%
H}^{\left( \lambda _{2}+\mu _{2},\left( \lambda _{1}+1\mid \mu _{1}\right)
_{1},\left( \lambda _{2}+1\mid \mu _{2}-1\right) _{1}\right) }\left(
z\right) }\right) \right)  \notag \\
&=&\omega x/2+\sqrt{\frac{\omega }{2}}\frac{d}{dz}\left( \log \left( \frac{%
\mathcal{H}^{\left( \left( \lambda _{1}+1\mid \mu _{1}\right) _{1},\left(
\lambda _{2}+1\mid \mu _{2}-1\right) _{1}\right) }\left( z\right) }{\mathcal{%
H}^{\left( \left( \lambda _{1}+1\mid \mu _{1}\right) _{1},\left( \lambda
_{2}+1\mid \mu _{2}\right) _{1}\right) }\left( z\right) }\right) \right) ,
\end{eqnarray}%
and (see Eq(\ref{Eqext5}))

\begin{eqnarray}
w_{5}(x) &=&w_{0}^{\left( \left( \lambda _{1}+1\mid \mu _{1}\right)
_{1},\left( \lambda _{2}+1\mid \mu _{2}\right) _{1}\right) }(x)  \notag \\
&=&\omega x/2+\sqrt{\frac{\omega }{2}}\frac{d}{dz}\left( \log \left( \frac{%
\mathcal{H}^{\left( \left( \lambda _{1}+1\mid \mu _{1}\right) _{1},\left(
\lambda _{2}+1\mid \mu _{2}\right) _{1}\right) }\left( z\right) }{\mathcal{H}%
^{\left( 0,\left( \lambda _{1}+1\mid \mu _{1}\right) _{1},\left( \lambda
_{2}+1\mid \mu _{2}\right) _{1}\right) }\left( z\right) }\right) \right) 
\notag \\
&=&\omega x/2+\sqrt{\frac{\omega }{2}}\frac{d}{dz}\left( \log \left( \frac{%
\mathcal{H}^{\left( \left( \lambda _{1}+1\mid \mu _{1}\right) _{1},\left(
\lambda _{2}+1\mid \mu _{2}\right) _{1}\right) }\left( z\right) }{\mathcal{H}%
^{\left( \left( \lambda _{1}\mid \mu _{1}\right) _{1},\left( \lambda
_{2}\mid \mu _{2}\right) _{1}\right) }\left( z\right) }\right) \right) .
\end{eqnarray}

\subsubsection{k=5}

The $5$-cyclic extension with $k=5$ is associated to a $5$-Okamoto Maya
diagram of the form

\begin{equation}
N_{m}=\left( \left( 1\mid \alpha _{1}\right) _{5},\left( 2\mid \alpha
_{2}\right) _{5},\left( 3\mid \alpha _{3}\right) _{5},\left( 4\mid \alpha
_{4}\right) _{5}\right) =\Omega _{\alpha _{1},\alpha _{2},\alpha _{3},\alpha
_{4}},\ m=\alpha _{1}+\alpha _{2}+\alpha _{3}+\alpha _{4}.
\end{equation}

The corresponding $5$-cyclic chain are obtained by permutation from

\begin{equation}
\left\{ 1+5\alpha _{1},2+5\alpha _{2},3+5\alpha _{3},4+5\alpha
_{4},0\right\} .
\end{equation}

The parameters in the dressing chain system of period 5 (see Eq(\ref%
{pcyclicdress})) for the chain $\left( 1+5\alpha _{1},2+5\alpha
_{2},3+5\alpha _{3},4+5\alpha _{4},0\right) $ are given by

\begin{equation}
\left\{ 
\begin{array}{c}
\varepsilon _{12}=\left( -1-5\left( \alpha _{2}-\alpha _{1}\right) \right)
\omega \\ 
\varepsilon _{23}=\left( -1-5\left( \alpha _{3}-\alpha _{2}\right) \right)
\omega \\ 
\varepsilon _{34}=\left( -1-5\left( \alpha _{4}-\alpha _{1}\right) \right)
\omega \\ 
\varepsilon _{45}=\left( 5\alpha _{4}+4\right) \omega \\ 
\varepsilon _{51}-\Delta =\left( -6-5\alpha _{1}\right) \omega .%
\end{array}%
\right.
\end{equation}%
and the corresponding solutions of the dressing chain system are (see Eq(\ref%
{soldress}))

\begin{eqnarray}
w_{2}(x) &=&w_{1+5\alpha _{1}}^{\left( \left( 1\mid \alpha _{1}\right)
_{5},\left( 2\mid \alpha _{2}\right) _{5},\left( 3\mid \alpha _{3}\right)
_{5},\left( 4\mid \alpha _{4}\right) _{5}\right) }(x)  \notag \\
&=&\omega x/2+\sqrt{\frac{\omega }{2}}\frac{d}{dz}\left( \log \left( \frac{%
\mathcal{H}^{\left( \left( 1\mid \alpha _{1}\right) _{5},\left( 2\mid \alpha
_{2}\right) _{5},\left( 3\mid \alpha _{3}\right) _{5},\left( 4\mid \alpha
_{4}\right) _{5}\right) }\left( z\right) }{\mathcal{H}^{\left( 1+5\alpha
_{1},\left( 1\mid \alpha _{1}\right) _{5},\left( 2\mid \alpha _{2}\right)
_{5},\left( 3\mid \alpha _{3}\right) _{5},\left( 4\mid \alpha _{4}\right)
_{5}\right) }\left( z\right) }\right) \right)  \notag \\
&=&\omega x/2+\sqrt{\frac{\omega }{2}}\frac{d}{dz}\left( \log \left( \frac{%
\mathcal{H}^{\left( \left( 1\mid \alpha _{1}\right) _{5},\left( 2\mid \alpha
_{2}\right) _{5},\left( 3\mid \alpha _{3}\right) _{5},\left( 4\mid \alpha
_{4}\right) _{5}\right) }\left( z\right) }{\mathcal{H}^{\left( \left( 1\mid
\alpha _{1}+1\right) _{5},\left( 2\mid \alpha _{2}\right) _{5},\left( 3\mid
\alpha _{3}\right) _{5},\left( 4\mid \alpha _{4}\right) _{5}\right) }\left(
z\right) }\right) \right) ,
\end{eqnarray}

\begin{eqnarray}
w_{2}(x) &=&w_{2+5\alpha _{2}}^{\left( \left( 1\mid \alpha _{1}+1\right)
_{5},\left( 2\mid \alpha _{2}\right) _{5},\left( 3\mid \alpha _{3}\right)
_{5},\left( 4\mid \alpha _{4}\right) _{5}\right) }(x)  \notag \\
&=&\omega x/2+\sqrt{\frac{\omega }{2}}\frac{d}{dz}\left( \log \left( \frac{%
\mathcal{H}^{\left( \left( 1\mid \alpha _{1}+1\right) _{5},\left( 2\mid
\alpha _{2}\right) _{5},\left( 3\mid \alpha _{3}\right) _{5},\left( 4\mid
\alpha _{4}\right) _{5}\right) }\left( z\right) }{\mathcal{H}^{\left(
2+5\alpha _{2},\left( 1\mid \alpha _{1}+1\right) _{5},\left( 2\mid \alpha
_{2}\right) _{5},\left( 3\mid \alpha _{3}\right) _{5},\left( 4\mid \alpha
_{4}\right) _{5}\right) }\left( z\right) }\right) \right)  \notag \\
&=&\omega x/2+\sqrt{\frac{\omega }{2}}\frac{d}{dz}\left( \log \left( \frac{%
\mathcal{H}^{\left( \left( 1\mid \alpha _{1}+1\right) _{5},\left( 2\mid
\alpha _{2}\right) _{5},\left( 3\mid \alpha _{3}\right) _{5},\left( 4\mid
\alpha _{4}\right) _{5}\right) }\left( z\right) }{\mathcal{H}^{\left( \left(
1\mid \alpha _{1}+1\right) _{5},\left( 2\mid \alpha _{2}+1\right)
_{5},\left( 3\mid \alpha _{3}\right) _{5},\left( 4\mid \alpha _{4}\right)
_{5}\right) }\left( z\right) }\right) \right) ,
\end{eqnarray}

\begin{eqnarray}
w_{3}(x) &=&w_{3+5\alpha _{3}}^{\left( \left( 1\mid \alpha _{1}+1\right)
_{5},\left( 2\mid \alpha _{2}+1\right) _{5},\left( 3\mid \alpha _{3}\right)
_{5},\left( 4\mid \alpha _{4}\right) _{5}\right) }(x)  \notag \\
&=&\omega x/2+\sqrt{\frac{\omega }{2}}\frac{d}{dz}\left( \log \left( \frac{%
\mathcal{H}^{\left( \left( 1\mid \alpha _{1}+1\right) _{5},\left( 2\mid
\alpha _{2}+1\right) _{5},\left( 3\mid \alpha _{3}\right) _{5},\left( 4\mid
\alpha _{4}\right) _{5}\right) }\left( z\right) }{\mathcal{H}^{\left(
3+5\alpha _{3},\left( 1\mid \alpha _{1}+1\right) _{5},\left( 2\mid \alpha
_{2}+1\right) _{5},\left( 3\mid \alpha _{3}\right) _{5},\left( 4\mid \alpha
_{4}\right) _{5}\right) }\left( z\right) }\right) \right)  \notag \\
&=&\omega x/2+\sqrt{\frac{\omega }{2}}\frac{d}{dz}\left( \log \left( \frac{%
\mathcal{H}^{\left( \left( 1\mid \alpha _{1}+1\right) _{5},\left( 2\mid
\alpha _{2}+1\right) _{5},\left( 3\mid \alpha _{3}\right) _{5},\left( 4\mid
\alpha _{4}\right) _{5}\right) }\left( z\right) }{\mathcal{H}^{\left( \left(
1\mid \alpha _{1}+1\right) _{5},\left( 2\mid \alpha _{2}+1\right)
_{5},\left( 3\mid \alpha _{3}+1\right) _{5},\left( 4\mid \alpha _{4}\right)
_{5}\right) }\left( z\right) }\right) \right) ,
\end{eqnarray}

\begin{eqnarray}
w_{4}(x) &=&w_{4+5\alpha _{4}}^{\left( \left( 1\mid \alpha _{1}+1\right)
_{5},\left( 2\mid \alpha _{2}+1\right) _{5},\left( 3\mid \alpha
_{3}+1\right) _{5},\left( 4\mid \alpha _{4}\right) _{5}\right) }(x)  \notag
\\
&=&\omega x/2+\sqrt{\frac{\omega }{2}}\frac{d}{dz}\left( \log \left( \frac{%
\mathcal{H}^{\left( \left( 1\mid \alpha _{1}+1\right) _{5},\left( 2\mid
\alpha _{2}+1\right) _{5},\left( 3\mid \alpha _{3}+1\right) _{5},\left(
4\mid \alpha _{4}\right) _{5}\right) }\left( z\right) }{\mathcal{H}^{\left(
4+5\alpha _{4},\left( 1\mid \alpha _{1}+1\right) _{5},\left( 2\mid \alpha
_{2}+1\right) _{5},\left( 3\mid \alpha _{3}+1\right) _{5},\left( 4\mid
\alpha _{4}\right) _{5}\right) }\left( z\right) }\right) \right)  \notag \\
&=&\omega x/2+\sqrt{\frac{\omega }{2}}\frac{d}{dz}\left( \log \left( \frac{%
\mathcal{H}^{\left( \left( 1\mid \alpha _{1}+1\right) _{5},\left( 2\mid
\alpha _{2}+1\right) _{5},\left( 3\mid \alpha _{3}+1\right) _{5},\left(
4\mid \alpha _{4}\right) _{5}\right) }\left( z\right) }{\mathcal{H}^{\left(
\left( 1\mid \alpha _{1}+1\right) _{5},\left( 2\mid \alpha _{2}+1\right)
_{5},\left( 3\mid \alpha _{3}+1\right) _{5},\left( 4\mid \alpha
_{4}+1\right) _{5}\right) }\left( z\right) }\right) \right) ,
\end{eqnarray}%
and (see Eq(\ref{Eqext5}))

\begin{eqnarray}
w_{5}(x) &=&w_{0}^{\left( \left( 1\mid \alpha _{1}+1\right) _{5},\left(
2\mid \alpha _{2}+1\right) _{5},\left( 3\mid \alpha _{3}+1\right)
_{5},\left( 4\mid \alpha _{4}+1\right) _{5}\right) }(x)  \notag \\
&=&\omega x/2+\sqrt{\frac{\omega }{2}}\frac{d}{dz}\left( \log \left( \frac{%
\mathcal{H}^{\left( \left( 1\mid \alpha _{1}+1\right) _{5},\left( 2\mid
\alpha _{2}+1\right) _{5},\left( 3\mid \alpha _{3}+1\right) _{5},\left(
4\mid \alpha _{4}+1\right) _{5}\right) }\left( z\right) }{\mathcal{H}%
^{\left( 0,\left( 1\mid \alpha _{1}+1\right) _{5},\left( 2\mid \alpha
_{2}+1\right) _{5},\left( 3\mid \alpha _{3}+1\right) _{5},\left( 4\mid
\alpha _{4}+1\right) _{5}\right) }\left( z\right) }\right) \right)  \notag \\
&=&\omega x/2+\sqrt{\frac{\omega }{2}}\frac{d}{dz}\left( \log \left( \frac{%
\mathcal{H}^{\left( \left( 1\mid \alpha _{1}+1\right) _{5},\left( 2\mid
\alpha _{2}+1\right) _{5},\left( 3\mid \alpha _{3}+1\right) _{5},\left(
4\mid \alpha _{4}+1\right) _{5}\right) }\left( z\right) }{\mathcal{H}%
^{\left( \left( 1\mid \alpha _{1}\right) _{5},\left( 2\mid \alpha
_{2}\right) _{5},\left( 3\mid \alpha _{3}\right) _{5},\left( 4\mid \alpha
_{4}\right) _{5}\right) }\left( z\right) }\right) \right) ,
\end{eqnarray}%
where we have used

\begin{eqnarray}
&&\left( 0,\left( 1\mid \alpha _{1}+1\right) _{5},\left( 2\mid \alpha
_{2}+1\right) _{5},\left( 3\mid \alpha _{3}+1\right) _{5},\left( 4\mid
\alpha _{4}+1\right) _{5}\right)  \notag \\
&=&\left( 0,1,2,3,4\right) \cup \left( \left( 6\mid \alpha _{1}+1\right)
_{5},\left( 7\mid \alpha _{2}+1\right) _{5},\left( 8\mid \alpha
_{3}+1\right) _{5},\left( 9\mid \alpha _{4}+1\right) _{5}\right)
=N_{m}\oplus 5.
\end{eqnarray}

\subsubsection{k=3}

The $5$-cyclic Maya diagram with $k=3$ is of the form

\begin{equation}
N_{m}=\left( \left( 1\mid \alpha _{1}\right) _{3},\left( 2\mid \alpha
_{2}\right) _{3},\left( \lambda _{1}\mid \mu _{1}\right) _{3}\right) ,\
m=\alpha _{1}+\alpha _{2}+\mu _{1}.
\end{equation}%
with the convention that a two times repeated index is suppressed from the
list.

The corresponding $5$-cyclic chain is obtained by permutations from

\begin{equation}
\left\{ 1+3\alpha _{1},2+3\alpha _{2},\lambda _{1},\lambda _{1}+\mu
_{1},0\right\}
\end{equation}%
and for the chain $\left( 1+3\alpha _{1},2+3\alpha _{2},\lambda _{1},\lambda
_{1}+\mu _{1},0\right) $ the parameters in the $5$-cyclic dressing chain
system (see Eq(\ref{pcyclicdress})) are, with the order above

\begin{equation}
\left\{ 
\begin{array}{c}
\varepsilon _{12}=\left( -1-3\left( \alpha _{2}-\alpha _{1}\right) \right)
\omega \\ 
\varepsilon _{23}=\left( 2+3\alpha _{2}-\lambda _{1}\right) \omega \\ 
\varepsilon _{34}=-\mu _{1}\omega \\ 
\varepsilon _{45}=\left( \lambda _{1}+\mu _{1}-3\right) \omega \\ 
\varepsilon _{51}-\Delta =\left( -4-3\alpha _{1}\right) \omega .%
\end{array}%
\right.
\end{equation}

The solutions of the dressing chain system are given by (see Eq(\ref%
{soldress}))

\begin{eqnarray}
w_{1}(x) &=&w_{1+3\alpha _{1}}^{\left( \left( 1\mid \alpha _{1}\right)
_{3},\left( 2\mid \alpha _{2}\right) _{3},\left( \lambda _{1}\mid \mu
_{1}\right) _{3}\right) }(x)  \notag \\
&=&\omega x/2+\sqrt{\frac{\omega }{2}}\frac{d}{dz}\left( \log \left( \frac{%
\mathcal{H}^{\left( \left( 1\mid \alpha _{1}\right) _{3},\left( 2\mid \alpha
_{2}\right) _{3},\left( \lambda _{1}\mid \mu _{1}\right) _{3}\right) }\left(
z\right) }{\mathcal{H}^{\left( 1+3\alpha _{1},\left( 1\mid \alpha
_{1}\right) _{3},\left( 2\mid \alpha _{2}\right) _{3},\left( \lambda
_{1}\mid \mu _{1}\right) _{3}\right) }\left( z\right) }\right) \right) 
\notag \\
&=&\omega x/2+\sqrt{\frac{\omega }{2}}\frac{d}{dz}\left( \log \left( \frac{%
\mathcal{H}^{\left( \left( 1\mid \alpha _{1}\right) _{5},\left( 2\mid \alpha
_{2}\right) _{5},\left( 3\mid \alpha _{3}\right) _{5},\left( 4\mid \alpha
_{4}\right) _{5}\right) }\left( z\right) }{\mathcal{H}^{\left( \left( 1\mid
\alpha _{1}+1\right) _{3},\left( 2\mid \alpha _{2}\right) _{3},\left(
\lambda _{1}\mid \mu _{1}\right) _{3}\right) }\left( z\right) }\right)
\right) ,
\end{eqnarray}

\begin{eqnarray}
w_{2}(x) &=&w_{2+3\alpha _{2}}^{\left( \left( 1\mid \alpha _{1}+1\right)
_{3},\left( 2\mid \alpha _{2}\right) _{3},\left( \lambda _{1}\mid \mu
_{1}\right) _{3}\right) }(x)  \notag \\
&=&\omega x/2+\sqrt{\frac{\omega }{2}}\frac{d}{dz}\left( \log \left( \frac{%
\mathcal{H}^{\left( \left( 1\mid \alpha _{1}+1\right) _{3},\left( 2\mid
\alpha _{2}\right) _{3},\left( \lambda _{1}\mid \mu _{1}\right) _{3}\right)
}\left( z\right) }{\mathcal{H}^{\left( 2+3\alpha _{2},\left( 1\mid \alpha
_{1}+1\right) _{3},\left( 2\mid \alpha _{2}\right) _{3},\left( \lambda
_{1}\mid \mu _{1}\right) _{3}\right) }\left( z\right) }\right) \right) 
\notag \\
&=&\omega x/2+\sqrt{\frac{\omega }{2}}\frac{d}{dz}\left( \log \left( \frac{%
\mathcal{H}^{\left( \left( 1\mid \alpha _{1}+1\right) _{3},\left( 2\mid
\alpha _{2}\right) _{3},\left( \lambda _{1}\mid \mu _{1}\right) _{3}\right)
}\left( z\right) }{\mathcal{H}^{\left( \left( 1\mid \alpha _{1}+1\right)
_{3},\left( 2\mid \alpha _{2}+1\right) _{3},\left( \lambda _{1}\mid \mu
_{1}\right) _{3}\right) }\left( z\right) }\right) \right) ,
\end{eqnarray}

\begin{eqnarray}
w_{3}(x) &=&w_{\lambda _{1}}^{\left( \left( 1\mid \alpha _{1}+1\right)
_{3},\left( 2\mid \alpha _{2}+1\right) _{3},\left( \lambda _{1}\mid \mu
_{1}\right) _{3}\right) }(x)  \notag \\
&=&-\omega x/2+\sqrt{\frac{\omega }{2}}\frac{d}{dz}\left( \log \left( \frac{%
\mathcal{H}^{\left( \left( 1\mid \alpha _{1}+1\right) _{3},\left( 2\mid
\alpha _{2}+1\right) _{3},\left( \lambda _{1}\mid \mu _{1}\right)
_{3}\right) }\left( z\right) }{\mathcal{H}^{\left( \lambda _{1},\left( 1\mid
\alpha _{1}+1\right) _{3},\left( 2\mid \alpha _{2}+1\right) _{3},\left(
\lambda _{1}\mid \mu _{1}\right) _{3}\right) }\left( z\right) }\right)
\right)  \notag \\
&=&-\omega x/2+\sqrt{\frac{\omega }{2}}\frac{d}{dz}\left( \log \left( \frac{%
\mathcal{H}^{\left( \left( 1\mid \alpha _{1}+1\right) _{3},\left( 2\mid
\alpha _{2}+1\right) _{3},\left( \lambda _{1}\mid \mu _{1}\right)
_{3}\right) }\left( z\right) }{\mathcal{H}^{\left( \left( 1\mid \alpha
_{1}+1\right) _{3},\left( 2\mid \alpha _{2}+1\right) _{3},\left( \lambda
_{1}+3\mid \mu _{1}-1\right) _{3}\right) }\left( z\right) }\right) \right) ,
\end{eqnarray}

\begin{eqnarray}
w_{4}(x) &=&w_{\lambda _{1}+\mu _{1}}^{\left( \left( 1\mid \alpha
_{1}+1\right) _{3},\left( 2\mid \alpha _{2}+1\right) _{3},\left( \lambda
_{1}+3\mid \mu _{1}-1\right) _{3}\right) }(x)  \notag \\
&=&\omega x/2+\sqrt{\frac{\omega }{2}}\frac{d}{dz}\left( \log \left( \frac{%
\mathcal{H}^{\left( \left( 1\mid \alpha _{1}+1\right) _{3},\left( 2\mid
\alpha _{2}+1\right) _{3},\left( \lambda _{1}+3\mid \mu _{1}-1\right)
_{3}\right) }\left( z\right) }{\mathcal{H}^{\left( \lambda _{1}+\mu
_{1},\left( 1\mid \alpha _{1}+1\right) _{3},\left( 2\mid \alpha
_{2}+1\right) _{3},\left( \lambda _{1}+3\mid \mu _{1}-1\right) _{3}\right)
}\left( z\right) }\right) \right)  \notag \\
&=&\omega x/2+\sqrt{\frac{\omega }{2}}\frac{d}{dz}\left( \log \left( \frac{%
\mathcal{H}^{\left( \left( 1\mid \alpha _{1}+1\right) _{3},\left( 2\mid
\alpha _{2}+1\right) _{3},\left( \lambda _{1}+3\mid \mu _{1}-1\right)
_{3}\right) }\left( z\right) }{\mathcal{H}^{\left( \left( 1\mid \alpha
_{1}+1\right) _{3},\left( 2\mid \alpha _{2}+1\right) _{3},\left( \lambda
_{1}+3\mid \mu _{1}\right) _{3}\right) }\left( z\right) }\right) \right) ,
\end{eqnarray}%
and (see Eq(\ref{Eqext5}))

\begin{eqnarray}
w_{5}(x) &=&w_{0}^{\left( \left( 1\mid \alpha _{1}+1\right) _{3},\left(
2\mid \alpha _{2}+1\right) _{3},\left( \lambda _{1}+3\mid \mu _{1}\right)
_{3}\right) }(x)  \notag \\
&=&\omega x/2+\sqrt{\frac{\omega }{2}}\frac{d}{dz}\left( \log \left( \frac{%
\mathcal{H}^{\left( \left( 1\mid \alpha _{1}+1\right) _{3},\left( 2\mid
\alpha _{2}+1\right) _{3},\left( \lambda _{1}+3\mid \mu _{1}\right)
_{3}\right) }\left( z\right) }{\mathcal{H}^{\left( 0,\left( 1\mid \alpha
_{1}+1\right) _{3},\left( 2\mid \alpha _{2}+1\right) _{3},\left( \lambda
_{1}+3\mid \mu _{1}\right) _{3}\right) }\left( z\right) }\right) \right) 
\notag \\
&=&\omega x/2+\sqrt{\frac{\omega }{2}}\frac{d}{dz}\left( \log \left( \frac{%
\mathcal{H}^{\left( \left( 1\mid \alpha _{1}+1\right) _{3},\left( 2\mid
\alpha _{2}+1\right) _{3},\left( \lambda _{1}+3\mid \mu _{1}\right)
_{3}\right) }\left( z\right) }{\mathcal{H}^{\left( \left( 1\mid \alpha
_{1}\right) _{3},\left( 2\mid \alpha _{2}\right) _{3},\left( \lambda
_{1}\mid \mu _{1}\right) _{3}\right) }\left( z\right) }\right) \right) ,
\end{eqnarray}%
where we have used:

\begin{equation}
\begin{array}{c}
\left( 0,\left( 1\mid \alpha _{1}+1\right) _{3},\left( 2\mid \alpha
_{2}+1\right) _{3},\left( \lambda _{1}+3\mid \mu _{1}\right) _{3}\right) \\ 
=\left( 0,1,2\right) \cup \left( \left( 4\mid \alpha _{1}+1\right)
_{3},\left( 5\mid \alpha _{2}+1\right) _{3},\left( \lambda _{1}+3\mid \mu
_{1}\right) _{3}\right) =N_{m}\oplus 3.%
\end{array}%
\end{equation}

\section{Rational extensions of the IO and rational solutions of the
dressing chains of even periodicity}

\subsection{Rational extensions of the IO}

The IO potential (with zero ground level $E_{0}=0$)) is defined on the
positive half line $\left] 0,+\infty \right[ $ by

\begin{equation}
V\left( x;\omega ,\alpha \right) =\frac{\omega ^{2}}{4}x^{2}+\frac{\left(
\alpha +1/2\right) (\alpha -1/2)}{x^{2}}-\omega \left( \alpha +1\right)
,\quad \left\vert \alpha \right\vert >1/2.  \label{OI}
\end{equation}

If we add Dirichlet boundary conditions at $0$ and infinity and if we
suppose $\alpha >1/2$, it has the following spectrum ($z=\omega x^{2}/2$)%
\begin{equation}
\left\{ 
\begin{array}{c}
E_{n}\left( \omega \right) =2n\omega \\ 
\psi _{n\otimes \varnothing }\left( x;\omega ,\alpha \right) =\psi
_{0\otimes \varnothing }\left( x;\omega ,\alpha \right) \mathit{L}%
_{n}^{\alpha }\left( z\right)%
\end{array}%
\right. ,\quad n\geq 0,  \label{spec OI}
\end{equation}%
with $\psi _{0\otimes \varnothing }\left( x;\omega ,\alpha \right)
=z^{\left( \alpha +1/2\right) /2}e^{-z/2}$. The interest of this notation
for the spectral indices ($n\otimes \varnothing $ rather than $n$) will
become clear later.

$V\left( x;\omega ,\alpha \right) $ is translationally shape invariant with (%
$\alpha _{n}=\alpha +n$)

\begin{equation}
V^{\left( 0\otimes \varnothing \right) }\left( x;\omega ,\alpha \right)
=V\left( x;\omega ,\alpha _{1}\right) +2\omega .  \label{SI IO}
\end{equation}

It possesses also three discrete parametric symmetries:

*The $\Gamma _{1}$\textbf{\ symmetry, }$\left( \omega ,\alpha \right) 
\overset{\Gamma _{1}}{\rightarrow }\left( -\omega ,\alpha \right) $, which
acts as

\begin{equation}
V(x;\omega ,\alpha )\overset{\Gamma _{1}}{\rightarrow }V(x;-\omega ,\alpha
)=V(x;\omega ,\alpha )+2\omega \left( \alpha +1\right) ,
\end{equation}%
and generates the \textbf{conjugate shadow spectrum} of $V\left( x;\omega
,\alpha \right) $ :

\begin{equation}
\left\{ 
\begin{array}{c}
E_{n}\left( \omega \right) \\ 
\psi _{n\otimes \varnothing }(x;\omega ,\alpha )%
\end{array}%
\right. \overset{\Gamma _{1}}{\rightarrow }\left\{ 
\begin{array}{c}
E_{\left( -n-1\right) -\alpha }\left( \omega \right) =-2\left( n+1+\alpha
\right) \omega <0 \\ 
\psi _{\varnothing \otimes \left( -n-1\right) }(x;\omega ,\alpha )=z^{\left(
\alpha +1/2\right) /2}e^{z/2}\mathit{L}_{n}^{\alpha }\left( -z\right)%
\end{array}%
\right. ,\quad n\geq 0.  \label{conjshadOI}
\end{equation}

*The $\Gamma _{2}$\textbf{\ symmetry, }$\left( \omega ,\alpha \right) 
\overset{\Gamma _{2}}{\rightarrow }\left( \omega ,-\alpha \right) $, which
acts as

\begin{equation}
V(x;\omega ,\alpha )\overset{\Gamma _{2}}{\rightarrow }V(x;\omega ,\alpha
)+2\omega \alpha ,
\end{equation}%
and generates the \textbf{shadow spectrum} of $V\left( x;\omega ,\alpha
\right) :$

\begin{equation}
\left\{ 
\begin{array}{c}
E_{n}\left( \omega \right) \\ 
\psi _{n\otimes \varnothing }(x;\omega ,\alpha )%
\end{array}%
\right. \overset{\Gamma _{2}}{\rightarrow }\left\{ 
\begin{array}{c}
E_{n-\alpha }\left( \omega \right) =2\left( n-\alpha \right) \omega \\ 
\psi _{\varnothing \otimes n}(x;\omega ,\alpha )=z^{\left( -\alpha
+1/2\right) /2}e^{-z/2}\mathit{L}_{n}^{-\alpha }\left( z\right)%
\end{array}%
\right. ,\quad n\geq 0.  \label{shadOI}
\end{equation}

Note that

\begin{equation}
\psi _{\varnothing \otimes 0}(x;\omega ,\alpha )=\psi _{0}(x;\omega ,\alpha
)z^{-\alpha }
\end{equation}

* The $\Gamma _{3}=\Gamma _{1}\circ \Gamma _{2}$\textbf{\ symmetry,} $\left(
\omega ,\alpha \right) \overset{\Gamma _{3}}{\rightarrow }\left( -\omega
,-\alpha \right) $, which acts as

\begin{equation}
V(x;\omega ,\alpha )\overset{\Gamma _{3}}{\rightarrow }V(x;-\omega ,-\alpha
)=V(x;\omega ,\alpha )+2\omega ,
\end{equation}%
and generates the \textbf{conjugate spectrum} of $V\left( x;\omega ,\alpha
\right) $ :

\begin{equation}
\left\{ 
\begin{array}{c}
E_{n}\left( \omega \right) \\ 
\psi _{n\otimes \varnothing }(x;\omega ,\alpha )%
\end{array}%
\right. \overset{\Gamma _{3}}{\rightarrow }\left\{ 
\begin{array}{c}
E_{-n-1}\left( \omega \right) =-2\left( n+1\right) \omega <0 \\ 
\psi _{\left( -n-1\right) \otimes \varnothing }(x;\omega ,\alpha )=z^{\left(
-\alpha +1/2\right) /2}e^{z/2}\mathit{L}_{n}^{-\alpha }\left( -z\right)%
\end{array}%
\right. ,\quad n\geq 0.  \label{conjOi}
\end{equation}

The union of the spectrum and the conjugate spectrum forms the \textbf{%
extended spectrum} and the union of the shadow and conjugate shadow spectra
forms the \textbf{extended shadow spectrum}. All together they contain all
the quasi-polynomial eigenfunctions of the IO. To avoid the specific case
where the extended and extended shadow spectra merge, we restrict the values
of $\alpha $ to be non-integer:

\begin{equation}
\alpha \notin 
\mathbb{N}
.  \label{alphaconst}
\end{equation}

The rational extensions of the IO, which\ are obtained via chains of\ DT
associated to seed functions of this type, can then be indexed by a couple
of Maya diagrams, that we call a \textbf{universal character (UC) }\cite%
{koike,tsuda,tsuda2,tsuda3}

\begin{equation}
N_{m}\otimes L_{r}=\left( n_{1},...,n_{m}\right) \otimes \left(
l_{1},...,l_{r}\right) ,  \label{bi-tuple}
\end{equation}%
$N_{m}=\left( n_{1},...,n_{m}\right) $ containing the spectral indices of
the seed functions belonging to the extended spectrum and $L_{r}=\left(
l_{1},...,l_{r}\right) $ those belonging to the extended shadow spectrum.
The UC is said to be \textbf{canonical} if $N_{m}$\ and $L_{r}$ are
canonical Maya diagrams. If $N_{m}$\ and $L_{r}$ are respectively $p_{1}$%
-cyclic with translation of $k_{1}>0$ and $p_{2}$-cyclic with translation of 
$k_{2}>0$ then we say that is the UC $N_{m}\otimes L_{r}$ is $p$\textbf{%
-cyclic, }$p=p_{1}+p_{2}$.\textbf{\ }For reasons which will become clear
later, we call the quantity $k=k_{1}-k_{2}$, the \textbf{balanced
translation amplitude} of $N_{m}\otimes L_{r}$.

We have proven in \cite{GGM}, that if we have two \textbf{equivalent UC}%
\begin{equation}
N_{m}\otimes L_{r}\approx N_{m^{\prime }}^{\prime }\otimes L_{r^{\prime
}}^{\prime },  \label{eqUC}
\end{equation}%
that is, if 
\begin{equation}
N_{m}\approx N_{m^{\prime }}^{\prime }\text{ and }L_{r}\approx L_{r^{\prime
}}^{\prime },
\end{equation}
then

\begin{equation}
V^{N_{m}\otimes L_{r}}(x;\omega ,\alpha )\ =V^{N_{m^{\prime }}^{\prime
}\otimes L_{r^{\prime }}^{\prime }}(x;\omega ,\alpha _{s})+2q\omega ,\
q,s\in 
\mathbb{Z}
.  \label{EqextIO1}
\end{equation}

In particular, if $N_{m}$ and $L_{r}$ are canonical and $k_{1},k_{2}>0$

\begin{equation}
V^{\left( N_{m}\oplus k_{1}\right) \otimes \left( L_{r}\oplus k_{2}\right)
}(x;\omega ,\alpha )\ =V^{N_{m}\otimes L_{r}}(x;\omega ,\alpha
_{k_{1}-k_{2}})+2k_{1}\omega .  \label{EqextIO2}
\end{equation}

For the Wronskians, we have the following equivalence relation \cite{GGM2}

\begin{eqnarray}
W^{\left( N_{m}\oplus k_{1}\right) \otimes \left( L_{r}\oplus k_{2}\right)
}(x;\omega ,\alpha )\ &=&\prod\limits_{i=0}^{k_{1}-1}\left( \psi _{0\otimes
\varnothing }(x;\omega ,\alpha _{i})\right)
\prod\limits_{j=0}^{k_{2}-1}\left( \psi _{\varnothing \otimes 0}(x;\omega
,\alpha _{k_{1}-j})\right) W^{N_{m}\otimes L_{r}}(x;\omega ,\alpha
_{k_{1}-k_{2}})  \label{EqextIO3} \\
&=&z^{\alpha (k_{1}-k_{2})/2+\left( k_{1}-k_{2}\right) ^{2}/4}\times
e^{-(k_{1}+k_{2})z/2}\times W^{N_{m}\otimes L_{r}}(x;\omega ,\alpha
_{k_{1}-k_{2}}),  \notag
\end{eqnarray}%
or, if $k_{1}\leq n_{1}\leq ...\leq n_{m}$ and $k_{2}\leq l_{1}\leq ...\leq
l_{r}$

\begin{eqnarray}
&&W^{\left( \left( 0,...,k_{1}-1\right) \cup N_{m}\right) \otimes \left(
\left( 0,...,k_{2}-1\right) \cup L_{r}\right) }(x;\omega ,\alpha )\ 
\label{EqextIO4} \\
&=&\prod\limits_{i=0}^{k_{1}-1}\left( \psi _{0\otimes \varnothing
}(x;\omega ,\alpha _{i})\right) \prod\limits_{j=0}^{k_{2}-1}\left( \psi
_{\varnothing \otimes 0}(x;\omega ,\alpha _{k_{1}-j})\right) W^{\left(
N_{m}-k_{1}\right) \otimes \left( L_{r}-k_{2}\right) }(x;\omega ,\alpha
_{k_{1}-k_{2}})  \notag \\
&=&z^{\alpha (k_{1}-k_{2})/2+\left( k_{1}-k_{2}\right) ^{2}/4}\times
e^{-(k_{1}+k_{2})z/2}W^{\left( N_{m}-k_{1}\right) \otimes \left(
L_{r}-k_{2}\right) }(x;\omega ,\alpha _{k_{1}-k_{2}}).  \notag
\end{eqnarray}

Using Eq(\ref{spec OI}), Eq(\ref{shadOI}), Eq(\ref{wronskprop}) and the
derivation properties of the Laguerre \cite{magnus,szego,GGM2} (see also Eq(%
\ref{poch}))

\begin{equation}
\left\{ 
\begin{array}{c}
\frac{d^{j}}{dz^{j}}\left( L_{n}^{\alpha }\left( z\right) \right) =\left(
-1\right) ^{j}L_{n-j}^{\alpha +j}\left( z\right) \\ 
\frac{d^{j}}{dz^{j}}\left( z^{-\alpha }L_{n}^{-\alpha }\left( z\right)
\right) =\left( n-\alpha \right) _{j}L_{n}^{-\alpha -j}\left( z\right) ,%
\end{array}%
\right.
\end{equation}%
we can write

\begin{equation}
W^{N_{m}\otimes L_{r}}(x;\omega ,\alpha )\ \propto z^{(m-r)^{2}/4-r\left(
r-1\right) }\times z^{\alpha (m-r)/2}\times e^{-\left( m+r\right) z/2}\times 
\mathcal{L}^{N_{m}\otimes L_{r}}\left( z;\alpha \right) ,  \label{WL2}
\end{equation}%
where $\mathcal{L}^{N_{m}\otimes L_{r}}$ is the following determinant ($%
i=0,...,m+r-1$)%
\begin{equation}
\mathcal{L}^{N_{m}\otimes L_{r}}\left( z;\alpha \right) =\left\vert 
\overrightarrow{L}_{n_{1}}\left( z;\alpha \right) ,...,\overrightarrow{L}%
_{n_{m}}\left( z;\alpha \right) ,\overrightarrow{\Lambda }_{l_{1}}\left(
z;\alpha \right) ,...,\overrightarrow{\Lambda }_{l_{r}}\left( z;\alpha
\right) \right\vert ,  \label{PWL}
\end{equation}%
with ($L_{n}^{\alpha }\left( z\right) =0$ if $n<0$)

\begin{equation}
\overrightarrow{L}_{n}\left( z;\alpha \right) =\left( 
\begin{array}{c}
L_{n}^{\alpha }\left( z\right) \\ 
... \\ 
\left( -1\right) ^{i}L_{n-i}^{\alpha +i}\left( z\right) \\ 
... \\ 
\left( -1\right) ^{m+r-1}L_{n-m-r+1}^{\alpha +m+r-1}\left( z\right)%
\end{array}%
\right) ,\ \overrightarrow{\Lambda }_{l}\left( z;\alpha \right) =\left( 
\begin{array}{c}
z^{m+r-1}L_{l}^{-\alpha }\left( z\right) \\ 
... \\ 
\left( l-\alpha \right) _{i}z^{m+r-1-i}L_{l}^{-\alpha -j}\left( z\right) \\ 
... \\ 
\left( l-\alpha \right) _{m+r-1}L_{l}^{-\alpha -m-r+1}\left( z\right)%
\end{array}%
\right) .
\end{equation}

$\mathcal{L}^{N_{m}\otimes L_{r}}$ is called a \textbf{Laguerre
pseudowronskian} \cite{GGM,GGM2}.

Eq(\ref{EqextIO3}) and Eq(\ref{EqextIO4})\ give then

\begin{equation}
\left\{ 
\begin{array}{c}
\mathcal{L}^{\left( N_{m}\oplus k_{1}\right) \otimes \left( L_{r}\oplus
k_{2}\right) }(z;\alpha )\ \propto z^{2rk_{2}+k_{2}\left( k_{2}-1\right) }%
\mathcal{L}^{\left( N_{m}\right) \otimes \left( L_{r}\right) }(z;\alpha
_{k_{1}-k_{2}}) \\ 
\mathcal{L}^{\left( \left( 0,...,k_{1}-1\right) \cup N_{m}\right) \otimes
\left( \left( 0,...,k_{2}-1\right) \cup L_{r}\right) }(z;\alpha )\ \propto
z^{2rk_{2}+k_{2}\left( k_{2}-1\right) }\mathcal{L}^{\left(
N_{m}-k_{1}\right) \otimes \left( L_{r}-k_{2}\right) }(z;\alpha
_{k_{1}-k_{2}}).%
\end{array}%
\right.  \label{EqextIO5}
\end{equation}

The equivalence property Eq(\ref{EqextIO1}) can be viewed as the most
general transcription of the shape invariance property Eq(\ref{SI IO}) at
the level of the rational extensions of the IO potential \cite{GGM2}.

Due to this equivalence property, to describe all the rational extensions of
the IO, it is sufficient to consider those associated to canonical UC.

We can also\ note the useful symmmetry relation

\begin{equation}
W^{L_{r}\otimes N_{m}}(x;\omega ,\alpha )=\Gamma _{2}\left( W^{N_{m}\otimes
L_{r}}(x;\omega ,\alpha )\right) ,  \label{sym1}
\end{equation}%
which leads immediately to

\begin{equation}
\ V^{L_{r}\otimes N_{m}}(x;\omega ,\alpha )=\Gamma _{2}\left(
V^{N_{m}\otimes L_{r}}(x;\omega ,\alpha )\right) +2\omega \alpha .
\label{sym2}
\end{equation}

\subsection{$p$-cyclic extensions of the IO and rational solutions of the
even periodic dressing chains}

The cyclicity of the canonical UC $N_{m}\otimes L_{r}$ alone is not
sufficient to ensure the cyclicity of the associated rational extension of
the IO $V^{N_{m}\otimes L_{r}}(x;\omega ,\alpha )$. Nevertheless we have the
following lemma

\begin{lemma}
If the UC $N_{m}\otimes L_{r}$ is $p$-cyclic with a zero balanced
translation amplitude, then $V^{N_{m}\otimes L_{r}}(x;\omega ,\alpha )$ is a 
$p$-cyclic potential.
\end{lemma}

\begin{proof}
If $N_{m}\otimes L_{r}$ is $p$-cyclic, there exists a chain of DT $\left(
\nu _{1},...,\nu _{p_{1}}\right) \otimes \left( \lambda _{1},...,\lambda
_{p_{2}}\right) $ with $p=p_{1}+p_{2}$, and $k_{1},k_{2}\in 
\mathbb{N}
^{\ast }$ such that%
\begin{equation}
V^{\left( N_{m},\nu _{1},...,\nu _{p_{1}}\right) \otimes \left(
L_{r},\lambda _{1},...,\lambda _{p_{2}}\right) }(x;\omega ,\alpha
)=V^{\left( N_{m}\oplus k_{1}\right) \otimes \left( L_{r}\oplus k_{2}\right)
}(x;\omega ,\alpha ),
\end{equation}%
which, combined to Eq(\ref{EqextIO2}), leads to%
\begin{equation}
V^{\left( N_{m},\nu _{1},...,\nu _{p_{1}}\right) \otimes \left(
L_{r},\lambda _{1},...,\lambda _{p_{2}}\right) }(x;\omega ,\alpha
)=V^{N_{m}\otimes L_{r}}(x;\omega ,\alpha _{k_{1}-k_{2}})+2k_{1}\omega .
\end{equation}%
In the case of a zero balanced translation amplitude $k_{1}=k_{2}$ and $%
V^{N_{m}\otimes L_{r}}(x;\omega ,\alpha )$ is then $p$-cyclic with an energy
shift $\Delta =2k_{1}\omega $.
\end{proof}

Combining this lemma with Theorem 1, we arrive directly to the theorem

\begin{theorem}
\textit{The rational extensions of the IO }$V^{N_{m}\otimes L_{r}}(x;\omega
,\alpha )$\textit{\ with}%
\begin{equation}
\left\{ 
\begin{array}{c}
N_{m}=\left( \left( 1\mid a_{1}\right) _{k},...,\left( k-1\mid
a_{k-1}\right) _{k};\left( \lambda _{1}\mid \mu _{1}\right) _{k},...,\left(
\lambda _{j_{1}}\mid \mu _{j_{1}}\right) _{k}\right)  \\ 
L_{r}=\left( \left( 1\mid b_{1}\right) _{k},...,\left( k-1\mid
b_{k-1}\right) _{k};\left( \rho _{1}\mid \sigma _{1}\right) _{k},...,\left(
\rho _{j_{2}}\mid \sigma _{j_{2}}\right) _{k}\right) ,%
\end{array}%
\right.   \label{th4}
\end{equation}%
\textit{where }$a_{i},b_{i},\lambda _{i},\mu _{i},\rho _{i},\sigma _{i}$%
\textit{\ are arbitrary positive integers, solve the dressing chain of even
period }$p=p_{1}+p_{2}$\textit{\ with }$p_{l}=2j_{l}+k,\ l=1,2$\textit{, (}$%
p_{1},p_{2}$\textit{\ and }$k$\textit{\ have the same parity with }$0<k\leq
\min (p_{1},p_{2})$\textit{) for the following values of the parameters }%
\begin{equation}
\Delta =2k\omega \text{ and }\varepsilon _{i,i+1}=\left\{ 
\begin{array}{c}
2\left( \nu _{P(i)}-\nu _{P(i+1)}\right) \omega ,\text{ if }\nu _{P(i)},\nu
_{P(i+1)}\in \left\{ 1,...,p_{1}\right\} \text{ or }\nu _{P(i)},\nu
_{P(i+1)}\in \left\{ p_{1}+1,...,p\right\} , \\ 
2\left( \nu _{P(i)}-\nu _{P(i+1)}+\alpha \right) \omega ,\text{ if }\nu
_{P(i)}\in \left\{ 1,...,p_{1}\right\} \text{ and }\nu _{P(i+1)}\in \left\{
p_{1}+1,...,p\right\} , \\ 
2\left( \nu _{P(i)}-\nu _{P(i+1)}-\alpha \right) \omega ,\text{ if }\nu
_{P(i+1)}\in \left\{ 1,...,p_{1}\right\} \text{ or }\nu _{P(i)}\in \left\{
p_{1}+1,...,p\right\} ,%
\end{array}%
\right. \   \label{th41}
\end{equation}%
\textit{(}$\nu _{p+1}=\nu _{1}$\textit{) where }$P$\textit{\ is any
permutation of }$S_{p}$\textit{\ and}%
\begin{eqnarray}
\left( \nu _{1},...,\nu _{p_{1}}\right) \otimes \left( \nu
_{p_{1}+1},...,\nu _{p}\right)  &=&\left( 0,1+a_{1}k,...,\left( k-1\right)
+a_{k-1}k,\ \ \lambda _{1},\lambda _{1}+\mu _{1}k,...,\ \lambda
_{j_{1}},\lambda _{j_{1}}+\mu _{j_{1}}k\right)   \label{th42} \\
&&\otimes \left( 0,1+b_{1}k,...,\left( k-1\right) +b_{k-1}k,\ \ \rho
_{1},\rho _{1}+\sigma _{1}k,...,\ \rho _{j_{1}},\rho _{j_{1}}+\sigma
_{j_{1}}k\right) .  \notag
\end{eqnarray}%
\textit{As for the }$w$\textit{\ solutions of the dressing chain system, for
the chain above, they are given by}%
\begin{eqnarray}
w_{\nu _{i}\otimes \varnothing }^{\left( N_{m},\nu _{1},...,\nu
_{i-1}\right) \otimes L_{r}}(x;\omega ,\alpha ) &=&-\omega x/2+\frac{\alpha
-3/2+m-r+i}{x}  \label{th43} \\
&&+\omega x\frac{d}{dz}\left( \log \left( \frac{\mathcal{L}^{\left(
N_{m},\nu _{1},...,\nu _{i-1}\right) \otimes L_{r}}(z;\alpha )}{\mathcal{L}%
^{\left( N_{m},\nu _{1},...,\nu _{i}\right) \otimes L_{r}}(z;\alpha )}%
\right) \right) ,  \notag \\
\text{if }i &\leq &p_{1}\text{\ and the flip in }\nu _{i}\text{ is positive,}
\notag
\end{eqnarray}%
\begin{eqnarray}
w_{\nu _{i}\otimes \varnothing }^{\left( N_{m},\nu _{1},...,\nu
_{i-1}\right) \otimes L_{r}}(x;\omega ,\alpha ) &=&\omega x/2-\frac{\alpha
-1/2+m-r+i}{x}  \label{th44} \\
&&+\omega x\frac{d}{dz}\left( \log \left( \frac{\mathcal{L}^{\left(
N_{m},\nu _{1},...,\nu _{i-1}\right) \otimes L_{r}}(z;\alpha )}{\mathcal{L}%
^{\left( N_{m},\nu _{1},...,\nu _{i}\right) \otimes L_{r}}(z;\alpha )}%
\right) \right) ,  \notag \\
\text{ if }i &\leq &p_{1}\text{\ and the flip in }\nu _{i}\text{ is negative,%
}  \notag
\end{eqnarray}%
\begin{eqnarray}
w_{\varnothing \otimes \nu _{i}}^{\left( N_{m},\nu _{1},...,\nu
_{p_{1}}\right) \otimes \left( L_{r},\nu _{p_{1}+1},...,\nu _{i-1}\right)
}(x;\omega ,\alpha ) &=&-\omega x/2-\frac{\alpha -13/2+m+p_{1}+3\left(
r+i\right) }{x}  \label{th45} \\
&&+\omega x\frac{d}{dz}\left( \log \left( \frac{\mathcal{L}^{\left(
N_{m},\nu _{1},...,\nu _{p_{1}}\right) \otimes \left( L_{r},\nu
_{p_{1}+1},...,\nu _{i-1}\right) }(z;\alpha )}{\mathcal{L}^{\left( N_{m},\nu
_{1},...,\nu _{p_{1}}\right) \otimes \left( L_{r},\nu _{p_{1}+1},...,\nu
_{i}\right) }(z;\alpha )}\right) \right) ,  \notag \\
\text{if }i &=&p_{1}+j,\ j>0,\text{\ and the flip in }\nu _{i}\text{ is
positive,}  \notag
\end{eqnarray}%
\begin{eqnarray}
w_{\varnothing \otimes \nu _{i}}^{\left( N_{m},\nu _{1},...,\nu
_{p_{1}}\right) \otimes \left( L_{r},\nu _{p_{1}+1},...,\nu _{i-1}\right)
}(x;\omega ,\alpha ) &=&\omega x/2+\frac{\alpha -7/2+m+p_{1}+3\left(
r+i\right) }{x}  \label{th46} \\
&&+\omega x\frac{d}{dz}\left( \log \left( \frac{\mathcal{L}^{\left(
N_{m},\nu _{1},...,\nu _{p_{1}}\right) \otimes \left( L_{r},\nu
_{p_{1}+1},...,\nu _{i-1}\right) }(z;\alpha )}{\mathcal{L}^{\left( N_{m},\nu
_{1},...,\nu _{p_{1}}\right) \otimes \left( L_{r},\nu _{p_{1}+1},...,\nu
_{i}\right) }(z;\alpha )}\right) \right) ,  \notag \\
\text{if }i &=&p_{1}+j,\ j>0,\text{\ and the flip in }\nu _{i}\text{ is
negative,}  \notag
\end{eqnarray}%
\textit{with the convention that if a spectral index is repeated two times
in the tuples constituting the UC, then we suppress the corresponding
eigenfunction in the Laguerre pseudowronskians }$\mathcal{L}$\textit{.}
\end{theorem}

\begin{proof}
Note that due to the relations Eq(\ref{sym1}) and Eq(\ref{sym2}), we can
without loss of generality restrict the study of the solutions to the case $%
p_{1}\geq p_{2}.$\newline
If $i\leq p_{1}$\ and the flip in $\nu _{i}$ is negative, we have, using Eq(%
\ref{WL2})%
\begin{eqnarray}
w_{\nu \otimes \varnothing }^{N_{m}\otimes L_{r}}(x;\omega ,\alpha )
&=&-\left( \log \left( \frac{W^{\left( N_{m},\nu \right) \otimes
L_{r}}(x;\omega ,\alpha )}{W^{N_{m}\otimes L_{r}}(x;\omega ,\alpha )}\right)
\right) ^{\prime }  \notag \\
&=&-\left( \log \left( x^{\alpha +1/2+m-r}\times e^{-\omega x^{2}/4}\times 
\frac{\mathcal{L}^{\left( N_{m},\nu \right) \otimes L_{r}}(z;\alpha )}{%
\mathcal{L}^{N_{m}\otimes L_{r}}(z;\alpha )}\right) \right) ^{\prime } 
\notag \\
&=&\omega x/2-\frac{\alpha -1/2+m-r}{x}+\omega x\frac{d}{dz}\left( \log
\left( \frac{\mathcal{L}^{N_{m}\otimes L_{r}}(z;\alpha )}{\mathcal{L}%
^{\left( N_{m},\nu \right) \otimes L_{r}}(z;\alpha )}\right) \right) .
\end{eqnarray}%
In the same manner, if $i\leq p_{1}$\ and the flip in $\nu _{i}$ is
positive, we have%
\begin{equation}
w_{\nu \otimes \varnothing }^{N_{m}\otimes L_{r}}(x;\omega ,\alpha )=-\omega
x/2+\frac{\alpha -1/2+m-r}{x}+\omega x\frac{d}{dz}\left( \log \left( \frac{%
\mathcal{L}^{N_{m}\otimes L_{r}}(z;\alpha )}{\mathcal{L}^{\left( N_{m},\nu
\right) \otimes L_{r}}(z;\alpha )}\right) \right) .
\end{equation}%
If $i>p_{1}$\ and the flip in $\nu _{i}$ is negative, we have, using Eq(\ref%
{WL2})%
\begin{eqnarray}
w_{\varnothing \otimes \nu }^{N_{m}\otimes L_{r}}(x;\omega ,\alpha )
&=&-\left( \log \left( \frac{W^{N_{m}\otimes \left( L_{r},\nu \right)
}(x;\omega ,\alpha )}{W^{N_{m}\otimes L_{r}}(x;\omega ,\alpha )}\right)
\right) ^{\prime }  \notag \\
&=&-\left( \log \left( x^{-\alpha +1/2-m-3r}\times e^{-\omega x^{2}/4}\times 
\frac{\mathcal{L}^{N_{m}\otimes \left( L_{r},\nu \right) }(z;\alpha )}{%
\mathcal{L}^{N_{m}\otimes L_{r}}(z;\alpha )}\right) \right) ^{\prime } 
\notag \\
&=&\omega x/2+\frac{\alpha -1/2+m+3r}{x}  \notag \\
&&+\omega x\frac{d}{dz}\left( \log \left( \frac{\mathcal{L}^{N_{m}\otimes
L_{r}}(z;\alpha )}{\mathcal{L}^{N_{m}\otimes \left( L_{r},\nu \right)
}(z;\alpha )}\right) \right) .
\end{eqnarray}%
In the same manner, if $i>p_{1}$\ and the flip in $\nu _{i}$ is positive, we
have%
\begin{eqnarray}
w_{\varnothing \otimes \nu }^{N_{m}\otimes L_{r}}(x;\omega ,\alpha )
&=&-\omega x/2-\frac{\alpha -7/2+m+3r}{x}  \notag \\
&&+\omega x\frac{d}{dz}\left( \log \left( \frac{\mathcal{L}^{\left(
N_{m},\nu _{1},...,\nu _{p_{1}}\right) \otimes \left( L_{r},\nu
_{p_{1}+1},...,\nu _{i-1}\right) }(z;\alpha )}{\mathcal{L}^{\left( N_{m},\nu
_{1},...,\nu _{p_{1}}\right) \otimes \left( L_{r},\nu _{p_{1}+1},...,\nu
_{i}\right) }(z;\alpha )}\right) \right) .
\end{eqnarray}
\end{proof}

\subsection{$\protect\bigskip $Examples}

\subsubsection{Dressing chain of period 2.}

We have $p=2$ which implies $p_{1}=p_{2}=k=1$ and $j_{1}=j_{2}=0$. It leads
to $N_{m}\otimes L_{r}=\varnothing \otimes \varnothing $ and the potential
which solves the dressing chain of period $2$ is the IO itself. The
corresponding $2$-cyclic chain is $\left( 0\right) \otimes \left( 0\right) $%
. These results are in agreement with 
 section I.

\subsubsection{4-cyclic extensions of the IO and rational solutions of PV}

For $p=4$, we can have $p_{1}=3$ and $p_{2}=1$ with $k=1$ ($j_{1}=1,j_{2}=0$%
) or $p_{1}=p_{2}=k=2$ ($j_{1}=j_{2}=0$). As shown before, the dressing
chain system is equivalent to the PV equation (see Eq(\ref{PV})) whose
rational solutions are associated to the Umemura polynomials \cite{clarkson4}%
.

\paragraph{$\left( p_{1},p_{2}\right) =\left( 3,1\right) $}

\textit{Theorem 4} gives then $N_{m}\otimes L_{r}=\left( \lambda \mid \mu
\right) _{1}\otimes \varnothing $ and a corresponding $4$-cyclic chain is
given by $\left( \lambda ,\lambda +\mu ,0\right) \otimes \left( 0\right) $.
The solutions of the dressing chain system of period $4$ with parameters

\begin{equation}
\left\{ 
\begin{array}{c}
\varepsilon _{12}=-2\mu \omega \\ 
\varepsilon _{23}=2\left( \lambda +\mu \right) \omega \\ 
\varepsilon _{34}=2\alpha \omega \\ 
\varepsilon _{41}-\Delta =2\left( -1-\lambda -\alpha \right) \omega ,%
\end{array}%
\right.
\end{equation}%
are (see Eq(\ref{th43}-\ref{th46}))

\begin{eqnarray}
w_{1}(x) &=&w_{\lambda \otimes \varnothing }^{\left( \lambda \mid \mu
\right) _{1}\otimes \varnothing }(x;\omega ,\alpha )  \notag \\
&=&-\omega x/2+\frac{\alpha +\mu -1/2}{x}+\omega x\frac{d}{dz}\left( \log
\left( \frac{\mathcal{L}^{\left( \lambda \mid \mu \right) _{1}\otimes
\varnothing }(z;\alpha )}{\mathcal{L}^{\left( \lambda +1\mid \mu -1\right)
_{1}\otimes \varnothing }(z;\alpha )}\right) \right)
\end{eqnarray}%
(the flip in $\lambda \otimes \varnothing $ is positive, $m=\mu ,\ r=0$),

\begin{eqnarray}
w_{2}(x) &=&w_{\lambda +\mu \otimes \varnothing }^{\left( \lambda ,\left(
\lambda \mid \mu \right) _{1}\right) \otimes \varnothing }(x;\omega ,\alpha
)=w_{\lambda +\mu \otimes \varnothing }^{\left( \left( \lambda +1\mid \mu
-1\right) _{1}\right) \otimes \varnothing }(x;\omega ,\alpha )  \notag \\
&=&\omega x/2-\frac{\alpha +\mu -1/2}{x}+\omega x\frac{d}{dz}\left( \log
\left( \frac{\mathcal{L}^{\left( \lambda +1\mid \mu -1\right) _{1}\otimes
\varnothing }(z;\alpha )}{\mathcal{L}^{\left( \lambda +1\mid \mu \right)
_{1}\otimes \varnothing }(z;\alpha )}\right) \right)
\end{eqnarray}%
(the flip in $\left( \lambda +\mu \right) \otimes \varnothing $ is negative, 
$m=\mu -1,\ r=0$),

\begin{eqnarray}
w_{3}(x) &=&w_{0\otimes \varnothing }^{\left( \lambda +\mu ,\lambda ,\left(
\lambda \mid \mu \right) _{1}\right) \otimes \varnothing }(x;\omega ,\alpha )%
\text{ }=w_{0\otimes \varnothing }^{\left( \left( \lambda +1\mid \mu \right)
_{1}\right) \otimes \varnothing }(x;\omega ,\alpha )  \notag \\
&=&\omega x/2-\frac{\alpha +\mu +1/2}{x}+\omega x\frac{d}{dz}\left( \log
\left( \frac{\mathcal{L}^{\left( \left( \lambda +1\mid \mu \right)
_{1}\right) \otimes \varnothing }(z;\alpha )}{\mathcal{L}^{\left( 0,\left(
\lambda +1\mid \mu \right) _{1}\right) \otimes \varnothing }(z;\alpha )}%
\right) \right)  \notag \\
&=&\omega x/2-\frac{\alpha +\mu +1/2}{x}+\omega x\frac{d}{dz}\left( \log
\left( \frac{\mathcal{L}^{\left( \left( \lambda +1\mid \mu \right)
_{1}\right) \otimes \varnothing }(z;\alpha )}{\mathcal{L}^{\left( \lambda
\mid \mu \right) _{1}\otimes \varnothing }(z;\alpha _{1})}\right) \right)
\end{eqnarray}%
(the flip in $0\otimes \varnothing $ is negative, $m=\mu ,\ r=0$) and (see
Eq(\ref{EqextIO5}))

\begin{eqnarray}
w_{4}(x) &=&w_{\varnothing \otimes 0}^{\left( 0,\lambda +\mu ,\lambda
,\left( \lambda \mid \mu \right) _{k}\right) \otimes \varnothing }(x;\omega
,\alpha )=w_{\varnothing \otimes 0}^{\left( 0,\left( \lambda +1\mid \mu
\right) _{k}\right) \otimes \varnothing }(x;\omega ,\alpha )  \notag \\
&=&\omega x/2+\frac{\alpha +\mu +1/2}{x}+\omega x\frac{d}{dz}\left( \log
\left( \frac{\mathcal{L}^{\left( 0,\left( \lambda +1\mid \mu \right)
_{k}\right) \otimes \varnothing }(z;\alpha )}{\mathcal{L}^{\left( 0,\left(
\lambda +1\mid \mu \right) _{k}\right) \otimes \left( 0\right) }(z;\alpha )}%
\right) \right)  \notag \\
&=&\omega x/2+\frac{\alpha +\mu +1/2}{x}+\omega x\frac{d}{dz}\left( \log
\left( \frac{\mathcal{L}^{\left( \left( \lambda \mid \mu \right) _{k}\right)
\otimes \varnothing }(z;\alpha _{1})}{\mathcal{L}^{\left( \left( \lambda
\mid \mu \right) _{k}\right) \otimes \varnothing }(z;\alpha )}\right)
\right) .
\end{eqnarray}%
(the flip in $\varnothing \otimes 0$ is negative, $m=\mu +1,\ r=0$), where
we have used

\begin{equation}
\left\{ 
\begin{array}{c}
\left( \lambda ,\left( \lambda \mid \mu \right) _{1}\right) =\left( \lambda
+1\mid \mu -1\right) _{1} \\ 
\left( \mu ,\left( \lambda +1\mid \mu -1\right) _{1}\right) =\left( \lambda
+1\mid \mu \right) _{1}.%
\end{array}%
\right.
\end{equation}

Taking $\omega =2$ ($t=z$), the corresponding solution of PV (see Eq(\ref{PV}%
)) with parameters (see Eq(\ref{paramPV}))%
\begin{equation}
\left\{ 
\begin{array}{c}
a=2\mu ^{2} \\ 
b=-2\alpha ^{2} \\ 
c=4(\alpha +2\lambda +\mu +1) \\ 
d=-1/2,%
\end{array}%
\right.
\end{equation}%
is then (see Eq(\ref{defPV}))

\begin{equation}
y(t)=1-1/\frac{d}{dt}\left( \log \left( \frac{\mathcal{L}^{\left( \lambda
\mid \mu \right) _{1}\otimes \varnothing }(t;\alpha )}{\mathcal{L}^{\left(
\lambda +1\mid \mu \right) _{1}\otimes \varnothing }(t;\alpha )}\right)
\right) .
\end{equation}

\paragraph{$\left( p_{1},p_{2}\right) =\left( 2,2\right) $}

\bigskip \textit{Theorem 4} gives $N_{m}\otimes L_{r}=\left( 1\mid
a_{1}\right) _{2}\otimes \left( 1\mid b_{1}\right) _{2}$ and the
corresponding $4$-cyclic chain is $\left( 1+2a_{1},0\right) \otimes \left(
1+2b_{1},0\right) $. The solutions of the dressing chain system of period $4$
with parameters

\begin{equation}
\left\{ 
\begin{array}{c}
\varepsilon _{12}=2\left( 1+2a_{1}\right) \omega \\ 
\varepsilon _{23}=2\left( \alpha -1-2b_{1}\right) \omega \\ 
\varepsilon _{34}=2\left( 1+2b_{1}\right) \omega \\ 
\varepsilon _{41}-\Delta =2\left( -3-2a_{1}-\alpha \right) \omega ,%
\end{array}%
\right.
\end{equation}%
are (see Eq(\ref{th43}-\ref{th46})) and Eq(\ref{EqextIO5}))

\begin{eqnarray}
w_{1}(x) &=&w_{\left( 1+2a_{1}\right) \otimes \varnothing }^{\left( 1\mid
a_{1}\right) _{2}\otimes \left( 1\mid b_{1}\right) _{2}}(x;\omega ,\alpha ) 
\notag \\
&=&\omega x/2-\frac{\alpha +1/2+a_{1}-b_{1}}{x}+\omega x\frac{d}{dz}\left(
\log \left( \frac{\mathcal{L}^{\left( 1\mid a_{1}\right) _{2}\otimes \left(
1\mid b_{1}\right) _{2}}(z;\alpha )}{\mathcal{L}^{\left( 1\mid
a_{1}+1\right) _{2}\otimes \left( 1\mid b_{1}\right) _{2}}(z;\alpha )}%
\right) \right)
\end{eqnarray}%
(the flip in $\left( 1+2a_{1}\right) \otimes \varnothing $ is negative, $%
m=a_{1},\ r=b_{1}$),

\begin{eqnarray}
w_{2}(x) &=&w_{0\otimes \varnothing }^{\left( 1\mid a_{1}+1\right)
_{2}\otimes \left( 1\mid b_{1}\right) _{2}}(x;\omega ,\alpha )  \notag \\
&=&\omega x/2-\frac{\alpha +3/2+a_{1}-b_{1}}{x}+\omega x\frac{d}{dz}\left(
\log \left( \frac{\mathcal{L}^{\left( 1\mid a_{1}+1\right) _{2}\otimes
\left( 1\mid b_{1}\right) _{2}}(z;\alpha )}{\mathcal{L}^{\left( \left( 1\mid
a_{1}\right) _{2}\oplus 2\right) \otimes \left( 1\mid b_{1}\right)
_{2}}(z;\alpha )}\right) \right)  \notag \\
&=&\omega x/2-\frac{\alpha +3/2+a_{1}-b_{1}}{x}+\omega x\frac{d}{dz}\left(
\log \left( \frac{\mathcal{L}^{\left( 1\mid a_{1}+1\right) _{2}\otimes
\left( 1\mid b_{1}\right) _{2}}(z;\alpha )}{\mathcal{L}^{\left( 1\mid
a_{1}\right) _{2}\otimes \left( 1\mid b_{1}\right) _{2}}(z;\alpha _{2})}%
\right) \right)
\end{eqnarray}%
(the flip in $0\otimes \varnothing $ is negative, $m=a_{1}+1,\ r=b_{1}$),

\begin{eqnarray}
w_{3}(x) &=&w_{\varnothing \otimes (1+2b_{1})}^{\left( 0,\left( 1\mid
a_{1}+1\right) \right) _{2}\otimes \left( 1\mid b_{1}\right) _{2}}(x;\omega
,\alpha )  \notag \\
&=&\omega x/2+\frac{\alpha +3/2+a_{1}+3b_{1}}{x}+\omega x\frac{d}{dz}\left(
\log \left( \frac{\mathcal{L}^{\left( \left( 1\mid a_{1}\right) _{2}\oplus
2\right) \otimes \left( 1\mid b_{1}\right) _{2}}(z;\alpha )}{\mathcal{L}%
^{\left( \left( 1\mid a_{1}\right) _{2}\oplus 2\right) \otimes \left( 1\mid
b_{1}+1\right) _{2}}(z;\alpha )}\right) \right)  \notag \\
&=&\omega x/2+\frac{\alpha +3/2+a_{1}+3b_{1}}{x}+\omega x\frac{d}{dz}\left(
\log \left( \frac{\mathcal{L}^{\left( 1\mid a_{1}\right) _{2}\otimes \left(
1\mid b_{1}\right) _{2}}(z;\alpha _{2})}{\mathcal{L}^{\left( 1\mid
a_{1}\right) _{2}\otimes \left( 1\mid b_{1}+1\right) _{2}}(z;\alpha _{2})}%
\right) \right)
\end{eqnarray}%
(the flip in $\varnothing \otimes (1+2b_{1})$ is negative, $m=a_{1}+2,\
r=b_{1}$) and

\begin{eqnarray}
w_{4}(x) &=&w_{\varnothing \otimes 0}^{\left( 1\mid a_{1}\right) _{2}\otimes
\left( 1\mid b_{1}+1\right) _{2}}(x;\omega ,\alpha _{2})  \notag \\
&=&\omega x/2+\frac{\alpha _{2}+5/2+a_{1}+3b_{1}}{x}+\omega x\frac{d}{dz}%
\left( \log \left( \frac{\mathcal{L}^{\left( 1\mid a_{1}\right) _{2}\otimes
\left( 1\mid b_{1}+1\right) _{2}}(z;\alpha _{2})}{\mathcal{L}^{\left( 1\mid
a_{1}\right) _{2}\otimes \left( \left( 1\mid b_{1}\right) _{2}\oplus
2\right) }(z;\alpha _{2})}\right) \right)  \notag \\
&=&\omega x/2+\frac{\alpha +9/2+a_{1}+3b_{1}}{x}+\omega x\frac{d}{dz}\left(
\log \frac{\mathcal{L}^{\left( 1\mid a_{1}\right) _{2}\otimes \left( 1\mid
b_{1}+1\right) _{2}}(z;\alpha _{2})}{z^{4b_{1}+2}\mathcal{L}^{\left( 1\mid
a_{1}\right) _{2}\otimes \left( 1\mid b_{1}\right) _{2}}(z;\alpha )}\right) 
\notag \\
&=&\omega x/2+\frac{\alpha +1/2+a_{1}-5b_{1}}{x}+\omega x\frac{d}{dz}\left(
\log \frac{\mathcal{L}^{\left( 1\mid a_{1}\right) _{2}\otimes \left( 1\mid
b_{1}+1\right) _{2}}(z;\alpha _{2})}{\mathcal{L}^{\left( 1\mid a_{1}\right)
_{2}\otimes \left( 1\mid b_{1}\right) _{2}}(z;\alpha )}\right) ,
\end{eqnarray}%
(the flip in $\varnothing \otimes 0$ is negative, $m=a_{1},\ r=b_{1}+1$),
where we have used

\begin{equation}
\left\{ 
\begin{array}{c}
\left( 1+2a_{1},\left( 1\mid a_{1}\right) _{2}\right) =\left( 1\mid
a_{1}+1\right) _{2} \\ 
\left( 0,\left( 1\mid a_{1}+1\right) _{2}\right) =\left( 0,1)\cup \left(
3\mid a_{1}+1\right) _{2}\right) =\left( 1\mid a_{1}\right) _{2}\oplus 2.%
\end{array}%
\right.
\end{equation}

Taking $\omega =2$ ($t=z$), the corresponding solution of PV Eq(\ref{PV})
with parameters (see Eq(\ref{paramPV}))%
\begin{equation}
\left\{ 
\begin{array}{c}
a=\left( 1+2a_{1}\right) ^{2}/8 \\ 
b=-\left( 1+2b_{1}\right) ^{2}/8 \\ 
c=8(\alpha +1+a_{1}-b_{1}) \\ 
d=-2,%
\end{array}%
\right.
\end{equation}%
is then (see Eq(\ref{defPV}))

\begin{equation}
y(t)=1-2/\left( 1-\frac{\alpha +1+a_{1}-b_{1}}{t}+\frac{d}{dt}\left( \log
\left( \frac{\mathcal{L}^{\left( 1\mid a_{1}\right) _{2}\otimes \left( 1\mid
b_{1}\right) _{2}}(z;\alpha )}{\mathcal{L}^{\left( 1\mid a_{1}\right)
_{2}\otimes \left( 1\mid b_{1}\right) _{2}}(z;\alpha _{2})}\right) \right)
\right) .
\end{equation}

\subsubsection{Rational solutions of the Painlev\'{e} chain of period 6 (A$_{%
\text{5}}$-PV)}

For $p=6$, we can have:

* $p_{1}=5$ and $p_{2}=1$ with $k=1$ ($j_{1}=2,j_{2}=0$)

* $p_{1}=4$ and $p_{2}=2$ with $k=2$ ($j_{1}=1,j_{2}=0$)

* $p_{1}=p_{2}=k=3$ ($j_{1}=j_{2}=0$) .

\paragraph{$\left( p_{1},p_{2}\right) =\left( 5,1\right) $}

\textit{Theorem 4} gives then $N_{m}\otimes L_{r}=\left( \left( \lambda
_{1}\mid \mu _{1}\right) _{1},\left( \lambda _{2}\mid \mu _{2}\right)
_{1}\right) \otimes \varnothing $ and the corresponding $6$-cyclic chain is
given by $\left( \lambda _{1},\lambda _{1}+\mu _{1},\lambda _{2},\lambda
_{2}+\mu _{2},0\right) \otimes \left( 0\right) $. The solutions of the
dressing chain system of period $6$ with parameters

\begin{equation}
\left\{ 
\begin{array}{c}
\varepsilon _{12}=-2\mu _{1}\omega \\ 
\varepsilon _{23}=2\left( \lambda _{1}+\mu _{1}-\lambda _{2}\right) \omega
\\ 
\varepsilon _{34}=-2\mu _{2}\omega \\ 
\varepsilon _{45}=2\left( \lambda _{2}+\mu _{2}\right) \omega \\ 
\varepsilon _{56}=2\alpha \omega \\ 
\varepsilon _{61}-\Delta =2\left( -1-\lambda _{1}-\alpha \right) \omega ,%
\end{array}%
\right.
\end{equation}%
are (see Eq(\ref{th43}-\ref{th46}))%
\begin{eqnarray}
w_{1}(x) &=&w_{\lambda _{1}\otimes \varnothing }^{\left( \left( \lambda
_{1}\mid \mu _{1}\right) _{1},\left( \lambda _{2}\mid \mu _{2}\right)
_{1}\right) \otimes \varnothing }(x;\omega ,\alpha )\text{ (the flip in }%
\lambda _{1}\otimes \varnothing \text{ is positive, }m=\mu _{1}+\mu _{2},\
r=0\text{)}  \notag \\
&=&-\omega x/2+\frac{\alpha +\mu _{1}+\mu _{2}-1/2}{x}+\omega x\frac{d}{dz}%
\left( \log \left( \frac{\mathcal{L}^{\left( \left( \lambda _{1}\mid \mu
_{1}\right) _{1},\left( \lambda _{2}\mid \mu _{2}\right) _{1}\right) \otimes
\varnothing }(z;\alpha )}{\mathcal{L}^{\left( \left( \lambda _{1}+1\mid \mu
_{1}-1\right) _{1},\left( \lambda _{2}\mid \mu _{2}\right) _{1}\right)
\otimes \varnothing }(z;\alpha )}\right) \right) ,
\end{eqnarray}%
\begin{eqnarray}
w_{2}(x) &=&w_{\lambda _{1}+\mu _{1}\otimes \varnothing }^{\left( \left(
\lambda +1\mid \mu -1\right) _{1}\right) \otimes \varnothing }(x;\omega
,\alpha )  \notag \\
&=&\omega x/2-\frac{\alpha +\mu _{1}+\mu _{2}-1/2}{x}+\omega x\frac{d}{dz}%
\left( \log \left( \frac{\mathcal{L}^{\left( \left( \lambda _{1}+1\mid \mu
_{1}-1\right) _{1},\left( \lambda _{2}\mid \mu _{2}\right) _{1}\right)
\otimes \varnothing }(z;\alpha )}{\mathcal{L}^{\left( \left( \lambda
_{1}+1\mid \mu _{1}\right) _{1},\left( \lambda _{2}\mid \mu _{2}\right)
_{1}\right) \otimes \varnothing }(z;\alpha )}\right) \right)
\end{eqnarray}%
(the flip in $\left( \lambda _{1}+\mu _{1}\right) \otimes \varnothing $ is
negative, $m=\mu _{1}+\mu _{2}-1,\ r=0$),%
\begin{eqnarray}
w_{3}(x) &=&w_{\lambda _{2}\otimes \varnothing }^{\left( \left( \lambda
_{1}+1\mid \mu _{1}\right) _{1},\left( \lambda _{2}\mid \mu _{2}\right)
_{1}\right) \otimes \varnothing }(x;\omega ,\alpha )  \notag \\
&=&-\omega x/2+\frac{\alpha +\mu _{1}+\mu _{2}-1/2}{x}+\omega x\frac{d}{dz}%
\left( \log \left( \frac{\mathcal{L}^{\left( \left( \lambda _{1}+1\mid \mu
_{1}\right) _{1},\left( \lambda _{2}\mid \mu _{2}\right) _{1}\right) \otimes
\varnothing }(z;\alpha )}{\mathcal{L}^{\left( \left( \lambda _{1}+1\mid \mu
_{1}\right) _{1},\left( \lambda _{2}+1\mid \mu _{2}-1\right) _{1}\right)
\otimes \varnothing }(z;\alpha )}\right) \right)
\end{eqnarray}%
(the flip in $\lambda _{2}\otimes \varnothing $ is positive, $m=\mu _{1}+\mu
_{2},\ r=0$),%
\begin{eqnarray}
w_{4}(x) &=&w_{\left( \lambda _{2}+\mu _{2}\right) \otimes \varnothing
}^{\left( \left( \lambda _{1}+1\mid \mu _{1}\right) _{1},\left( \lambda
_{2}+1\mid \mu _{2}-1\right) _{1}\right) \otimes \varnothing }(x;\omega
,\alpha )  \notag \\
&=&\omega x/2-\frac{\alpha +\mu _{1}+\mu _{2}-1/2}{x}+\omega x\frac{d}{dz}%
\left( \log \left( \frac{\mathcal{L}^{\left( \left( \lambda _{1}+1\mid \mu
_{1}\right) _{1},\left( \lambda _{2}+1\mid \mu _{2}-1\right) _{1}\right)
\otimes \varnothing }(z;\alpha )}{\mathcal{L}^{\left( \left( \lambda
_{1}+1\mid \mu _{1}\right) _{1},\left( \lambda _{2}+1\mid \mu _{2}\right)
_{1}\right) \otimes \varnothing }(z;\alpha )}\right) \right)
\end{eqnarray}%
(the flip in $\left( \lambda _{2}+\mu _{2}\right) \otimes \varnothing $ is
negative, $m=\mu _{1}+\mu _{2}-1,\ r=0$),%
\begin{eqnarray}
w_{5}(x) &=&w_{0\otimes \varnothing }^{\left( \left( \lambda _{1}+1\mid \mu
_{1}\right) _{1},\left( \lambda _{2}+1\mid \mu _{2}\right) _{1}\right)
\otimes \varnothing }(x;\omega ,\alpha )  \notag \\
&=&\omega x/2-\frac{\alpha +\mu _{1}+\mu _{2}-1/2}{x}+\omega x\frac{d}{dz}%
\left( \log \left( \frac{\mathcal{L}^{\left( \left( \lambda _{1}+1\mid \mu
_{1}\right) _{1},\left( \lambda _{2}+1\mid \mu _{2}\right) _{1}\right)
\otimes \varnothing }(z;\alpha )}{\mathcal{L}^{\left( 0,\left( \lambda
_{1}+1\mid \mu _{1}\right) _{1},\left( \lambda _{2}+1\mid \mu _{2}\right)
_{1}\right) \otimes \varnothing }(z;\alpha )}\right) \right)  \notag \\
&=&\omega x/2-\frac{\alpha +\mu _{1}+\mu _{2}-1/2}{x}+\omega x\frac{d}{dz}%
\left( \log \left( \frac{\mathcal{L}^{\left( \left( \lambda _{1}+1\mid \mu
_{1}\right) _{1},\left( \lambda _{2}+1\mid \mu _{2}\right) _{1}\right)
\otimes \varnothing }(z;\alpha )}{\mathcal{L}^{\left( \left( \lambda
_{1}\mid \mu _{1}\right) _{1},\left( \lambda _{2}\mid \mu _{2}\right)
_{1}\right) \otimes \varnothing }(z;\alpha _{1})}\right) \right)
\end{eqnarray}%
(the flip in $0\otimes \varnothing $ is negative, $m=\mu _{1}+\mu _{2},\ r=0$%
) and

\begin{eqnarray}
w_{6}(x) &=&w_{\varnothing \otimes 0}^{\left( 0,\left( \lambda _{1}+1\mid
\mu _{1}\right) _{1},\left( \lambda _{2}+1\mid \mu _{2}\right) _{1}\right)
\otimes \varnothing }(x;\omega ,\alpha )  \notag \\
&=&\omega x/2+\frac{\alpha +\mu _{1}+\mu _{2}-1/2}{x}+\omega x\frac{d}{dz}%
\left( \log \left( \frac{\mathcal{L}^{\left( 0,\left( \lambda _{1}+1\mid \mu
_{1}\right) _{1},\left( \lambda _{2}+1\mid \mu _{2}\right) _{1}\right)
\otimes \varnothing }(z;\alpha )}{\mathcal{L}^{\left( 0,\left( \lambda
_{1}+1\mid \mu _{1}\right) _{1},\left( \lambda _{2}+1\mid \mu _{2}\right)
_{1}\right) \otimes \left( 0\right) }(z;\alpha )}\right) \right)  \notag \\
&=&\omega x/2+\frac{\alpha +\mu _{1}+\mu _{2}-1/2}{x}+\omega x\frac{d}{dz}%
\left( \log \left( \frac{\mathcal{L}^{\left( \left( \lambda _{1}\mid \mu
_{1}\right) _{1},\left( \lambda _{2}\mid \mu _{2}\right) _{1}\right) \otimes
\varnothing }(z;\alpha _{1})}{\mathcal{L}^{\left( \left( \lambda _{1}\mid
\mu _{1}\right) _{1},\left( \lambda _{2}\mid \mu _{2}\right) _{1}\right)
\otimes \varnothing }(z;\alpha )}\right) \right)
\end{eqnarray}%
(the flip in $\varnothing \otimes 0$ is negative, $m=\mu _{1}+\mu _{2},\ r=0$%
).

\paragraph{$\left( p_{1},p_{2}\right) =\left( 4,2\right) $}

\textit{Theorem 4} gives $N_{m}\otimes L_{r}=\left( \left( 1\mid
a_{1}\right) _{2},\left( \lambda _{1}\mid \mu _{1}\right) _{2}\right)
\otimes \left( 1\mid b_{1}\right) _{2}$ and the corresponding $6$-cyclic
chain is $\left( 1+2a_{1},\lambda _{1},\lambda _{1}+\mu _{1},0\right)
\otimes \left( 1+2b_{1},0\right) $. The solutions of the dressing chain
system of period $6$ with parameters

\begin{equation}
\left\{ 
\begin{array}{c}
\varepsilon _{12}=2\left( 1+2a_{1}-\lambda _{1}\right) \omega \\ 
\varepsilon _{23}=-2\mu _{1}\omega \\ 
\varepsilon _{34}=2\left( \lambda _{1}+\mu _{1}\right) \omega \\ 
\varepsilon _{45}=2\left( \alpha -1-2b_{1}\right) \omega \\ 
\varepsilon _{56}=2\left( 1+2b_{1}\right) \omega \\ 
\varepsilon _{61}-\Delta =2\left( -3-2a_{1}-\alpha \right) \omega ,%
\end{array}%
\right.
\end{equation}%
are (see Eq(\ref{th43}-\ref{th46})) and Eq(\ref{EqextIO5}))

\begin{eqnarray}
w_{1}(x) &=&w_{\left( 1+2a_{1}\right) \otimes \varnothing }^{\left( \left(
1\mid a_{1}\right) _{2},\left( \lambda _{1}\mid \mu _{1}\right) _{2}\right)
\otimes \left( 1\mid b_{1}\right) _{2}}(x;\omega ,\alpha )\text{ (the flip
in }\left( 1+2a_{1}\right) \otimes \varnothing \text{ is negative, }%
m=a_{1}+\mu _{1},\ r=b_{1}\text{)}  \notag \\
&=&\omega x/2-\frac{\alpha +1/2+a_{1}+\mu _{1}-b_{1}}{x}+\omega x\frac{d}{dz}%
\left( \log \left( \frac{\mathcal{L}^{\left( \left( 1\mid a_{1}\right)
_{2},\left( \lambda _{1}\mid \mu _{1}\right) _{2}\right) \otimes \left(
1\mid b_{1}\right) _{2}}(z;\alpha )}{\mathcal{L}^{\left( \left( 1\mid
a_{1}+1\right) _{2},\left( \lambda _{1}\mid \mu _{1}\right) _{2}\right)
\otimes \left( 1\mid b_{1}\right) _{2}}(z;\alpha )}\right) \right) ,
\end{eqnarray}

\begin{eqnarray}
w_{2}(x) &=&w_{\lambda _{1}\otimes \varnothing }^{\left( \left( 1\mid
a_{1}+1\right) _{2},\left( \lambda _{1}\mid \mu _{1}\right) _{2}\right)
\otimes \left( 1\mid b_{1}\right) _{2}}(x;\omega ,\alpha )  \notag \\
&=&-\omega x/2+\frac{\alpha +1/2+a_{1}+\mu _{1}-b_{1}}{x}+\omega x\frac{d}{dz%
}\left( \log \left( \frac{\mathcal{L}^{\left( \left( 1\mid a_{1}+1\right)
_{2},\left( \lambda _{1}\mid \mu _{1}\right) _{2}\right) \otimes \left(
1\mid b_{1}\right) _{2}}(z;\alpha )}{\mathcal{L}^{\left( \left( 1\mid
a_{1}+1\right) _{2},\left( \lambda _{1}+2\mid \mu _{1}-1\right) _{2}\right)
\otimes \left( 1\mid b_{1}\right) _{2}}(z;\alpha )}\right) \right)
\end{eqnarray}%
(the flip in $\lambda _{1}\otimes \varnothing $ is positive, $m=a_{1}+1+\mu
_{1},\ r=b_{1}$),

\begin{eqnarray}
w_{3}(x) &=&w_{\left( \lambda _{1}+\mu _{1}\right) \otimes \varnothing
}^{\left( \left( 1\mid a_{1}+1\right) _{2},\left( \lambda _{1}+1\mid \mu
_{1}-1\right) _{2}\right) \otimes \left( 1\mid b_{1}\right) _{2}}(x;\omega
,\alpha )  \notag \\
&=&\omega x/2-\frac{\alpha +1/2+a_{1}+\mu _{1}-b_{1}}{x}+\omega x\frac{d}{dz}%
\left( \log \left( \frac{\mathcal{L}^{\left( \left( 1\mid a_{1}+1\right)
_{2},\left( \lambda _{1}+2\mid \mu _{1}-1\right) _{2}\right) \otimes \left(
1\mid b_{1}\right) _{2}}(z;\alpha )}{\mathcal{L}^{\left( \left( 1\mid
a_{1}+1\right) _{2},\left( \lambda _{1}+2\mid \mu _{1}\right) _{2}\right)
\otimes \left( 1\mid b_{1}\right) _{2}}(z;\alpha )}\right) \right)
\end{eqnarray}%
(the flip in $\left( \lambda _{1}+\mu _{1}\right) \otimes \varnothing $ is
negative, $m=a_{1}+\mu _{1},\ r=b_{1}$),%
\begin{eqnarray}
w_{4}(x) &=&w_{0\otimes \varnothing }^{\left( \left( 1\mid a_{1}+1\right)
_{2},\left( \lambda _{1}+1\mid \mu _{1}\right) _{2}\right) \otimes \left(
1\mid b_{1}\right) _{2}}(x;\omega ,\alpha )  \notag \\
&=&\omega x/2-\frac{\alpha +3/2+a_{1}+\mu _{1}+3b_{1}}{x}+\omega x\frac{d}{dz%
}\left( \log \left( \frac{\mathcal{L}^{\left( \left( 1\mid a_{1}+1\right)
_{2},\left( \lambda _{1}+2\mid \mu _{1}\right) _{2}\right) \otimes \left(
1\mid b_{1}\right) _{2}}(z;\alpha )}{\mathcal{L}^{\left( 0,\left( 1\mid
a_{1}+1\right) _{2},\left( \lambda _{1}+2\mid \mu _{1}\right) _{2}\right)
\otimes \left( 1\mid b_{1}\right) _{2}}(z;\alpha )}\right) \right)  \notag \\
&=&\omega x/2-\frac{\alpha +3/2+a_{1}+\mu _{1}+3b_{1}}{x}+\omega x\frac{d}{dz%
}\left( \log \left( \frac{\mathcal{L}^{\left( \left( 1\mid a_{1}+1\right)
_{2},\left( \lambda _{1}+2\mid \mu _{1}\right) _{2}\right) \otimes \left(
1\mid b_{1}\right) _{2}}(z;\alpha )}{\mathcal{L}^{\left( \left( 1\mid
a_{1}\right) _{2},\left( \lambda _{1}\mid \mu _{1}\right) _{2}\right)
\otimes \left( 1\mid b_{1}\right) _{2}}(z;\alpha _{2})}\right) \right) ,
\end{eqnarray}%
(the flip in $0\otimes \varnothing $ is negative, $m=\mu _{1}+a_{1}+1,\
r=b_{1}$, and we have $\left( 0,\left( 1\mid a_{1}+1\right) _{2},\left(
\lambda _{1}+2\mid \mu _{1}\right) _{2}\right) =\left( 0,1)\right.$ $\left.\cup \left( 3\mid
a_{1}+1\right) _{2},\left( \lambda _{1}+2\mid \mu _{1}\right) _{2}\right)
=\left( \left( 1\mid a_{1}\right) _{2},\left( \lambda _{1}\mid \mu
_{1}\right) _{2}\right) \oplus 2$),

\begin{eqnarray}
w_{5}(x) &=&w_{\varnothing \otimes \left( 1+2b_{1}\right) }^{\left( 0,\left(
1\mid a_{1}+1\right) _{2},\left( \lambda _{1}+2\mid \mu _{1}\right)
_{2}\right) \otimes \left( 1\mid b_{1}\right) _{2}}(x;\omega ,\alpha )\text{
(the flip in }\varnothing \otimes \left( 1+2b_{1}\right) \text{ is negative, 
}m=\mu _{1}+a_{1}+2,\ r=b_{1}\text{)}  \notag \\
&=&\omega x/2+\frac{\alpha +3/2+\mu _{1}+a_{1}+3b_{1}}{x}+\omega x\frac{d}{dz%
}\left( \log \left( \frac{\mathcal{L}^{\left( 0,\left( 1\mid a_{1}+1\right)
_{2},\left( \lambda _{1}+2\mid \mu _{1}\right) _{2}\right) \otimes \left(
1\mid b_{1}\right) _{2}}(z;\alpha )}{\mathcal{L}^{\left( 0,\left( 1\mid
a_{1}+1\right) _{2},\left( \lambda _{1}+2\mid \mu _{1}\right) _{2}\right)
\otimes \left( 1\mid b_{1}+1\right) _{2}}(z;\alpha )}\right) \right)  \notag
\\
&=&\omega x/2+\frac{\alpha +3/2+\mu _{1}+a_{1}+3b_{1}}{x}+\omega x\frac{d}{dz%
}\left( \log \left( \frac{\mathcal{L}^{\left( \left( 1\mid a_{1}\right)
_{2},\left( \lambda _{1}\mid \mu _{1}\right) _{2}\right) \otimes \left(
1\mid b_{1}\right) _{2}}(z;\alpha _{2})}{\mathcal{L}^{\left( \left( 1\mid
a_{1}\right) _{2},\left( \lambda _{1}\mid \mu _{1}\right) _{2}\right)
\otimes \left( 1\mid b_{1}+1\right) _{2}}(z;\alpha _{2})}\right) \right) ,
\end{eqnarray}%
and ($\left( 0,\left( 1\mid b_{1}+1\right) _{2}\right) =(0,1)\cup \left(
1\mid b_{1}+1\right) _{2}=\left( 1\mid b_{1}\right) _{2}\oplus 2$, see also
Eq(\ref{EqextIO5}))

\begin{eqnarray}
w_{6}(x) &=&w_{\varnothing \otimes 0}^{\left( \left( 1\mid a_{1}\right)
_{2},\left( \lambda _{1}\mid \mu _{1}\right) _{2}\right) \otimes \left(
1\mid b_{1}+1\right) _{2}}(x;\omega ,\alpha _{2})  \notag \\
&=&\omega x/2+\frac{\alpha +9/2+\mu _{1}+a_{1}+3b_{1}}{x}+\omega x\frac{d}{dz%
}\left( \log \left( \frac{\mathcal{L}^{\left( \left( 1\mid a_{1}\right)
_{2},\left( \lambda _{1}\mid \mu _{1}\right) _{2}\right) \otimes \left(
1\mid b_{1}+1\right) _{2}}(z;\alpha _{2})}{\mathcal{L}^{\left( \left( 1\mid
a_{1}\right) _{2},\left( \lambda _{1}\mid \mu _{1}\right) _{2}\right)
\otimes \left( \left( 1\mid b_{1}\right) _{2}\oplus 2\right) }(z;\alpha _{2})%
}\right) \right)  \notag \\
&=&\omega x/2+\frac{\alpha +9/2+\mu _{1}+a_{1}+3b_{1}}{x}+\omega x\frac{d}{dz%
}\left( \log \left( \frac{\mathcal{L}^{\left( \left( 1\mid a_{1}\right)
_{2},\left( \lambda _{1}\mid \mu _{1}\right) _{2}\right) \otimes \left(
1\mid b_{1}+1\right) _{2}}(z;\alpha _{2})}{z^{4b_{1}+2}\mathcal{L}^{\left(
\left( 1\mid a_{1}\right) _{2},\left( \lambda _{1}\mid \mu _{1}\right)
_{2}\right) \otimes \left( 1\mid b_{1}\right) _{2}}(z;\alpha )}\right)
\right)  \notag \\
&=&\omega x/2+\frac{\alpha +1/2+\mu _{1}+a_{1}-5b_{1}}{x}+\omega x\frac{d}{dz%
}\left( \log \left( \frac{\mathcal{L}^{\left( \left( 1\mid a_{1}\right)
_{2},\left( \lambda _{1}\mid \mu _{1}\right) _{2}\right) \otimes \left(
1\mid b_{1}+1\right) _{2}}(z;\alpha _{2})}{\mathcal{L}^{\left( \left( 1\mid
a_{1}\right) _{2},\left( \lambda _{1}\mid \mu _{1}\right) _{2}\right)
\otimes \left( 1\mid b_{1}\right) _{2}}(z;\alpha )}\right) \right)
\end{eqnarray}%
(the flip in $\varnothing \otimes 0$ is negative, $m=\mu _{1}+a_{1},\
r=b_{1}+1$).

\paragraph{$\left( p_{1},p_{2}\right) =\left( 3,3\right) $ and $k=3$}

\textit{Theorem 4} gives $N_{m}\otimes L_{r}=\left( \left( 1\mid
a_{1}\right) _{3},\left( 2\mid a_{2}\right) _{3}\right) \otimes \left(
\left( 1\mid b_{1}\right) _{2},\left( 2\mid b_{2}\right) _{2}\right) $ and
the corresponding $6$-cyclic chain is $\left( 1+3a_{1},2+3a_{2},0\right)
\otimes \left( 1+3b_{1},2+3b_{2},0\right) $. The solutions of the dressing
chain system of period $6$ with parameters

\begin{equation}
\left\{ 
\begin{array}{c}
\varepsilon _{12}=2\left( -1+3\left( a_{1}-a_{2}\right) \right) \omega \\ 
\varepsilon _{23}=2\left( 2+3a_{2}\right) \omega \\ 
\varepsilon _{34}=2\left( \alpha -1-3b_{1}\right) \omega \\ 
\varepsilon _{23}=2\left( -1+3\left( b_{1}-b_{2}\right) \right) \omega \\ 
\varepsilon _{34}=2\left( 2+3b_{2}\right) \omega \\ 
\varepsilon _{41}-\Delta =2\left( -4-3a_{1}-\alpha \right) \omega ,%
\end{array}%
\right.
\end{equation}%
are (see Eq(\ref{th43}-\ref{th46})) and Eq(\ref{EqextIO5}))

\begin{eqnarray}
w_{1}(x) &=&w_{\left( 1+3a_{1}\right) \otimes \varnothing }^{\left( \left(
1\mid a_{1}\right) _{3},\left( 2\mid a_{2}\right) _{3}\right) \otimes \left(
\left( 1\mid b_{1}\right) _{3},\left( 2\mid b_{2}\right) _{3}\right)
}(x;\omega ,\alpha )\text{ }  \\
&=&\omega x/2-\frac{\alpha +1/2+a_{1}+a_{2}-b_{1}-b_{2}}{x}+\omega x\frac{d}{%
dz}\left( \log \left( \frac{\mathcal{L}^{\left( \left( 1\mid a_{1}\right)
_{3},\left( 2\mid a_{2}\right) _{3}\right) \otimes \left( \left( 1\mid
b_{1}\right) _{3},\left( 2\mid b_{2}\right) _{3}\right) }(z;\alpha )}{%
\mathcal{L}^{\left( \left( 1\mid a_{1}+1\right) _{3},\left( 2\mid
a_{2}\right) _{3}\right) \otimes \left( \left( 1\mid b_{1}\right)
_{3},\left( 2\mid b_{2}\right) _{3}\right) }(z;\alpha )}\right) \right)\notag 
\end{eqnarray}%
(the flip in $\left( 1+3a_{1}\right) \otimes \varnothing $ is negative, $%
m=a_{1}+a_{2},\ r=b_{1}+b_{2}$),%
\begin{eqnarray}
w_{2}(x) &=&w_{\left( 2+3a_{2}\right) \otimes \varnothing }^{\left( \left(
1\mid a_{1}+1\right) _{3},\left( 2\mid a_{2}\right) _{3}\right) \otimes
\left( \left( 1\mid b_{1}\right) _{3},\left( 2\mid b_{2}\right) _{3}\right)
}(x;\omega ,\alpha )   \\
&=&\omega x/2-\frac{\alpha +3/2+a_{1}+a_{2}-b_{1}-b_{2}}{x}+\omega x\frac{d}{%
dz}\left( \log \left( \frac{\mathcal{L}^{\left( \left( 1\mid a_{1}+1\right)
_{3},\left( 2\mid a_{2}\right) _{3}\right) \otimes \left( \left( 1\mid
b_{1}\right) _{3},\left( 2\mid b_{2}\right) _{3}\right) }(z;\alpha )}{%
\mathcal{L}^{\left( \left( 1\mid a_{1}+1\right) _{3},\left( 2\mid
a_{2}+1\right) _{3}\right) \otimes \left( \left( 1\mid b_{1}\right)
_{3},\left( 2\mid b_{2}\right) _{3}\right) }(z;\alpha )}\right) \right)\notag
\end{eqnarray}%
(the flip in $\left( 2+3a_{2}\right) \otimes \varnothing $ is negative, $%
m=a_{1}+a_{2}+1,\ r=b_{1}+b_{2}$),

\begin{eqnarray}
w_{3}(x) &=&w_{0\otimes \varnothing }^{\left( \left( 1\mid a_{1}+1\right)
_{3},\left( 2\mid a_{2}+1\right) _{3}\right) \otimes \left( \left( 1\mid
b_{1}\right) _{3},\left( 2\mid b_{2}\right) _{3}\right) }(x;\omega ,\alpha )%
\text{ }  \\
&=&\omega x/2-\frac{\alpha +5/2+a_{1}+a_{2}-b_{1}-b_{2}}{x}+\omega x\frac{d}{%
dz}\left( \log \left( \frac{\mathcal{L}^{\left( \left( 1\mid a_{1}+1\right)
_{3},\left( 2\mid a_{2}+1\right) _{3}\right) \otimes \left( \left( 1\mid
b_{1}\right) _{3},\left( 2\mid b_{2}\right) _{3}\right) }(z;\alpha )}{%
\mathcal{L}^{\left( \left( \left( 1\mid a_{1}\right) _{3},\left( 2\mid
a_{2}\right) _{3}\right) \oplus 3\right) \otimes \left( \left( 1\mid
b_{1}\right) _{3},\left( 2\mid b_{2}\right) _{3}\right) }(z;\alpha )}\right)
\right)  \notag \\
&=&\omega x/2-\frac{\alpha +5/2+a_{1}+a_{2}-b_{1}-b_{2}}{x}+\omega x\frac{d}{%
dz}\left( \log \left( \frac{\mathcal{L}^{\left( \left( 1\mid a_{1}+1\right)
_{3},\left( 2\mid a_{2}+1\right) _{3}\right) \otimes \left( \left( 1\mid
b_{1}\right) _{3},\left( 2\mid b_{2}\right) _{3}\right) }(z;\alpha )}{%
\mathcal{L}^{\left( \left( 1\mid a_{1}\right) _{3},\left( 2\mid a_{2}\right)
_{3}\right) \otimes \left( \left( 1\mid b_{1}\right) _{3},\left( 2\mid
b_{2}\right) _{3}\right) }(z;\alpha _{3})}\right) \right)\notag 
\end{eqnarray}%
(the flip in $0\otimes \varnothing $ is negative, $m=a_{1}+a_{2}+2,\
r=b_{1}+b_{2}$ and we have used $\left( 0,\left( 1\mid a_{1}+1\right)
_{3},\left( 2\mid a_{2}+1\right) _{3}\right) =\left( 0,1,2\right) \cup
\left( \left( 4\mid a_{1}\right) _{3},\left( 5\mid a_{2}\right) _{3}\right)
=\left( \left( 1\mid a_{1}\right) _{3},\left( 2\mid a_{2}\right) _{3}\right)
\oplus 3$),%
\begin{eqnarray}
w_{4}(x) &=&w_{\varnothing \otimes \left( 1+2b_{1}\right) }^{\left( \left(
\left( 1\mid a_{1}\right) _{3},\left( 2\mid a_{2}\right) _{3}\right) \oplus
3\right) \otimes \left( \left( 1\mid b_{1}\right) _{3},\left( 2\mid
b_{2}\right) _{3}\right) }(x;\omega ,\alpha )  \\
&=&\omega x/2+\frac{\alpha +5/2+a_{1}+a_{2}+3b_{1}+3b_{2}}{x}+\omega x\frac{d%
}{dz}\left( \log \left( \frac{\mathcal{L}^{\left( \left( \left( 1\mid
a_{1}\right) _{3},\left( 2\mid a_{2}\right) _{3}\right) \oplus 3\right)
\otimes \left( \left( 1\mid b_{1}\right) _{3},\left( 2\mid b_{2}\right)
_{3}\right) }(z;\alpha )}{\mathcal{L}^{\left( \left( \left( 1\mid
a_{1}\right) _{3},\left( 2\mid a_{2}\right) _{3}\right) \oplus 3\right)
\otimes \left( \left( 1\mid b_{1}+1\right) _{3},\left( 2\mid b_{2}\right)
_{3}\right) }(z;\alpha )}\right) \right)  \notag \\
&=&\omega x/2+\frac{\alpha +5/2+a_{1}+a_{2}+3b_{1}+3b_{2}}{x}+\omega x\frac{d%
}{dz}\left( \log \left( \frac{\mathcal{L}^{\left( \left( 1\mid a_{1}\right)
_{3},\left( 2\mid a_{2}\right) _{3}\right) \otimes \left( \left( 1\mid
b_{1}+1\right) _{3},\left( 2\mid b_{2}\right) _{3}\right) }(z;\alpha _{3})}{%
\mathcal{L}^{\left( \left( 1\mid a_{1}\right) _{3},\left( 2\mid a_{2}\right)
_{3}\right) \otimes \left( \left( 1\mid b_{1}+1\right) _{3},\left( 2\mid
b_{2}\right) _{3}\right) }(z;\alpha _{3})}\right) \right)\notag 
\end{eqnarray}%
(the flip in $\varnothing \otimes \left( 1+2b_{1}\right) $ is negative, $%
m=a_{1}+a_{2}+3,\ r=b_{1}+b_{2}$),

\begin{eqnarray}
w_{5}(x) &=&w_{\varnothing \otimes \left( 2+3b_{1}\right) }^{\left( \left(
1\mid a_{1}\right) _{3},\left( 2\mid a_{2}\right) _{3}\right) \otimes \left(
\left( 1\mid b_{1}+1\right) _{3},\left( 2\mid b_{2}\right) _{3}\right)
}(x;\omega ,\alpha _{3})   \\
&=&\omega x/2+\frac{\alpha +11/2+a_{1}+a_{2}+3b_{1}+3b_{2}}{x}+\omega x\frac{%
d}{dz}\left( \log \left( \frac{\mathcal{L}^{\left( \left( 1\mid a_{1}\right)
_{3},\left( 2\mid a_{2}\right) _{3}\right) \otimes \left( \left( 1\mid
b_{1}+1\right) _{3},\left( 2\mid b_{2}\right) _{3}\right) }(z;\alpha _{3})}{%
\mathcal{L}^{\left( \left( 1\mid a_{1}\right) _{3},\left( 2\mid a_{2}\right)
_{3}\right) \otimes \left( \left( 1\mid b_{1}+1\right) _{3},\left( 2\mid
b_{2}+1\right) _{3}\right) }(z;\alpha _{3})}\right) \right)\notag
\end{eqnarray}%
(the flip in $\varnothing \otimes \left( 2+3b_{1}\right) $ is negative, $%
m=a_{1}+a_{2},\ r=b_{1}+b_{2}+1$) and

\begin{eqnarray}
w_{6}(x) &=&w_{\varnothing \otimes 0}^{\left( \left( 1\mid a_{1}\right)
_{3},\left( 2\mid a_{2}\right) _{3}\right) \otimes \left( \left( 1\mid
b_{1}+1\right) _{3},\left( 2\mid b_{2}+1\right) _{3}\right) }(x;\omega
,\alpha _{3})   \\
&=&\omega x/2+\frac{\alpha +17/2+a_{1}+a_{2}+3b_{1}+3b_{2}}{x}+\omega x\frac{%
d}{dz}\left( \log \left( \frac{\mathcal{L}^{\left( \left( 1\mid a_{1}\right)
_{3},\left( 2\mid a_{2}\right) _{3}\right) \otimes \left( \left( 1\mid
b_{1}+1\right) _{3},\left( 2\mid b_{2}+1\right) _{3}\right) }(z;\alpha _{3})%
}{\mathcal{L}^{\left( \left( 1\mid a_{1}\right) _{3},\left( 2\mid
a_{2}\right) _{3}\right) \otimes \left( \left( \left( 1\mid b_{1}\right)
_{3},\left( 2\mid b_{2}\right) _{3}\right) \oplus 3\right) }(z;\alpha _{3})}%
\right) \right)  \notag \\
&=&\omega x/2+\frac{\alpha +17/2+a_{1}+a_{2}+3b_{1}+3b_{2}}{x}+\omega x\frac{%
d}{dz}\left( \log \left( \frac{\mathcal{L}^{\left( \left( 1\mid a_{1}\right)
_{3},\left( 2\mid a_{2}\right) _{3}\right) \otimes \left( \left( 1\mid
b_{1}+1\right) _{3},\left( 2\mid b_{2}+1\right) _{3}\right) }(z;\alpha _{3})%
}{z^{6\left( b_{1}+b_{2}\right) +6}\mathcal{L}^{\left( \left( 1\mid
a_{1}\right) _{3},\left( 2\mid a_{2}\right) _{3}\right) \otimes \left(
\left( 1\mid b_{1}\right) _{3},\left( 2\mid b_{2}\right) _{3}\right)
}(z;\alpha )}\right) \right)  \notag \\
&=&\omega x/2+\frac{\alpha -7/2+a_{1}+a_{2}-9b_{1}-9b_{2}}{x}+\omega x\frac{d%
}{dz}\left( \log \left( \frac{\mathcal{L}^{\left( \left( 1\mid a_{1}\right)
_{3},\left( 2\mid a_{2}\right) _{3}\right) \otimes \left( \left( 1\mid
b_{1}+1\right) _{3},\left( 2\mid b_{2}+1\right) _{3}\right) }(z;\alpha _{3})%
}{\mathcal{L}^{\left( \left( 1\mid a_{1}\right) _{3},\left( 2\mid
a_{2}\right) _{3}\right) \otimes \left( \left( 1\mid b_{1}\right)
_{3},\left( 2\mid b_{2}\right) _{3}\right) }(z;\alpha )}\right) \right)\notag
\end{eqnarray}%
(the flip in $\varnothing \otimes 0$ is negative, $m=a_{1}+a_{2},\
r=b_{1}+b_{2}+2$, see also Eq(\ref{EqextIO5})).

\paragraph{$\left( p_{1},p_{2}\right) =\left( 3,3\right) $ and $k=1$}

\textit{Theorem 4} gives $N_{m}\otimes L_{r}=\left( \lambda _{1}\mid \mu
_{1}\right) _{1}\otimes \left( \rho _{1}\mid \sigma _{1}\right) _{1}$ and
the corresponding $6$-cyclic chain is $\left( \lambda _{1},\lambda _{1}+\mu
_{1},0\right) \otimes \left( \rho _{1},\rho _{1}+\sigma _{1},0\right) $. The
solutions of the dressing chain system of period $6$ with parameters

\begin{equation}
\left\{ 
\begin{array}{c}
\varepsilon _{12}=-2\mu _{1}\omega \\ 
\varepsilon _{23}=2\left( \lambda _{1}+\mu _{1}\right) \omega \\ 
\varepsilon _{34}=2\left( \alpha -\rho _{1}\right) \omega \\ 
\varepsilon _{23}=-2\sigma _{1}\omega \\ 
\varepsilon _{34}=2\left( \rho _{1}+\sigma _{1}\right) \omega \\ 
\varepsilon _{41}-\Delta =2\left( -1-\lambda _{1}-\alpha \right) \omega ,%
\end{array}%
\right.
\end{equation}%
are (see Eq(\ref{th43}-\ref{th46})) and Eq(\ref{EqextIO5})):

\begin{eqnarray}
w_{1}(x) &=&w_{\lambda _{1}\otimes \varnothing }^{\left( \lambda _{1}\mid
\mu _{1}\right) _{1}\otimes \left( \rho _{1}\mid \sigma _{1}\right)
_{1}}(x;\omega ,\alpha )\text{ }  \notag \\
&=&-\omega x/2+\frac{\alpha -1/2+\mu _{1}-\sigma _{1}}{x}+\omega x\frac{d}{dz%
}\left( \log \left( \frac{\mathcal{L}^{\left( \lambda _{1}\mid \mu
_{1}\right) _{1}\otimes \left( \rho _{1}\mid \sigma _{1}\right)
_{1}}(z;\alpha )}{\mathcal{L}^{\left( \lambda _{1}+1\mid \mu _{1}-1\right)
_{1}\otimes \left( \rho _{1}\mid \sigma _{1}\right) _{1}}(z;\alpha )}\right)
\right)
\end{eqnarray}%
(the flip in $\lambda _{1}\otimes \varnothing $ is positive, $m=\mu _{1},\
r=\sigma _{1}$),%
\begin{eqnarray}
w_{2}(x) &=&w_{\left( \lambda _{1}+\mu _{1}\right) \otimes \varnothing
}^{\left( \lambda _{1}+1\mid \mu _{1}-1\right) _{1}\otimes \left( \rho
_{1}\mid \sigma _{1}\right) _{1}}(x;\omega ,\alpha )  \notag \\
&=&\omega x/2-\frac{\alpha -1/2+\mu _{1}-\sigma _{1}}{x}+\omega x\frac{d}{dz}%
\left( \log \left( \frac{\mathcal{L}^{\left( \lambda _{1}+1\mid \mu
_{1}-1\right) _{1}\otimes \left( \rho _{1}\mid \sigma _{1}\right)
_{1}}(z;\alpha )}{\mathcal{L}^{\left( \lambda _{1}+1\mid \mu _{1}\right)
_{1}\otimes \left( \rho _{1}\mid \sigma _{1}\right) _{1}}(z;\alpha )}\right)
\right)
\end{eqnarray}%
(the flip in $\left( \lambda _{1}+\mu _{1}\right) \otimes \varnothing $ is
negative, $m=\mu _{1}-1,\ r=\sigma _{1}$),

\begin{eqnarray}
w_{3}(x) &=&w_{0\otimes \varnothing }^{\left( \lambda _{1}+1\mid \mu
_{1}\right) _{1}\otimes \left( \rho _{1}\mid \sigma _{1}\right)
_{1}}(x;\omega ,\alpha )\text{ }  \notag \\
&=&\omega x/2-\frac{\alpha +1/2+\mu _{1}-\sigma _{1}}{x}+\omega x\frac{d}{dz}%
\left( \log \left( \frac{\mathcal{L}^{\left( \lambda _{1}+1\mid \mu
_{1}\right) _{1}\otimes \left( \rho _{1}\mid \sigma _{1}\right)
_{1}}(z;\alpha )}{\mathcal{L}^{\left( \left( \lambda _{1}\mid \mu
_{1}\right) _{1}\oplus 1\right) \otimes \left( \rho _{1}\mid \sigma
_{1}\right) _{1}}(z;\alpha )}\right) \right)  \notag \\
&=&\omega x/2-\frac{\alpha +1/2+\mu _{1}-\sigma _{1}}{x}+\omega x\frac{d}{dz}%
\left( \log \left( \frac{\mathcal{L}^{\left( \lambda _{1}+1\mid \mu
_{1}\right) _{1}\otimes \left( \rho _{1}\mid \sigma _{1}\right)
_{1}}(z;\alpha )}{\mathcal{L}^{\left( \lambda _{1}\mid \mu _{1}\right)
_{1}\otimes \left( \rho _{1}\mid \sigma _{1}\right) _{1}}(z;\alpha _{1})}%
\right) \right)
\end{eqnarray}%
(the flip in $0\otimes \varnothing $ is negative, $m=\mu _{1},\ r=\sigma
_{1} $ and we have used $\left( 0,\left( \lambda _{1}+1\mid \mu _{1}\right)
_{1}\right) =\left( \lambda _{1}\mid \mu _{1}\right) _{1}\oplus 1$),%
\begin{eqnarray}
w_{4}(x) &=&w_{\varnothing \otimes \rho _{1}}^{\left( \left( \lambda
_{1}\mid \mu _{1}\right) _{1}\oplus 1\right) \otimes \left( \rho _{1}\mid
\sigma _{1}\right) _{1}}(x;\omega ,\alpha )  \notag \\
&=&-\omega x/2-\frac{\alpha -5/2+\mu _{1}+3\sigma _{1}}{x}+\omega x\frac{d}{%
dz}\left( \log \left( \frac{\mathcal{L}^{\left( \left( \lambda _{1}\mid \mu
_{1}\right) _{1}\oplus 1\right) \otimes \left( \rho _{1}\mid \sigma
_{1}\right) _{1}}(z;\alpha )}{\mathcal{L}^{\left( \left( \lambda _{1}\mid
\mu _{1}\right) _{1}\oplus 1\right) \otimes \left( \rho _{1}+1\mid \sigma
_{1}-1\right) _{1}}(z;\alpha )}\right) \right)  \notag \\
&=&-\omega x/2-\frac{\alpha -5/2+\mu _{1}+3\sigma _{1}}{x}+\omega x\frac{d}{%
dz}\left( \log \left( \frac{\mathcal{L}^{\left( \lambda _{1}\mid \mu
_{1}\right) _{1}\otimes \left( \rho _{1}\mid \sigma _{1}\right)
_{1}}(z;\alpha _{1})}{\mathcal{L}^{\left( \lambda _{1}\mid \mu _{1}\right)
_{1}\otimes \left( \rho _{1}+1\mid \sigma _{1}-1\right) _{1}}(z;\alpha _{1})}%
\right) \right)
\end{eqnarray}%
(the flip in $\varnothing \otimes \rho _{1}$ is positive, $m=\mu _{1}+1,\
r=\sigma _{1}$),

\begin{eqnarray}
w_{5}(x) &=&w_{\varnothing \otimes \left( \rho _{1}+\sigma _{1}\right)
}^{\left( \lambda _{1}\mid \mu _{1}\right) _{1}\otimes \left( \rho
_{1}+1\mid \sigma _{1}-1\right) _{1}}(x;\omega ,\alpha _{1})  \notag \\
&=&\omega x/2+\frac{\alpha -7/2+\mu _{1}+3\sigma _{1}}{x}+\omega x\frac{d}{dz%
}\left( \log \left( \frac{\mathcal{L}^{\left( \lambda _{1}\mid \mu
_{1}\right) _{1}\otimes \left( \rho _{1}+1\mid \sigma _{1}-1\right)
_{1}}(z;\alpha _{1})}{\mathcal{L}^{\left( \lambda _{1}\mid \mu _{1}\right)
_{1}\otimes \left( \rho _{1}+1\mid \sigma _{1}\right) _{1}}(z;\alpha _{1})}%
\right) \right)
\end{eqnarray}%
(the flip in $\varnothing \otimes \left( \rho _{1}+\sigma _{1}\right) $ is
negative, $m=\mu _{1},\ r=\sigma _{1}-1$) and

\begin{eqnarray}
w_{6}(x) &=&w_{\varnothing \otimes 0}^{\left( \lambda _{1}\mid \mu
_{1}\right) _{1}\otimes \left( \rho _{1}+1\mid \sigma _{1}\right)
_{1}}(x;\omega ,\alpha _{1})  \notag \\
&=&\omega x/2+\frac{\alpha +1/2+\mu _{1}+3\sigma _{1}}{x}+\omega x\frac{d}{dz%
}\left( \log \left( \frac{\mathcal{L}^{\left( \lambda _{1}\mid \mu
_{1}\right) _{1}\otimes \left( \rho _{1}+1\mid \sigma _{1}\right)
_{1}}(z;\alpha _{1})}{\mathcal{L}^{\left( \lambda _{1}\mid \mu _{1}\right)
_{1}\otimes \left( \left( \rho _{1}\mid \sigma _{1}\right) _{1}\oplus
1\right) }(z;\alpha _{1})}\right) \right)  \notag \\
&=&\omega x/2+\frac{\alpha +1/2+\mu _{1}+3\sigma _{1}}{x}+\omega x\frac{d}{dz%
}\left( \log \left( \frac{\mathcal{L}^{\left( \lambda _{1}\mid \mu
_{1}\right) _{1}\otimes \left( \rho _{1}+1\mid \sigma _{1}\right)
_{1}}(z;\alpha _{1})}{z^{2\sigma _{1}}\mathcal{L}^{\left( \lambda _{1}\mid
\mu _{1}\right) _{1}\otimes \left( \rho _{1}\mid \sigma _{1}\right)
_{1}}(z;\alpha )}\right) \right)  \notag \\
&=&\omega x/2+\frac{\alpha +1/2+\mu _{1}-\sigma _{1}}{x}+\omega x\frac{d}{dz}%
\left( \log \left( \frac{\mathcal{L}^{\left( \lambda _{1}\mid \mu
_{1}\right) _{1}\otimes \left( \rho _{1}+1\mid \sigma _{1}\right)
_{1}}(z;\alpha _{1})}{\mathcal{L}^{\left( \lambda _{1}\mid \mu _{1}\right)
_{1}\otimes \left( \rho _{1}\mid \sigma _{1}\right) _{1}}(z;\alpha )}\right)
\right)
\end{eqnarray}%
(the flip in $\varnothing \otimes 0$ is negative, $m=\mu _{1},\ r=\sigma
_{1} $, see also Eq(\ref{EqextIO5})).

\subsection{Rational extensions of the HO as limit cases of rational
extensions of the IO}

Consider a rational extension of the HO associated to the canonical Maya
diagram $N_{m}$. By splitting $N_{m}$ into two subsets of even and odd
integers respectively, we can write

\begin{equation}
N_{m}=\left( 2a_{1}+1,...,2a_{m_{1}}+1\right) \cup \left(
2b_{1},...,2b_{m_{2}}\right) ,\ m=m_{1}+m_{2},
\end{equation}%
where $a_{i},b_{i}\in 
\mathbb{N}
$, and

\begin{eqnarray}
W^{\left( N_{m}\right) }(x;\omega ) &\propto &e^{-mz/2}  \notag \\
&&\times W\left( H_{2a_{1}+1}\left( \sqrt{\frac{\omega }{2}}x\right)
,...,H_{2a_{m_{1}}+1}\left( \sqrt{\frac{\omega }{2}}x\right)
,H_{2b_{1}}\left( \sqrt{\frac{\omega }{2}}x\right) ,...,H_{2b_{m_{2}}}\left( 
\sqrt{\frac{\omega }{2}}x\right) \mid x\right) ,
\end{eqnarray}%
where $z=\sqrt{\omega /2}x$. Using Eq(\ref{wronskprop}) and Eq(\ref%
{correspHL}), we obtain immediately ($z=\omega x^{2}/2$)

\begin{equation}
W^{\left( N_{m}\right) }(x;\omega )\propto e^{-mz/2}W\left(
z^{1/2}L_{a_{1}}^{1/2}(z),...,z^{1/2}L_{a_{m_{1}}}^{1/2}(z),L_{b_{1}}^{-1/2}(z),...,L_{b_{m_{2}}}^{-1/2}(z)\mid x\right) ,
\end{equation}%
ie (see Eq(\ref{spec OI}) and Eq(\ref{shadOI}))

\begin{equation}
W^{\left( N_{m}\right) }(x;\omega )\propto W^{A_{m_{1}}\otimes
B_{m_{2}}}\left( x;\omega ,1/2\right) ,
\end{equation}%
where

\begin{equation}
A_{m_{1}}\otimes B_{m_{2}}=\left( a_{1},...,a_{m_{1}}\right) \otimes \left(
b_{1},...,b_{m_{2}}\right) .
\end{equation}

Since (see Eq(\ref{OH}) and Eq(\ref{OI}))

\begin{equation}
V(x;\omega ,1/2)=V(x;\omega )-\omega ,
\end{equation}%
we deduce

\begin{equation}
V^{\left( N_{m}\right) }(x;\omega )=V^{A_{m_{1}}\otimes B_{m_{2}}}(x;\omega
,1/2)+\omega .  \label{HO-IO}
\end{equation}

It means that all the rational extensions of the HO can be considered as
rational extensions of the IO in the limit case where the $\alpha $
parameter tends to $1/2$. In this limit, the impenetrable barrier term in
the IO potential disappears and the IO potential degenerates into the
harmonic potential which is regular on the whole real line, the spectrum of
the IO potential giving to the odd indexed eigenstates of the HO potential
and the shadow spectrum of the IO the even indexed eigenstates of the HO.
The result above shows that this degeneracy is still valid at the level of
the rational extensions

We can note that if we start from an extension of the IO with $\alpha $ half
integer $\alpha =k+1/2$, associated to the UC $A_{m_{1}}\otimes B_{m_{2}}$,
Eq(\ref{EqextIO2}) allows us to write

\begin{equation}
V^{A_{m_{1}}\otimes B_{m_{2}}}(x;\omega ,k+1/2)=V^{\left( A_{m_{1}}\oplus
k\right) \otimes B_{m_{2}}}(x;\omega ,1/2)+2k\omega ,
\end{equation}%
and, with the correspondence Eq(\ref{HO-IO}) above%
\begin{equation}
V^{A_{m_{1}}\otimes B_{m_{2}}}(x;\omega ,k+1/2)=V^{\left( \widetilde{N}%
_{m}\right) }(x;\omega )+\left( 2k-1\right) \omega ,
\end{equation}%
where

\begin{equation}
\widetilde{N}_{m}=\left( 1,...,2k-1,2a_{1}+2k+1,...,2a_{m_{1}}+2k+1\right)
\cup \left( 2b_{1},...,2b_{m_{2}}\right) .
\end{equation}

Consequently, the rational extensions of the IO with half integer $\alpha $
parameter can be identified to rational extensions of the HO and are
monodromy free. Conversely, the rational extensions of the which solves the
even periodic dressing chains cover the subset of monodromy free solutions.

\section{Conclusion}

We propose a new method for building the rational solutions of the dressing
chains for a Schr\"{o}dinger operator. The results are directly obtained in
closed determinantal form and we give some explicit examples which had never
been previously presented. For the chains with odd periodicity (A$_{\text{2n}%
}$-PIV system) as for those with even periodicity (A$_{\text{2n+1}}$-PV
system), we describe in a new systematic way some sets of rational solutions
of the dressing chain system. In the odd periodicity case we conjecture,
based on the theory of potentials with trivial monodromy \cite%
{grunbaum,oblomkov,veselov,GGM0}, that the construction described in this
work covers all rational solutions to higher order Painlev\'e systems. The
even periodicity case is more involved, but it seems plausible to surmise
that this is also the case: the trivial monodromy property having to be
replaced by a constraint of fixed monodromy at one point for all the
eigenfunctions of the potentials. These conjectures are the object of
further investigations. The question of the optimal pseudo-Wronskian
representation of these solutions \cite{GGM} will also be adressed in a
forthcoming work.

\section{Acknowledgments}

We wish to thank P.\ A.\ Clarkson, G.\ Filipuk, A.\ Hone and M.\ Mazzocco for enlightening
discussions. DGU and SL acknowledge financial support from the Royal Society via the International Exchanges Scheme.

\end{document}